\documentclass[11pt]{article}
\usepackage[utf8]{inputenc}
\usepackage{amsmath}
\usepackage{amssymb}  
\usepackage{mathrsfs}
\usepackage{tabularx}
\usepackage{booktabs}
\usepackage{enumerate}
\usepackage{amsthm}
\usepackage{color}
\usepackage{xcolor}
\usepackage{graphicx}
\usepackage{chngcntr}
\usepackage{enumitem}
\usepackage[sort]{natbib}
\usepackage{url}
\usepackage{ulem}
\usepackage{graphicx}
\usepackage{geometry}
\usepackage{varioref}
\usepackage{etoolbox,fontawesome}
\usepackage[section]{algorithm}
\usepackage{hyperref}
\usepackage{setspace}

\usepackage[capitalise,nameinlink
    ,noabbrev]{cleveref}
  \hypersetup{
    colorlinks=true,
    linkcolor=blue,
    citecolor=blue
}
\AtBeginEnvironment{thebibliography}{\def\url#1{\href{#1}{\faExternalLink}}}

\counterwithin{figure}{section}
\counterwithin{table}{section}
\usepackage{autonum} 
\newtheorem{theorem}{Theorem}[section]
\newtheorem{assumption}{Assumption}[section]
\newtheorem{definition}{Definition}[section]
\newtheorem{proposition}{Proposition}[section]

\newtheorem{remark}{Remark}[section]
\newtheorem{lemma}{Lemma}[section]
\newtheorem{corollary}{Corollary}[section]
\newcommand{\op}{o_{\mathbb{P}}}
\newcommand{\Op}{O_{\mathbb{P}}}

\newcommand{\E}{\mathbb{E}}
\newcommand{\F}{\mathcal{F}} %filtration L
\newcommand{\FF}{\mathcal{F}}
\newcommand{\B}{\mathcal{B}}
\newcommand{\R}{\mathbb{R}}

\newcommand{\lf}{\lfloor}
\newcommand{\rf}{\rfloor}
\newcommand{\proj}{\mathcal{P}}
\newcommand{\I}{\mathcal{I}}
\newcommand{\T}{\top}
\newcommand{\lt}{\left}
\newcommand{\rt}{\right}
\newcommand{\pp}{\mathbb{P}}

\numberwithin{equation}{section}

\definecolor{darkgreen}{rgb}{0.0, 0.5, 0.0}
\definecolor{ashgrey}{rgb}{0.7, 0.75, 0.71}
\newcommand{\WC}{\textcolor{black}}

\newcommand{\LJ}{\textcolor{black}}

\newcommand{\checkit}{\textcolor{black}} 
\usepackage{xr}
\makeatletter
\title{\bf Difference-based covariance matrix estimate in time series nonparametric regression with applications to specification tests}
\date{}
\author{ \small Lujia Bai \\
\small Center for Statistical Science, \\
	\small Department of Industrial Engineering, \\
	\small Tsinghua University\\
\and \small Weichi Wu\\
\small Center for Statistical Science, \\
	\small Department of Industrial Engineering, \\
	\small Tsinghua University\\}
\begin{document}
\maketitle
\begin{abstract}
  Long-run covariance matrix estimation is the building block of time series inference. The corresponding difference-based estimator,
 which avoids detrending,
has attracted considerable interest due to its  robustness to both smooth and abrupt structural breaks and its competitive finite sample performance. However,
existing methods mainly focus on estimators for the univariate process 
while their direct and multivariate extensions for most linear models are asymptotically biased.  We propose a novel difference-based and debiased long-run covariance matrix estimator for  functional linear models with time-varying regression coefficients,  allowing time series non-stationarity, long-range dependence, state-heteroscedasticity and their mixtures. 
We apply the new estimator to (i) the structural stability test, overcoming the notorious non-monotonic power phenomena caused by piecewise smooth alternatives for regression coefficients, and (ii) the nonparametric residual-based tests for long memory,
improving  the performance via the residual-free formula of the proposed estimator.  The effectiveness of the proposed method  is justified theoretically and demonstrated by superior performance in simulation studies, while its usefulness is elaborated via real data analysis. Our method is implemented in the R package \texttt{mlrv}.
\end{abstract}

\footnotetext[1]{E-mail addresses: \href{blj20@mails.tsinghua.edu.cn}{blj20@mails.tsinghua.edu.cn}(L.Bai), \href{wuweichi@mail.tsinghua.edu.cn}{wuweichi@mail.tsinghua.edu.cn}(W.Wu)} 
\section{Introduction}
The long-run variance plays a central role in the statistical inference of time series linear models.  Consider the following functional linear model (\citep{zhou2010simultaneous}) for the time series observations $(y_{i,n},x_{i,n})\ (i=1, \ldots,n)$,
\begin{align}
    y_{i,n} =   x^{\T}_{i,n} \beta_{i,n} + e_{i,n},\quad \beta_{i,n}=\beta(i/n)\quad (i=1,\ldots, n), \label{eq:tvlinear}
\end{align}
where $ x_{i,n}$ is a $p$-dimensional covariate process whose first element is $1$ and $e_{i,n}$ is the error process, both of which can be non-stationary, $(\beta_{i,n})$ are the time-varying regression coefficients and $\beta(\cdot):[0,1]\to \mathbb R^p$ is the regression function. If  $\hat \beta_{i,n}$ is assumed to be constant or smoothly changing,  it can be estimated by ordinary least squares  or nonparametric methods (\citep{chan2010local}, \citep{li2011non}, \citep{su2019}) 
and the variation of such estimates is determined by the $p$-dimensional possibly time-varying {\it long-run covariance matrix} of $(x_{i,n}e_{i,n})$, see \eqref{eq:long-run variance} for exact definition.  
Classic estimation of the long-run covariance matrix in linear models, see for instances \citep{newey1987} and \citep{andrews1991heteroskedasticity},  
requires consistent estimation of regression coefficients,
which is difficult to achieve under structural changes. A prevalent approach to overcome the obstacle is the difference-based estimation, which is built on $y_{i+m,n}-y_{i,n}$ where $m$ is a diverging tuning parameter such that $m=o(n)$ so that most $y_{i+m,n}-y_{i,n}$ are approximately zero-mean under both smooth and abrupt structural changes. Most of the existing results assume stationarity, see for example,  \citep{gomez2017sjs} and \citep{chan2021optimal}. Exceptions include \citep{dette2019detecting},  which allows non-stationary errors. To the best of the author's knowledge, the existing difference-based long-run variance estimators for time series data are designed for $p=1$. When $(e_{i,n})$ are \LJ{independent and identically distributed},  the difference-based estimators for variance has been studied by for example \citep{muller1987estimation}, \citep{hall1990bio} and \citep{brown2007aos}.
% \WC{Lj  add some citaiton. I think there is an AOS from a UPenn big prof. }. 

For $p>1$, the long-run variance is referred as the long-run covariance matrix, whose estimation is much more involved, see \citep{jansson_2002}, \citep{HIRUKAWA2021} and reviews therein. %but cannot be avoided in contrast to the stationary case ( see \citep{shao2015self} and review therein).
% \WC{Lj: please check the reference. I think \citep{chan2021optimal},\citep{dette2019detecting},\citep{BeranLongmemory} are for $p=1$. We should focus on $p>1$.}
Most existing methods are plug-in methods that utilize estimated residuals  which have been widely applied to the goodness of fit tests, tests of structural breaks (\citep{aue2013}, \citep{wu2018gradient}, \citep{kao2018}), detecting gradual changes (\citep{vogt2015detecting}), simultaneous confidence bands for coefficient functions (\citep{zhou2010simultaneous}), tests for long memory (\citep{BeranLongmemory}, \citep{bai2021}), etc. Since they \LJ{depend  critically on the accurate pre-estimation} of regression coefficients, they are inconsistent under abrupt structural breaks, causing the notorious non-monotonic power \LJ{(\citep{Kejriwal2009})} when applied to structural stability tests. Moreover,  most nonparametric specification tests involve both $\hat \beta_{i,n}$ and the  long-run covariance matrix,  while  $\hat \beta_{i,n}$ is also used in the formula of the plug-in estimator for the long-run covariance matrix. Therefore, the plug-in estimate tends to sensitize those tests to the tuning parameters chosen for $\hat \beta_{i,n}$. 

%In this paper, we develop a new difference-based estimator of the time-varying long-run covariance matrix to overcome the aforementioned difficulties for functional linear models with non-stationary dependent errors. Our estimator has an extra debiased term compared to their univariate counterpart because the extension of the latter estimator to linear models with $p>1$ is asymptotically biased. We derive the consistency of the new estimator under time series nonstationarity allowing flexible predictor-error relationships in Section \ref{sec:theory}. In particular, our difference-based estimator achieves faster convergence rates under certain moment conditions than that of the plug-in estimator in \citep{zhou2010simultaneous}.  We then apply the new estimator to the structural stability tests and the long-memory tests, studying the theoretical and numerical performance in Section \ref{sec:app} and \ref{sec:sim}, respectively. An empirical illustration using public health data can be found in Section \ref{sec:data}. %demonstrate the usefulness of our approach by means of real data analysis.

\section{Preliminaries}
% \WC{LJ:Please rewrite. Introduce the setting of locally stationary covariates and errors, and give the formula of tv-lrv covariance matrix. Introduce the time-varying model $y$ with (piecewise) smooth regression coefficient $\beta$ }
In this paper, we consider a general form of non-stationary called local stationarity for the covariate and the error processes of \eqref{eq:tvlinear}, which has received  substantial attention in the literature.  We employ the definition of locally stationary processes  based on Bernoulli shift processes, see \citep{wu2005nonlinear} and 
\citep{zhou2010simultaneous}, while there are also many other formulations, see \citep{dahlhaus1997}, \citep{nason2000wavelet}, and 
 \citep{dahlhaus2019bej} for a comprehensive review. \par We start by introducing necessary notation that will be used in the rest of the paper. Let $1(\cdot)$ denote the indicator function.
Define $\lf a \rf$ as the largest integer smaller than $a$, and $\lambda_{\min}(A)$ as the smallest eigenvalue of any symmetric squared matrix $A$. Let $|\cdot|$ denote the absolute value for scalars and the Frobenius norm for matrices. Let $Z$ denote the set of integers. Let  $\F_i = (\LJ{\ldots}, \eta_{i-1}, \eta_i)$, where $(\eta_i)_{i \in Z}$ are \LJ{independent and identically distributed} random elements, and the couple process $\F_{i,\{0\}}=(\LJ{\ldots},\eta_{-1},\eta_{0}^{\prime},\eta_{1},\LJ{\ldots},\eta_{i})$, where $(\eta_{i}^{\prime})_{i\in  Z}$ is the \LJ{independent and identically distributed} copy of $(\eta_{i})_{i\in  Z}$. We write $f(\cdot) \in C^q[0,1]$ if $f(\cdot)$ has $q$th order continuous derivative.  Consider the data generating mechanism $L(t,\F_i) \in R^{p}$, where $L$ is a filter function such that $L(t, \F_i)$ is well-defined. We say $L(t,\F_0) \in \mathrm{Lip}_s(I)$ if there exists a constant $c > 0$ such that for any $t_1, t_2 \in I, t_1 < t_2$,
\begin{align}
    \|L(t_1,\F_0)-L(t_2,\F_0) \|_s  \leq c|t_1-t_2|,
\end{align}
where $\|X\|_s = \LJ{\{E(|X|^s)\}^{1/s}}$. To measure the dependence of time series,  we adopt the physical dependence (\citep{wu2005nonlinear}).  For  $L(t,\F_i)$ on interval $I$, the physical dependence in $\mathcal L^r$ norm is defined by
% \WC{Lj: the next formula is wrong}
\begin{align}
    \delta_r(L,k,I) =\sup_{t\in I} \|L(t,\F_{k})-L(t, \F_{k,\{0\}})\|_r.
\end{align} 
%by viewing the filter $L(\cdot,\cdot)$ as a input-ouput system of $(\eta_i)$.
In the following, we give the definition of local stationarity and short-range dependence. 
\begin{definition}
The process $G(t, \F_i) \ (i\in Z)$ is of $r$-order short-range dependence on interval $I$ if $\sup_{t \in I} \| G(t, \F_0) \|_{r} <\infty $, $\delta_{r}(H, k, I) = O(\chi^k)$,  \LJ{for some $\chi \in (0,1)$}, $r \geq 1$, and $s$-order locally stationary on interval $I$, $s \geq 2$, if
 $G(t, \F_0) \in \mathrm{Lip}_{s}(I)$.
\end{definition}
In the functional linear model \eqref{eq:tvlinear}, we consider  $ x_{i,n}= W(t_i, \F_i)$, $e_{i,n}=H(t_i, \F_i)$, where  $t_i=i/n$ , $W$ and $H$ are measurable nonlinear filters mapping from $[0, 1]\times \mathbb R^{Z}$ to $\mathbb R^p$ and $(-\infty, 1]\times \mathbb R^{Z}$ to $\mathbb R$, respectively. Permitting $i$ to approach $-\infty$ enables us to develop our proposed method under long memory, see \cref{sec:app2} for detailed discussion. We further assume $e_{i,n}$ has the following state-heteroscedastic form, i.e., 
$e_{i,n} = \tilde H(t_i, \mathcal H_i ) \tilde G(t_i,  \mathcal{G}_i) $ for $-\infty <i\leq n$, where the nonlinear filters $\tilde H(t, \cdot)$ and  $\tilde G(t,\cdot)$ are $\mathcal H_i$ and $\mathcal G_i$ measurable functions for $t\in(-\infty,1]$, and filtrations  $\mathcal{G}_i$ and $\mathcal H_i$ are sub-$\sigma$-fields of $ \F_i$, independent of each other. Furthermore, $ \mathcal H_i \subset \sigma( x_{1,n},..., x_{i,n})$ if $1\leq i\leq n$, and for any $t\in (-\infty, 1]$, $E( \tilde G(t,  \mathcal{G}_i))  = 0$. The above formulation
admits the heteroscedastic errors considered in \citep{lackoffit} and \citep{kulik2012conditional}, where they assume $\tilde H(t_i, \mathcal H_i)$ to have the form of  $s( x_{i,n})$ for some unknown smooth function $s(\cdot)$. Moreover, $\tilde G(t_i,\mathcal G_i)$ allows the conditional heteroscedasticity  as considered by assumption 1(b) of \citep{cavaliere2017quasi}. 
 Further, let $ {U}\left(t_i, \mathcal{F}_{i}\right) = W(t_i, \F_i)H(t_i, \F_i) $ such that $x_{i,n}e_{i,n}= {U}\left(t_i, \mathcal{F}_{i}\right)\ (i=1,\ldots, n)$. The time-varying long-run covariance function for the functional linear model \eqref{eq:tvlinear} is defined by
\begin{align}
 {  \Sigma}(t)=\sum_{j=-\infty}^{\infty} \operatorname{cov}\LJ{\{}{U}\left(t, \mathcal{F}_{0}\right), {U}\left(t, \mathcal{F}_{j}\right)\LJ{\}}\quad (t \in [0,1]).\label{eq:long-run variance}
\end{align}
% which can be shown to be equivalent to the definition in \eqref{eq:long-run variance0}, see Lemma 3- Lemma 5 of \citep{zhou2010simultaneous}.  
For stationary $x_{i,n}$ and $e_{i,n}$, $\Sigma(t)$ will be time invariant. In this case for linear models with no structural changes, \citep{zhoushao2013} proposes a self-normalization method for statistical inference so that the estimation of the long-run covariance matrix can be avoided. However, their method relies crucially on strict stationarity. Therefore, for the statistical inference of non-stationary time series linear models, the accurate and robust estimation of  $\Sigma(t)$ is essential.

\section{The debiased difference-based estimator}\label{sec:diff}

% Difference-based estimators have been successful in removing deterministic trend without pre-estimation, see for instance
 %\citep{muller1987estimation}, \citep{hall1990bio}, %\citep{gomez2017sjs}
 
%We present the difference-based long-run covariance estimator for $p\geq 2$. For $p=1$ we recommend the difference statistic proposed by \citep{dette2019detecting} for $\sigma^2_H(t)$ which is built on the difference of $y_{i,n}$. %Nevertheless, the estimation of ${  \Sigma}(t)$ is much more involved since $(  x_{i,n} e_{i,n})$ is not directly observable and its magnitude of stochastic order is different under null and the alternatives.
%For $p \geq 2$, a direct extension of the estimator in 
When $p=1$,
\citep{dette2019detecting}   proposes the following   difference-based estimator $\acute{{  \Sigma}}(t)$ based on the difference of $  x_{i,n}y_{i,n}$ for the long run variance of $e_{i,n}$. Let $Q_{k, m}=\sum_{i=k}^{k+m-1}   x_{i,n}y_{i,n}$, and  for $t \in [m/n,1-m/n]$, 
\begin{align}
     \acute{{  \Sigma}}(t)=\sum_{j=m}^{n-m} \frac{m   \Delta_{j}  \Delta_{j}^{\T}}{2}\omega(t, j), \quad \Delta_{j}=\frac{  Q_{j-m+1, m}-   Q_{j+1, m}}{m},\label{eq:diff_based}
\end{align}
where for some bandwidth $\tau_n$ and the kernel function $K(\cdot)$ with support $(-1,1)$,
     \begin{align}
      \omega(t, i)= K_{\tau_n}\left(t_i-t\right) / \sum_{i=1}^{n} K_{\tau_{n}} \left(t_i-t\right), ~~K_{\tau_n}(\cdot)=K(\cdot/\tau_n).
     \end{align}
For $t \in [0,m/n)$,  $\acute{{  \Sigma}}(t) = \acute{{  \Sigma}}(m/n)$ and for $t \in (1-m/n,1]$, $\acute{{  \Sigma}}(t) = \acute{{  \Sigma}}(1-m/n)$. \par

Unfortunately, in \cref{proof:thm7.1}, we find that the difference-based estimator $\acute {  \Sigma}(t)$ is asymptotically biased for $p\ge 2$. \LJ{The bias is negligible when the covariates are fixed and continuous as assumed in Assumption 1 of \citep{zhou2015optimal}, since the bias is caused by the stochastic variation of covariates. We explain this in detail in Section \ref{sec:bias} of the supplemental material.} %Hence, bias correction for $\acute{{  \Sigma}}(t)$ is required and it is desired to use a $b_n$-free statistic so that \cref{algorithm}  will be less sensitive to the smoothing parameter $b_n$ and hence more stable.
In \cref{proof:thm7.1}, we also show that $E(\acute {  \Sigma}(t))-  \Sigma(t)$ can be uniformly approximated by the expectation of the following $ {\Sigma}^A(t)$:
     \begin{align}
         \Sigma^A (t) = \sum_{j=m}^{n-m} \frac{m   A_{j,m}   A_{j,m}^{\T}}{2}\omega(t, j)\quad (t \in [m/n,1-m/n]),\nonumber
    \end{align}
where for $j = m,\ldots, n$,
$$ A_{j,m} = \frac{1}{m} \sum_{i=j-m+1}^j\big\{  x_{i,n}   x_{i,n}^{\T} {  \beta}(t_i) -   x_{i+m, n}   x_{i+m, n}^{\T} {  \beta}(t_{i+m})\big\},$$
 and $  \Sigma^A (t) =   \Sigma^A(m/n)$  for $t \in [0,m/n)$, $  \Sigma^A(t) =   \Sigma^A(1-m/n)$ for $t \in (1-m/n,1]$.
 %  Though the non-negligible bias of $\acute{  \Sigma}(t)$ is by nature sophisticated,
 %   \cref{proof:thm7.1} in the supplement shows that  {\it $\E  \Sigma^A (t)$ is the leading term of the bias of $\acute{{  \Sigma}}(t)$}.
% Therefore, we propose the following {\it debiased} difference-based estimator of ${  \Sigma}(t)$. 
To utilize $  \Sigma^A(t)$, we  shall substitute $  \beta(\cdot)$ in $  A_{j,m}$ via an estimator without introducing additional tuning parameters.
For this purpose,  
define
    $
    \tilde{  Y}_{i,m} =   x_{i,n} y_{i,n}-  x_{i+m, n} y_{i+m, n},  \tilde{  X}_{i,m} =   x_{i,n}   x_{i,n}^{\T}-  x_{i+m, n}   x_{i+m, n}^{\T}
    $, and $\tilde{  E}_{j,m} = \frac{1}{m} \sum_{i=j-m+1}^j  (  x_{i,n} e_{i,n}-  x_{i+m, n} e_{i+m, n})$.
    By the continuity of $  \beta(t)$,
    \begin{align}
       \frac{1}{m} \sum_{i=j-m+1}^j \tilde{  Y}_{i,m} =   A_{j,m} +\tilde{  E}_{j,m} \approx \frac{1}{m} \sum_{i=j-m+1}^j \tilde{  X}_{i,m}   \beta(t_i) + \tilde{  E}_{j,m}\notag\\
       \approx \frac{1}{m}   \beta(t_j) \sum_{i=j-m+1}^j \tilde{  X}_{i,m} + \tilde{  E}_{j,m}.\label{Motivation}
    \end{align}
      Since $\LJ{E}(\tilde{  E}_{j,m}) = 0$, the  random vector $\surd{m} \tilde{  E}_{j,m}$ is  $\LJ{O_p}(1)$.  Treating $\sum_{i=j-m+1}^j \tilde{  Y}_{i,m}/\surd{m}$ as the response variable, $ \sum_{i=j-m+1}^j \tilde{  X}_{i,m}/\surd{m} $ as covariates and  $\surd m\tilde{  E}_{j,m}$ as errors, \eqref{Motivation} motives us to approximate $  \beta(t)$ in $  A_{j,m}$ by 
    \begin{align}
        \breve {  \beta}(t) =   \Omega^{-1}(t)  \varpi (t),\label{eq:brevebeta}
    \end{align}
    where $ \Omega(t)$ and $  \varpi (t)$ are the smoothed versions of $$\acute{  \Delta}_{j}/2 = \frac{1}{2m}\sum_{i=j-m+1}^{j} \tilde{  X}_{i,m}\tilde{  X}_{i,m}^{\T}~\text{and}~\breve {  \Delta}_{j}/2 = \frac{1}{2m}\sum_{i=j-m+1}^{j} \tilde{  X}_{i,m}^{\T}\tilde{  Y}_{i,m},$$ i.e., for $t \in [m/n, 1-m/n]$,
    \begin{align}
        %  &\acute{  \Delta}_{j} = \frac{1}{m}\sum_{i=j-m+1}^{j} \tilde{  X}_{i,m}\tilde{  X}_{i,m}^{\T}, \quad \breve {  \Delta}_{j} = \frac{1}{m}\sum_{i=j-m+1}^{j} \tilde{  X}_{i,m}^{\T}Y_{i,m},\nonumber\\ 
         & \Omega(t) = \sum_{j=m}^{n-m} \acute {  \Delta}_{j} \omega(t, j)/2,\nonumber \quad   \varpi (t) = \sum_{j=m}^{n-m} \breve {  \Delta}_{j} \tilde \omega(t, j)/2,
     \end{align}
     where $\tilde \omega(t, i)= K\{ (t_i-t)/\tau_n^{3/2} \} / \sum_{i=1}^{n} K \{(t_i-t)/\tau_n^{3/2} \}$, 
     while for $t \in [0,m/n)$, $  \Omega(t) =   \Omega(m/n)$, $ \varpi (t) =   \varpi (m/n)$ and for $t \in (1-m/n,1]$, $  \Omega(t) =   \Omega(1-m/n)$, $  \varpi (t) =   \varpi (1-m/n)$.  \LJ{The estimator $\breve \beta(t)$, which is also based on difference series, is accurate  except in the vicinity of abrupt changes.} Fortunately, the effect of abrupt changes can be mitigated by the local averaging in the formula of $A_{j,m}$.
  %  \begin{align}
     %  \breve \beta(t) = \Omega^{-1}(t)\varpi (t),\quad t \in [0,1].\label{eq:breve_beta}
    % \end{align}
   Replacing $  \beta(t)$ by $\breve {  \beta}(t)$ in $  A_{j,m}$ we obtain \begin{align}
    \hat{  A}_{j, m} = \frac{1}{m}\sum_{i= j-m+1}^{j} \{  x_{i,n}   x_{i,n}^{\T}\breve {  \beta}(t_i)-  x_{i+m, n}   x_{i+m, n}^{\T}\breve {  \beta}(t_{i+m})\},\nonumber %\label{brevebeta}
\end{align}
and the corresponding debiased difference-based estimator $\hat {  \Sigma}(t)$ for $t \in [0,1]$: 
     \begin{align}
         \hat {  \Sigma}(t) = \acute {  \Sigma}(t) - \breve {  \Sigma}(t), \quad \breve {  \Sigma}(t)= \sum_{j=m}^{n-m} \frac{m\hat{  A}_{j,m} \hat{  A}_{j,m}^{\T}}{2}\omega(t, j)\label{eq:diff_correct},
     \end{align}
 which is robust to structural breaks in regression coefficient functions due to differencing, and except for $m$, $\tau_n$ used for $\acute{\Sigma}(t)$, \LJ{the correction $\breve \Sigma(t)$ does not involve additional tuning parameters. Thanks to the residual-free formula, the estimator \eqref{eq:diff_correct} can preserve its consistency even when it is challenging to estimate $\beta(t)$ accurately.}
 %\WC{Why is $\breve{\beta}(t)$ robust to abrupt changes}\LJ{Not robust to abrupt changes (unless in the intercept)}
% \WC{Lj: what I want to stress is that the new estimator preserves the advantages of the difference-based lrv estimator for univariate series, which is exactly (not only extension) $\Acute{\Sigma}(t)$ when $p=1$. I move the requirement of kernel function to the long-memory section }
\section{Consistency under smooth structural changes}\label{sec:theory}
In this section, we discuss the uniform convergence of the debiased difference-based long-run covariance matrix estimator \eqref{eq:diff_correct}  for the functional linear model \eqref{eq:tvlinear} under smooth structural changes and short-range dependence with locally stationary predictors and errors, which accommodates many null hypotheses of nonparametric specification tests. The convergence of $\hat \Sigma(t)$ when both smooth and abrupt changes occur is deferred to \cref{sec:app1}, and the performance of $\hat \Sigma(t)$ under long memory is postponed to \cref{sec:app2}. 
% \WC{Lj, please check the above statements}

%\WC{Move the above to preliminary section}
% \WC{delete Kernel?}
\begin{assumption}
 % The following holds for the coefficient function
%  \begin{enumerate}[label=(C\arabic*)]
%      \item  
%     (C1) The kernel $K(\cdot)$ is continuous, symmetric and supported on $[-1,1]$. %bounded on [-1,1] %and 0 on $(-\infty,-1]\cup [1,\infty)$.
%    \label{A:K}\\
%    (C2)  
    Each coordinate of $  \beta(t)$ lies in $C^3[0,1]$.\label{A:beta}
%  \end{enumerate}
  \label{A:nonpar}
 \end{assumption}
Assumption \ref{A:beta} imposes smooth structural change, i.e.,  the coefficient function $ \beta(\cdot)$ is smooth. In \cref{sec:app1}, we shall relax Assumption \ref{A:beta} to allow abrupt structural changes.
% \WC{Should $u_i$ be $e_i$? Put in the preliminary?}  
% \WC{The next should be assumption for model identification, or put it in the preliminary}

\begin{assumption}\label{Ass-U} \
% The product of covariate process and the error process $ U(t,  \F_i) $ is of $2$-order locally stationarity and $4$-order short-range dependence on $[0,1]$,
%%\begin{enumerate}[label=(A\arabic*)] 
%%$ U(t,  \F_i) \in \mathrm{Lip}_2$, $\sup_{t \in [0,1]} \| U(t,  \F_i) \|_4< \infty$, $\delta_4( U, k) = O(\chi^k)~ \text{for some} ~\chi \in (0,1)$
% and
  The value $\lambda_{\min}(\Sigma(t))$ is bounded away from $0$ uniformly on $[0,1]$, and each element of ${  \Sigma}(t) \in  C^2[0,1]$.
%     \label{A:U_long-run variance}
%\end{enumerate}
\end{assumption}
Assumption \ref{Ass-U} guarantees the non-degeneracy of the long-run covariance matrix of the process of $ U(t,  \F_i)$ and each component of ${  \Sigma}(t)$ is smooth,  which are common in the analysis of functional linear models of locally stationary time series, see for example \citep{zhou2010simultaneous}.   Let $ J(t, \F_0) =  W(t, \F_0) W^{\T}(t, \F_0)$.
Define $\mu_W (\LJ{t}) = \LJ{E}\{W(t, \F_0)\}$, $M(t) =\LJ{E}\{J(t, \F_0)\}$, $t \in [0,1]$. 
\begin{assumption}\label{Ass-W}   
 Each element of the functions
 {$ M(t) \in C^1[0,1]$, $ \mu_W (t) \in C^1[0,1]$, and $\inf_{t \in [0,1]}\lambda_{\min}(M(t)) > 0$}. The covariate process $W(t, \F_i)\ (i=1,\ldots, n)$ is of $4$-order local stationarity and $16\kappa$-order short-range dependence on $[0,1]$ for  some $\kappa \geq 1$.\label{E:WW} \label{A:Mt}
 % \WC{Lj: I think the assumptions might be too strong for the consistency of the lrv. These are for long memory settings?}
% ${W}\left(t, \mathcal{F}_{i}\right) \in \mathrm{Lip}_{2}, \text { and } \sup _{t \in [0,1]}\left\|{W}\left(t, \mathcal{F}_{i}\right)\right\|_{8}<\infty$, %  \WC{Lujia: we need Lipschitz Continuous condition for $ W$}
%$\delta_{8}( W^{(-1)},k)=O(\chi^k)$~ \text{for some}~ $ \chi \in (0,1)$.%\end{enumerate}
% \WC{Please define $ X^{(-1)}$ for any $p$ dimensional vector $ x$, and $ X^{(-1)}=\emptyset$ if $p=1$. }\edit{done}
\end{assumption}

Assumption \ref{Ass-W} requires that $ M(t)$ is non-degenerate, implies that  $\sup_{t\in [0,1]} \| J(t, \F_0)\|_{8\kappa} < \infty$, $\delta_{8\kappa}( J, k) = O(\chi_1^{k})$ for some $\chi_1\in(0,1)$, and ensures that each element of $( x_{i,n}  x_{i,n}^{\T})$ is $2$-order locally stationary. 
 %To guarantee uniform consistency of $\hat {  \Sigma}(t)$, we need additional assumptions.

%  The following assumption is required to guarantee uniform consistency of $\hat {  \Sigma}(t)$ and $m^{-2d} \hat {  \Sigma}(t)$ when $d=0$ and $d>0$ respectively.
 %The following assumption is required to guarantee uniform consistency of $\hat {  \Sigma}(t)$ when $d=0$, and sufficiently small order of estimation error for $d>0$.
% \begin{assumption}
%Each element of ${  \Sigma}(t) \in  C^2[0,1]$.\label{E:smooth}
%%   \WC{use $\mathcal C^2[0,1]$}
%  \label{A:U_smooth}
% \end{assumption}
The following assumption ensures the invertibility of $\Omega(t)$ in \eqref{eq:brevebeta}.
\begin{assumption}\label{Ass-E}
%For some $\chi\in (0,1)$, $\kappa \geq 4$, $ J(t, \F_0) \in \mathrm{Lip}_{2}$, 
% $\sup_{t \in [0,1]} \|  W(t, \F_0) \|_{16\kappa}<\infty $, $\delta_{16\kappa}( W, k) = O(\chi^k)$,
 % Let   $\bar{ J}(t, \F_0) =  J(t, \F_0) - \E  J(t, \F_0)$ . 
Each element of the covariance function $\mathrm{cov} \left\{{ J}(t, \F_0),{ J}(t, \F_0)\right\} \in C^2 [0,1]$ and its smallest eigenvalue is strictly positive on $[0,1]$.
% \WC{do we need each element of?}
%\WC{have to be further rewritten for $e_{i,n}$}
%   \item  Each element of ${  \Sigma}(t) \in  C^2[0,1]$.\label{E:smooth}
%   \WC{use $\mathcal C^2[0,1]$}
%  \label{A:U_smooth}
%\end{enumerate}
\end{assumption}
% \WC{Maybe we put the following assumption in the preliminary, as well as Remark 1. }
\begin{assumption} The process $H(t, \F_i)\ (i \in Z)$ satisfies $16\kappa$-order of  short-range dependence and $4$-order of local stationarity on $[0, 1]$.
% $H(t, \F_0) \in \mathrm{Lip}_{2\kappa}(-\infty, 1]$, $\sup_{t \in  (-\infty, 1]} \| H(t, \F_0) \|_{16\kappa} <\infty $, $\delta_{16\kappa}(H, k,(-\infty, 1]) = O(\chi^k)$. 
 \label{E:HW}\label{E:U} \label{B:H_delta}
\end{assumption}

 Assumptions \ref{E:WW} and \ref{E:U} ensure that $U(t, \F_i)\ (i \in Z)$ is of $2$-order local stationarity and $8\kappa$-order of short-range dependence on $[0,1]$.
The existence of  $16\kappa$th moments of covariates and errors are assumed for technical convenience, and it is satisfied by sub-exponential random variables. We conjecture that it can be relaxed by substantially more involved mathematical arguments, see the simulation study of \citep{bai2021}. \LJ{Specifically, %we allow $\kappa = 1$, i.e., $16$-order of short-range dependence, which  can 
the condition of $16\kappa$-order dependence can be replaced by sub-exponential moment conditions and dependence measure in $\mathcal L^1$ norm. Further technical discussion on Assumption 5 can be found in \cref{sec:high} of the supplement.}
 Assumption \ref{Ass-E} is mild, and it excludes the scenario in which all of the time series covariates reduce to deterministic smooth trends; in this case, we recommend using the direct multivariate extension of the difference-based long-run covariance estimator of \citep{dette2019detecting}, namely $\acute \Sigma(t)$. \LJ{The following assumption gives the properties of the kernel function. The use of different kernels is discussed in \cref{sec:sel}.
\begin{assumption}
     The kernel function $K(\cdot)$ is a continuously differentiable, symmetric density function and supported on $(-1,1)$.
    \label{A:K}
 \end{assumption}
 }
Let ${I} = [\gamma_n, 1-\gamma_n] \subset (0,1)$, where $\gamma_{n}=\tau_{n}+(m+1) / n$. 
%The following theorem justifies the uniform consistency of  $\hat{  \Sigma}(t)$ over $\mathcal I$. 
% \WC{will double check the following results after checking the proofs}
\begin{theorem}  \label{thm:lrv_diff}
    Under Assumptions \ref{A:nonpar}, \ref{Ass-U},  \ref{Ass-W}, \ref{Ass-E}, \ref{A:K} and \ref{B:H_delta} with constant $\kappa \geq 1$, suppose
    $m = O(n^{1/3})$, $m \to \infty$, $\tau_n \to 0$, $m / (n \tau_n^{3/2 + 2/\kappa}) \to 0$, $\tau_n^{3-1/\kappa}\surd{m} \to 0$, $m\tau_n \to \infty$, we have
    \begin{equation}
      \sup _{t \in {I}}\left|\hat{{ \Sigma}}(t)-{  \Sigma}(t)\right|= 
      \LJ{O_p}\left(\frac{m^{1/2}}{n^{1/2} \tau_n^{3/4+1/\kappa}} + \frac{1}{m}+m^{1/2}\tau_n^{3-1/\kappa} \right).\nonumber
      \end{equation}
  \end{theorem}

%\LJ{Hi Weichi, I think the rate here is better than that in \citep{zhou2010simultaneous}.}
%\edit{confirm: 1. attaining better rates 2. a wider range for the tuning parameter selection? Check the my rewriting.}
The above equation and the bandwidth conditions for Theorem \ref{thm:lrv_diff} ensure the uniform consistency of $\hat{\Sigma}(t)$.  \LJ{As shown by \cref{lm:bias} in the supplement, the estimator $\breve \beta(t)$ is consistent when there are no jump points. As pointed out by a referee, it is viable to use any pilot estimator of $\beta (t)$ that satisfies the conditions in \cref{sec:high} of the supplement in $\breve \Sigma(t)$ for debiasing.}\par
% As pointed out by referee, it is viable to use any pilot estimator that satisfies $\sup_{t \in \I}\| (\beta(t) - \breve \beta_0(t))\|_{4\kappa} = \Op(\tau_n^2)$  to  achieve the debias effect.

We compare our results with \citep{zhou2010simultaneous} which proposes a plug-in estimator of long-run covariance matrix using nonparametric residuals. %, i.e, using  nonparametric residuals in the infeasible long-run covariance matrix estimator.
% which has been proposed and investigated in  Theorem 5 of \citep{zhou2010simultaneous}.  
According to Theorem 5 in \citep{zhou2010simultaneous}, %their rate is optimized at $b_n \asymp n^{-1/5}$, $m \asymp n^{1/4}$, $\tau_n \asymp n^{-1/8}$,  depending on the moment conditions. 
their best approximation rate can be  close to but not faster than $n^{-1/4}$.
In contrast, our uniform rate in \cref{thm:lrv_diff} is $n^{-4/15+2/(15\kappa)}$ by
taking $m \asymp n^{4/15}$, $\tau_n \asymp n^{-2/15}$, which is better  when $\kappa$ is sufficiently large.  \LJ{In unreported studies, we find that the performance of plug-in estimators using $\breve \beta(\cdot)$ is superior to \citep{zhou2010simultaneous} but inferior to our proposed estimator. This estimator is also available in R package \texttt{mlrv}.} %\LJ{The moment condition of $\kappa$ is due to technical convenience. We conjecture that under weaker moment assumptions our results still hold, see the simulation study of \citep{bai2021} for covariates and errors that have heavier tail.}

\section{Applications}\label{sec:app}
%In this section, we elaborate two applications in which the long-run covariance matrix estimators are important, i.e., structural change detection and tests for long memory. We show that with the difference-based long-run variance-covariance estimate \eqref{eq:diff_correct} both tests are asymptotically correct and powerful. Moreover, the application of the difference-based long-run variance-covariance estimate \eqref{eq:diff_correct} can effectively and consistently yield  higher power for all the tests considered compared to the residual-based testing procedures, see \cref{sec:sim} for simulation evidence.
\subsection{Structural change detection with monotonic power}\label{sec:app1}
The first application is the detection of structural changes in the stochastic linear regression  
\begin{align}
     y_{i,n} = x_{i,n}^{\T}  \beta_{i,n} + e_{i,n} \quad (i = 1,\ldots, n),
\end{align}
% \WC{Please change to $\beta_{i,n}$}
where $(x_{i,n})_{i=1}^n$ is the $p$-dimensional covariate process and $(e_{i,n})_{i=1}^n$ is the error process. The test for structural changes  of $(\beta_{i,n})$ considers the hypothesis $H_0: \beta_{1,n}=\beta_{2,n}=\LJ{\cdots}=\beta_{n,n}$. Allowing  general non-stationarity in the covariates and the errors, \citep{wu2018gradient} proposes to use the test statistic 
\begin{align}
    T_n = \max_{1 \leq j \leq n}\left| \sum_{i=1}^j \hat e_{i,n} x_{i,n} / \surd{n} \right|,\quad  \hat e_{i,n} = y_{i,n} - x_{i,n}^{\T} \hat{ \beta}_n  \quad (i = 1,\ldots, n), \label{eq:cpstat}
\end{align}
where $\hat{ \beta}_n = \mathrm{argmin}_{\beta} \sum_{i=1}^n (y_{i,n} - x_{i,n}^{\T}  \beta )^2$.
For the alternative hypothesis,  we consider
% \WC{Lujia, the following has identification problems. Maybe set $g(0)=0$}
\begin{align}
y_{i,n}=  x^{\T}_{i,n} \beta(i/n) + e_{i,n}, \quad  \beta(t) =  \beta +  L_ng(t)\quad (i = 1,\ldots, n;\ t\in [0,1]),\label{eq:alt}
\end{align}
where $\beta, g(t) \in \mathbb R^p$, $g(0) = 0$, $L_n$ is a positive real sequence, $\beta(t)$ has potential abrupt changes, i.e., $  \beta(t)=\sum_{j=0}^{q_n}  b_{j}(t) \mathrm{1}(a_{j} \leq t<a_{j+1})$, $0 = a_0 < a_1 < \cdots < a_{q_n} < a_{q_n+1} =1$,
%   \begin{align}
%       \beta_1(t))=\sum_{j=0}^{d} \beta_{1,j}(t) \mathbb I_{\left\{a_{j} \leq t<a_{j+1}\right\}},
%   \end{align}
   $ b_{j}(t) \in \mathcal C^3(a_j, a_{j+1})$, $\sup_{0 \leq j \leq q_n}| b_{j}(a_{j}) -  b_j(a_{j}^{-})| < \infty$, and $q_n$ is the number of abrupt changes.
   % The robustness to smooth and abrupt changes is a well-known advantage of difference-based estimator,  see, e.g., \citep{gomez2017sjs} and \citep{chan2021optimal}.  
\citep{wu2018gradient} proposes a general bootstrap statistics according to \eqref{eq:ncpdist}, which relies on the residual $\hat e_{i,n}=y_{i,n}-x_{i,n}^{\T}\hat \beta_n$. Their test can be applied to piecewise locally stationary covariates and errors, and is unified for testing structural changes in general $M$ estimation. They show their power approaches $1$ under local alternatives $L_n=n^{-1/2}$ or $L_n=o(1)$, $\surd n L_n\rightarrow \infty$. However,  if $L_n$ does not vanish, $\hat \beta_n$ is not consistent and the power of the test in \citep{wu2018gradient} in this case is not theoretically guaranteed.
If we focus on the locally stationary process which is quite general and on least squares regression which is arguably  the most widely applied M-estimator in practice, we could improve the power of the test for structural changes via our proposed long-run covariance matrix estimator $\hat \Sigma(t)$ and an alternative bootstrap procedure to \citep{wu2018gradient}.
%If we can estimate $ \Lambda(\cdot)$ and ${ \Sigma}^{1/2}(\cdot)$, we can use bootstrap statistics to mimic the distribution of $T_n$, namely
For this purpose, we define \begin{align}
    F_r = \max_{m \leq i \leq n-m+1}|\Psi_{i,m}^{(r)} - \hat{ \Lambda}(i/n) \hat{ \Lambda}^{-1}(1) \Psi_{n-m+1,m}^{(r)}|, \label{eq:bootstrap}
\end{align}
where
% \WC{Change all $\hat \Lambda_i$ to $\hat \Lambda(i/n)$. What is $l_n$}
$
    \Psi_{i,m}^{(r)} = n^{-1/2}\sum_{j=1}^i \hat { \Sigma}^{1/2}(t_j) R_j^{(r)}$, $\hat{ \Lambda}(i/n) = \sum_{j=1}^i x_{j,n} x_{j,n}^{\T}/n,
$
  $( R_j^{(r)})_{j=1}^n$ are the \LJ{independent and identically distributed} standard norm\LJ{al} random variables independent of data and are independently generated  in the $r_{th}$ bootstrap iteration, and $ \hat \Sigma(t) $ is the proposed difference-based estimator of $ \Sigma(t)$. Let $B$ denote the number of bootstrap iterations.
% \begin{proposition} \WC{Combine this with Theorem 2} \label{prop:bootstrap}
%     Under the conditions of \cref{prop:cptest}, if $\sup_{t \in [0,1]}|\hat{ \Sigma}^{1/2}(t) - { \Sigma}^{1/2}(t)| = \op(1)$, $l_n \asymp \surd{n}$, we have 
%     \begin{align}
%         F_r \Rightarrow  \sup_{t \in (0,1]}|U(t) -  \Lambda(t)  \Lambda^{-1}(1) U(1)|.
%     \end{align}
% \end{proposition}
% The uniform consistency of the difference-based estimator of \eqref{eq:diff_correct} is guaranteed by \cref{thm:lrv_diff} and an application of \citep{yu2015useful}.
% \WC{delete}Specifically, we can include the method of \WC{I don't think this is correct}\citep{wu2018gradient} as an example by letting $\hat { \Sigma}^{1/2}(t_j) =  \sum_{r = j}^{j + m - 1} \hat e_{r,n} x_{r, n}/\surd{m} $, $q_n = \surd{n-m+1}$.
Let $F_{(1)} \leq F_{(2)} \leq \cdots \leq F_{(B)}$ be the order statistics of $(F_r)$. We reject the structural stability test at the significance level of $\alpha$ if $T_n$ is greater than the $F_{\lfloor (1-\alpha)B \rfloor}$.
We proceed to relax Assumption \ref{A:nonpar} in \cref{thm:lrv_diff} allowing for the possible presence of abrupt changes and discuss the property of the bootstrap procedure.
\begin{theorem}\label{thm:cp}
 Under the Assumptions \ref{Ass-U},  \ref{Ass-W}, \ref{Ass-E}, \ref{A:K} and \ref{B:H_delta} with constant $\kappa \geq 1$, and the bandwidth conditions   $m = O(n^{1/3})$,  $\tau_n \to 0$,  $n \tau_n^{7/2}/m \to \infty$, $m\tau_n^{3} \to \infty$, assuming  $q_n = o\{\min(\tau_n^{-1/2}, n\tau_n / m^2)\}$,  under the alternative hypothesis \eqref{eq:alt},  we have
  \begin{equation}
      \sup _{t \in \mathcal{I}}|\hat{{ \Sigma}}(t)-{ \Sigma}(t)|= o_p\left(1 \right),
      \end{equation}
and as a result there exists a $p$-dimensional zero-mean Gaussian process $Z(t)$ with covariance function $\gamma(t,s) = \int_0^{\min(t,s)}  \Sigma(r)dr$ such that 
 \begin{align}
        F_r \Rightarrow   \sup_{t \in (0,1]}|G(t)| = \sup_{t \in (0,1]}|Z(t) -  \Lambda(t)  \Lambda^{-1}(1) Z(1)|, \quad \Lambda(t) = \int_0^t M(s) ds.
    \end{align}
\end{theorem}
\LJ{ Taking $m \asymp n^{4/15}$, $\tau_n \asymp n^{-2/15}$, which can achieve  $q_n = o(n^{1/15})$ for sufficiently large $\kappa$, which allows the number of abrupt jumps diverges as $n \to \infty$, though the estimator $\breve \beta(t)$ is inconsistent
due to the inconsistency of smoothing in the neighborhood of discontinuous points.} Define 
% \WC{$\Lambda(t)$ has defined twice}
\begin{align}
   \Lambda(t) = \int_0^t M(s) ds, \quad \Lambda(s,  g(\cdot)) =\int_0^s M(r)  g(r) dr,
\end{align}
and $F(t, g(\cdot)) =    \Lambda(t, g(\cdot)) -   \Lambda(t) \Lambda^{-1}(1) \Lambda(1,  g(\cdot))$. The following proposition gives the limiting distribution of the test statistic $T_n$ under the null hypothesis and ensures the monotonic power of the bootstrap procedure under $H_A:  \beta(t) =  \beta + L_n  g(t)$.
\begin{proposition}    \label{prop:stat_alt} \label{prop:cptest}
   (i) Under the conditions of   
\cref{thm:lrv_diff} and the null hypothesis of no structural changes, we have
    \begin{align}
        T_n \Rightarrow \sup_{t \in (0,1]}|G(t)|,\label{eq:ncpdist}
    \end{align}
    where $G(t)$ is as defined in \cref{thm:cp}.
   \par 
   (ii) Under the conditions of \cref{thm:cp} and the alternative hypothesis \eqref{eq:alt} with $L_n = O(1)$, $n^{1/2}L_n \to \infty$, we have 
   $T_n \to \infty$ in probability at the rate $\surd{n}L_n$, and
   $$\big|\LJ{\mathrm{pr}} (T_n \geq \hat q_{1-\alpha}) - \LJ{\mathrm{pr}} \big(\sup_{t \in [0,1]} |G(t)+ n^{1/2} L_n F(t,  g(\cdot))| \geq \hat q_{1-\alpha} \big)\big|=o(1),$$
   where $\hat q_{1-\alpha}$ is the bootstrap critical value of $F_r$ at the significance level $\alpha$.
      % $$T_n /\sqrt{n} \geq L_n  \sup_{t \in [0,1]} |F(t,  g(\cdot))\big| - o(1).$$ 
\end{proposition}
\cref{prop:stat_alt} shows that $T_n$ is of order $\max(1,\surd n L_n)$ under null hypothesis. 
\cref{thm:cp} and \cref{prop:stat_alt} guarantee the asymptotic correctness  of the bootstrap procedure with the difference-based estimator and that its asymptotic power approaches $1$ under the fixed alternative. Meanwhile, the bootstrap procedure with the difference-based estimator can detect the local alternatives at the parametric rate $\surd{n}$ in the sense that if $L_n = n^{-1/2}$, $T_n \Rightarrow \sup_{t\in [0,1]} |G(t) + F(t,  g(\cdot))|.$ 
% Notably, 
% $F_r$ shares the same limiting distribution as the limiting distribution of $T_n$  under the null hypothesis, while \WC{why while} 
\LJ{In contrast,} Proposition B.1 of \citep{wu2018gradient} shows that with the ordinary least squares residuals, the magnitude of $F_r$ using their bootstrap procedure is
$\surd{m}\max(L_n,  \log^2 n /\surd{n}
) \log n$ when $\surd{m} L_n \to \infty$, $L_n \to 0$, implying power loss due to the divergence of $F_r$. 
Therefore, the bootstrap procedure equipped with the difference-based estimator will be more powerful than that with ordinary least squares residuals and overcome non-monotonic power caused by the inflation of the bootstrap statistics in \citep{wu2018gradient}. In earlier work on the remedy of the non-monotonic power for the test of smooth structural changes, \citep{JUHL200914} proposes to estimate the long-run variance via plugging in the nonparametric residuals  for $p=1$. However, the improvement in the power of their approach does not carry over in the presence of abrupt structural changes.  % However, plugging in nonparametric residuals is not satisfactory in the regression setting for obtaining accurate sizes in finite samples.

\subsection{Testing for long memory}\label{sec:app2} %{Asymptotic behavior of the long-run covariance estimator under \texorpdfstring{$H_A$}{HA}}
% \label{bootstrap_asy}
%We use the kernel $K(\cdot)$ is continuous, symmetric and supported on $[-1,1]$, which can be satisfied by kernel functions such as Epanechnikov kernel.
Another application is testing for long memory  in the functional linear model 
\begin{align}\label{model1}
    y_{i,n}=x_{i,n}^{\T} \beta(t_i)+(1-\mathcal B)^{-d}e_{i,n} \quad (i=1,\ldots, n),
\end{align}
where $\mathcal B$ is the lag operator, $d\in [0,1/2)$ is the long-memory parameter. %and  $u_{i,n}= H(t, \F_i)$ is the short-range dependence locally stationary process. 
When $d=0$, the error process of the model is locally stationary and short-range dependent. %$e_{i,n}=u_{i,n}$.
We are interested in the following hypothesis testing problem
\begin{align}\label{WC-hypo1}
    H_0: d=0~~~\text{versus}~~~ H_A: 0<d<1/2. 
\end{align}
The rejection of $H_0$ implies that the short-memory linear model is inadequate for the data and long-range dependence should be considered.
\citep{bai2021} proposes to test $H_0$ using the jackknife corrected nonparametric residuals. They obtain the local linear estimate of $\beta(\cdot)$, i.e., 
\begin{equation}\label{eq:loclin}
(\hat{{\beta}}_{b_{n}}(t), \hat{{\beta}}_{b_{n}}^{\prime}(t))=\underset{ \eta_{0}, \eta_{1} \in \mathbb{R}^{p}}{\arg \min}\sum_{i=1}^{n}\LJ{\{}y_{i, n}-{x}_{i,n}^{\T}  \eta_{0}- {x}_{i,n}^{\T}   \eta_{1}(t_{i}-t)\LJ{\}}^{2} K_{b_{n}}(t_{i}-t),
\end{equation}
% \WC{could simplify}
where $K(\LJ{t})$ is a kernel function with finite support $(-1,1)$, $b_n$ is a bandwidth. Then, they consider the jackknife  estimator $\tilde{{\beta}}_{b_{n}}(t)=2 \hat{{\beta}}_{b_{n} / \surd{2}}(t)-\hat{{\beta}}_{b_{n}}(t)$ of which the asymptotic bias terms involving $ \beta^{\prime \prime }(\cdot)$ in the formula of $\hat \beta_{b_n}$ and $\hat \beta_{b_n/\surd 2}$ are canceled. Let $K^*(\cdot)$ denote the jackknife equivalent kernel $2 \surd{2} K(\surd{2}x) - K(x)$.  For the sake of simplicity, we write $n^{\prime}$ as $\lf nb_n \rf$ for short.
Define the nonparametric residuals and their partial sum as $
  \tilde e_{i,n} = y_{i,n} - x_{i,n}^{\T} \tilde{\beta}_{b_n}(t_i)$ and $\tilde S_{r,n} =\sum_{i=n^{\prime}+1}^r \tilde{e}_{i,n}$, $r=n^{\prime}+1,\LJ{\ldots}, n-n^{\prime}$, respectively. The KPSS, R/S, V/S and K/S-type test statistics of \citep{bai2021} are\\
1. KPSS-type statistic $
    K_n = \frac{1}{n(n - 2n^{\prime})}\sum_{r=n^{\prime}+1}^{n-n^{\prime}} \left(\tilde S_{r,n}\right)^2.$\\
2. R/S-type statistic $
 Q_n = \max_{n^{\prime} + 1 \leq k \leq n - n^{\prime} } 
 \tilde S_{k,n} - \min_{n^{\prime} + 1 \leq k \leq n - n^{\prime} } \tilde S_{k,n}.
$\\
3. V/S-type statistic $
  M_n = \frac{1}{n(n - 2n^{\prime})}\left\{\sum_{k=n^{\prime} + 1 }^{n-n^{\prime} } \tilde S_{k,n} ^2 - \frac{1}{n - 2n^{\prime}}\left(\sum_{k=n^{\prime} + 1}^{n-n^{\prime} } \tilde S_{k,n} \right)^2\right\}.
$\\
4. K/S-type statistic  $
    G_n =  \max_{n^{\prime} + 1 \leq k \leq n - n^{\prime} } 
 \left|\tilde S_{k,n} \right|.
$

\citep{bai2021} proposes to implement the above tests via the following bootstrap-assisted procedure. Define $\hat{{M}}(t) = \sum_{i=1}^n   {x}_{i,n} {x}_{i,n}^{\T} K_{\eta_n}(t_i - t^*)/(n \eta_n)$,
where $t^* = \max\{\eta_n,\min(t,1-\eta_n)\}$ for some bandwidth  $\eta_n \to 0$, $n\eta_n^2 \to \infty$. Let $\hat{ \Sigma}^*(\cdot)$ be any consistent long-run covariance matrix estimator  satisfying the regularity condition 5.1 in their paper,  and $\hat \sigma^{*2}_{H}(t)=(\hat { \Sigma}^*(t))_{1,1} $.  Generate $B$ \LJ{independent and identically distributed} copies of $N(0, I_p)$ vectors $V^{(r)}_i=(V^{(r)}_{i,1},...,V^{(r)}_{i,p})^\T\ (r = 1,\ldots, B)$, and for each $r$ calculate 
        \begin{align}
          \tilde G^{(r)}_k=-\sum_{j=1}^n\left\{\frac{1}{nb_n}\sum_{ i=n^{\prime}+1 }^k x_{i,n}^{\T} \hat{{M}}^{-1}(t_i)  K_{b_n}^*(t_i - t_j)\right\} \hat { \Sigma}^{*,1/2}(t_j){V}^{(r)}_j + \sum_{ i=n^{\prime}+1}^k \hat \sigma^*_{H}(t_i)V^{(r)}_{i,1}\nonumber
        \end{align}
         as well as the bootstrap statistics: $
         \widetilde{ \mathrm{K}}^{(r)}_{n}$, $ \widetilde{\mathrm{RS}}^{(r)}_{n} $, $ \widetilde{\mathrm{VS}}^{(r)}_{n}$ and  $ \widetilde{\mathrm{KS}}^{(r)}_{n}$ which  can be obtained by substituting $\tilde S_{k,n}$ in the corresponding statistics by $\tilde G_{k}^{(r)}$. 
%         $
%          \widetilde{ \mathrm{K}}^{(r)}_{n} = \frac{1}{n(n - 2n^{\prime})}\sum_{s=n^{\prime} + 1}^{ n - n^{\prime}} \left(\sum_{k = n^{\prime} + 1}^s \tilde G^{(r)}_k\right)^2$ \par 
%         $
%       \widetilde{\mathrm{RS}}^{(r)}_{n} = \max_{n^{\prime} + 1 \leq k \leq n - n^{\prime} }
%      \tilde G^{(r)}_k - \min_{n^{\prime} + 1 \leq k \leq n - n^{\prime} } \tilde G^{(r)}_k,
% $\par
%     $
%       \widetilde{\mathrm{VS}}^{(r)}_{n} = \frac{1}{n(n - 2n^{\prime})}\left\{\sum_{k=n^{\prime} + 1 }^{n-n^{\prime} }  (\tilde G^{(r)}_k)^2 - \frac{1}{n - 2n^{\prime}}\left(\sum_{k=n^{\prime} + 1}^{n-n^{\prime} }  \tilde G^{(r)}_k \right)^2\right\}, 
%    $ \par   $
%       \widetilde{\mathrm{KS}}^{(r)}_{n} = \max_{n^{\prime} + 1 \leq k \leq n - n^{\prime} } 
%       \left| \tilde G^{(r)}_k\right|.
%   $ \par
     Let $ \widetilde{ \mathrm{K}}_{n,(1)} \leq \widetilde{ \mathrm{K}}_{n,(2)} \leq \cdots  \leq \widetilde{ \mathrm{K}}_{n,(B)}$ be the ordered statistics of $\widetilde{ \mathrm{K}}_{n}^{(r)}\ (r = 1,\ldots,B)$, and $B^* = \max\{r: \widetilde{K}_{n,(r)} \leq K_n\}$. 
     % Let $\widetilde{\mathrm{RS}}_{n,(1)} \leq \widetilde{\mathrm{RS}}_{n,(2)} \leq \cdots  \leq \widetilde{\mathrm{RS}}_{n,(B)}$ be the ordered statistics of $\{\widetilde{\mathrm{RS}}_{n}^{(r)}\}_{r=1}^B$, $\widetilde{\mathrm{VS}}_{n,(1)} \leq \widetilde{\mathrm{VS}}_{n,(2)} \leq \cdots  \leq \widetilde{\mathrm{VS}}_{n,(B)}$ be the ordered statistics of $\{\widetilde{\mathrm{VS}}_{n}^{(r)}\}_{r=1}^B$, $\widetilde{\mathrm{KS}}_{n,(1)} \leq \widetilde{\mathrm{KS}}_{n,(2)} \leq \cdots  \leq \widetilde{\mathrm{KS}}_{n,(B)}$ be the ordered statistics of $\{\widetilde{\mathrm{KS}}_{n}^{(r)}\}_{r=1}^B$.
    % Let $B_{\mathrm{RS}}^* = \max\{r: \widetilde{\mathrm{RS}}_{n,(r)} \leq Q_n\}$, $B_{\mathrm{VS}}^* = \max\{r: \widetilde{\mathrm{VS}}_{n,(r)} \leq M_n\}$, $B_{\mathrm{KS}}^* = \max\{r: \widetilde{\mathrm{KS}}_{n,(r)} \leq G_n\}$. Then the $p$-value of KPSS-type test is $1-B^*/B$, the $p$-value of the R/S-type test is $1-B_{\mathrm{RS}}^*/B$, the $p$-value of the V/S-type test is $1-B_{\mathrm{VS}}^*/B$, and the $p$-value of the K/S-type test is $1-B_{\mathrm{KS}}^*/B$. Reject $H_0$ at the level of $\alpha$ for each type of test if  its $p$-value is smaller than $\alpha$.
    %      5. Let $\tilde{T}_{n,(1)} \leq \tilde{T}_{n,(2)} \leq \cdots  \leq \tilde{T}_{n,(B)}$ be the ordered statistics of $\{\tilde{T}_{n}^{(r)}\}_{r=1}^B$.
    %      Reject $H_0$ at level $\alpha$ if $T_n >\tilde{T}_{n,(\lfloor B(1-\alpha)\rfloor)}$. Let $B^* = \max\{r: \tilde{T}_{n,(r)} \leq T_n\}$. 
         Then the $p$-value of the KPSS-type test is $1-B^*/B$, and the $p$-values of R/S, V/S, and K/S-type tests can be obtained similarly. Given a nominal level $\alpha$, if the $p$-value is smaller than $\alpha$, we  reject the null hypothesis of short memory.\par
         In this paper we propose to set $\hat \Sigma^{*}(t)=\hat \Sigma(t)$. With the new difference-based long-run covariance matrix estimator, our testing procedure will be more robust than that using the plug-in estimator for $\hat \Sigma^*(t)$, since for the latter procedure, the parameter $b_n$ in \eqref{eq:loclin} for the estimation of the regression coefficients will additionally affect the estimate of the long-run covariance matrix through the nonparametric residuals $\tilde e_{i,n}$ as well as the selection of $m$ and $\tau_n$, as indicated by the discussion of Theorem 5 in  \citep{zhou2010simultaneous}  that $m$ and $\tau_n$ should be chosen from an interval determined implicitly by $b_n$. Moreover, adopting $\hat \Sigma(t)$ leads to more accurate type-I error control due to the faster convergence rate, see the discussion below \Cref{thm:lrv_diff}. 
         
In the following, we show the validity of $\hat \Sigma(t)$ via studying the  asymptotic behavior of $\hat { \Sigma}(t)$ under the fixed and local alternatives for the testing problem \eqref{WC-hypo1}, which is essential for the consistency of the aforementioned bootstrap tests.% in Algorithm 1 and Algorithm G.2 of \citep{bai2021}.
% We focus on $p\geq 2.$  
% Let $\sigma_{Hd}^2(t) = ( \Sigma_d(t))_{(1,1)}$.
% \WC{What is new in the following proposition?}\edit{Some changes in the proof. Just a reminder}
% \edit{

  \begin{assumption}\label{assumptionHp}
   Assumption \ref{E:HW} holds over $(-\infty,1]$, $H(t,\F_i)\ (i \in Z)$ is of $2\kappa$-order local stationarity on $(-\infty, 1]$, and its
% The zero-mean SRD process $(u_{i,n})_{i=-\infty}^{n}$  is of $2$-order locally stationary and  $4$-order short-range dependent on $(-\infty, 1]$. 
% \WC{locally stationary?}
%    \begin{enumerate}[label=(a\arabic*)'] 
   %    \item   $H(t,\F_0) \in \mathrm{Lip}_2(-\infty,1]$,  and 
   % $ \sup_{ t \in (-\infty,1] }\left\|H\left(t, \mathcal{F}_{0}\right)\right\|_{4}<\infty
   % $.
   % \item $ \delta_4(H, k,(-\infty,1]) = O(\chi^k)$ $\text{for some}~ \chi \in (0,1)$. 
   %    \item 
      long-run variance function 
     \begin{align}
      \sigma^{2}_H(t)=\sum_{k=-\infty}^{\infty} \operatorname{cov}\left\{H\left(t, \mathcal{F}_{0}\right), H\left(t, \mathcal{F}_{k}\right)\right\} \quad (t \in (-\infty, 1]),
      \label{eq:sigmaH_long} 
     \end{align}
     % Let $\sigma^2(0) = \underset{t \to 0}{\lim}\sigma^2(t)$. 
    %  \WC{I think we could delete (3.5), but please check referencing}
     % The long-run variance function 
     satisfies that ${\inf}_{t \in (-\infty,1] } \sigma_H^2(t)> 0$, ${\sup}_{t \in (-\infty,1] }\sigma_H^2(t) < \infty$, and  $\sigma_H^2 (\cdot)$ is twice continuous differentiable on $[0,1]$. \label{A:H_smooth_long} \label{A:H_long-run variance}\label{A:H_long}\label{A:H_delta_long}
%     \end{enumerate} 
  \end{assumption}
% \WC{LJ: in this case, we shall define the filter of $u$? Where is the exact rate?} 

\begin{theorem}\label{lm:Sigma_d}
Under Assumptions \ref{A:nonpar},  \ref{Ass-W}, \ref{Ass-E}, %\ref{B:H_delta} 
\ref{A:K} 
and \ref{assumptionHp}, assuming $m\tau_n^{3/2}/\log n \to \infty$, $\tau_n \to 0$, $n\tau_n^3 \to \infty$, $m / (n \tau_n^{3}) \to 0$, $\tau_n^{3-1/\kappa}\surd{m} \to 0$, $m = O(n^{1/3})$, $\kappa \geq \max\{4/(1/2-d),2/(3d),4\}$, it follows that under $H_A$
\begin{align}
\sup_{t \in I}\left| m^{-2d}\hat{{ \Sigma}}(t) -\kappa_2(d)\sigma_H^2(t)  \mu_W(t) \mu^{\T}_W(t)\right| = o_p(1),\nonumber
\end{align}
where $\kappa_2(d) = \Gamma^{-2}(d+1)\int_{0}^{\infty}\{t^d - (t-1)_+^d\}\{2t^d - (t-1)_+^d - (t+1)^d\} dt$.
% where $\kappa_2(d) = \Gamma^{-2}(d+1)\int_{0}^{\infty}(t^d - (t-1)^d)(2t^d - (t-1)^d - (t+1)^d )dt$.
\end{theorem}
% }
% \WC{Delete `similar to', rewrite this in the same way as the current version of the paragraph below proposition 7.1}\edit{done}   

% \LJ{In supplement rate; $m^{-2d}$.} \WC{demonstrate briefly why $m^{-2d}$ is important}
%\LJ{The following explain $m^{-2d}$}.
\cref{lm:Sigma_d} shows that the proposed difference-based estimator $\hat { \Sigma}(t)$ in \eqref{eq:diff_correct} inflates at the rate of $m^{2d}$ under long-range dependence with parameter $d > 0$, while its limit normalized by $m^{2d}$ depends on the $ {\mu}_W(t)$ along with the long-run variance  of  $e_{i,n}$ and the long-memory parameter $d$. The long-memory parameter $d$ also affects the theoretical properties of the long-run variance-covariance estimate through the moment condition of $\kappa$. The exact convergence rate is displayed in Step 6 of the proof of \cref{lm:Sigma_d} in supplement due to page limit. % When $d$ approaches $1/2$ and $0$, a higher moment condition is required for consistency. %while $d$ is close to $0$, the requirement of high moment condition of $\kappa$ is due to technical convenience. 
%But when $d = c/\log n = o(1)$, such moment condition can be relaxed.  
In the following \cref{lm:Sigma_dn}, we investigate the performance of the estimator $\hat { \Sigma}(t)$ under the local alternatives $d_n=c/ \log n$ for some constant $c>0$. For this purpose, we define the long-run {\it cross} covariance vector between the locally stationary processes $ U(t,\F_i)\ (i \in  Z)$ and $H(t,\F_j)\ (j \in  Z)$.
\begin{definition}\label{def:SUH}
 Define the long-run cross-covariance vector $s_{UH}(t) \in \mathbb R^p$ by
  \begin{align}
  s_{UH}(t) = \sum_{j=-\infty}^{\infty} \mathrm{Cov}\{U(t, \F_0), H(t, \F_j)\} \quad (t \in [0, 1]).\nonumber
  \end{align}
\end{definition}
%When $p=1$, $s_{UH}(t)$ degenerates into $\sigma_H^2(t)$.
For given constants $c>0$ and $\alpha_1 \in (0,1)$, define for $0 \leq t\leq 1$, the symmetric matrix $$\check { \Sigma}(t) =  { \Sigma}(t) + (e^{c\alpha_1}-1)^2\sigma_H^2(t)  \mu_W(t) {\mu}^{\T}_W(t) + (e^{c\alpha_1}-1)\{s_{UH}(t)  \mu_W^{\T}(t) + \mu_W(t)s^{\T}_{UH}(t)\}.$$ 
The following Assumption \ref{C:s} guarantees that $\check \Sigma(t)$ is smooth and non-degenerate. 
\begin{assumption}\label{C:s}
 $\check { \Sigma}(\cdot) \in C^2[0, 1]$, and $\lambda_{\min}\{\check { \Sigma}(t)\}$ is bounded above $0$ on $[0,1]$.
\end{assumption}
% When covariates and errors are independent, %$W(t, \F_i)$ and $H(t, \F_i)$ are independent, i.e., $\F_i$ contains the covariates and innovations contain$\varepsilon_i = (\zeta_i, \xi_i)^{\T}$, where $(\zeta_i)$ is independent of $(\xi_i)$ and they are  $i.i.d.$, and $W(t, \F_i) = W_0(t, (\cdots,\zeta_{i-1}, \zeta_i))$ and $H(t, \F_i) = H_0(t, (\cdots, \xi_{i-1}, \xi_{i}))$, where $W_0$ and $H_0$ are measurable functions,
%  \WC{say, Lujia please explain the independence using (a) in page 515 of \citep{zhou2010simultaneous}}\edit{done}, 
 %$s_{UH}(t) \mu_W^{\T}(t)$ reduces to $\sigma_H^2(t) \mu_W(t) \mu_W^{\T}(t)$, which is strictly positive since the first entry of $x_i$ is 1 and $\sigma_H^2(t)$ is non-degenerate on $[0,1]$. 
% Then under condition \ref{A:U_long-run variance} ($p=1$), \cref{C:s} is satisfied. 
% For the dependent case, 
Since $ \Sigma(t)$ and $(e^{c\alpha}-1)^2\sigma_H^2(t)  \mu_W(t) {\mu}^{\T}_W(t)$ are positive definite, by Weyl's inequality Assumption \ref{C:s} is satisfied for sufficiently small positive $c$.
%  \WC{What is new in the following proposition?}\edit{Some changes in the proof. Just a reminder}
%  \edit{
\begin{theorem}\label{lm:Sigma_dn}
  Let Assumptions \ref{A:nonpar}, \ref{Ass-W}, \ref{Ass-E}, %\ref{B:H_delta}, 
  \ref{A:K},
  \ref{assumptionHp} and \ref{C:s} be satisfied. If $ m \tau_n^{3/2} \to \infty$, $\tau_n \to 0$, $m / (n \tau_n^{3}) \to 0$, $\tau_n^{3-1/\kappa} \surd{m}\to 0$, $m = \lf n^{\alpha_1}\rf$, $\alpha_1 \in (0,1/3)$, we have
  %$m = O(n^{1/3})$,
  \begin{align}   
   \sup_{t \in I}\left|\hat{{ \Sigma}}(t) - \check { \Sigma}(t)\right|   = o_p(1).\nonumber
 \end{align}
%  where $ \check { \Sigma}(t) =  { \Sigma}(t) + (e^{c\alpha_1}-1)^2\sigma_H^2(t)  \mu_W(t) {\mu}^{\T}_W(t) + (e^{c\alpha_1}-1)s_{UH}(t)  \mu_W^{\T}(t) + (e^{c\alpha_1}-1) \mu_W(t)s^{\T}_{UH}(t).$
\end{theorem}
%  }
% \WC{Move the following comments to supplement and focus on $p\geq 2$ in the main!}  If $p=1$,  $ \mu_W(t) = 1$, $ \Sigma(t) = s_{UH}(t) =  \sigma_H^2(t)$, then $\check { \Sigma}(t)$ coincides with $e^{2c\alpha_1}\sigma_H^2(t)$, which differs from the limit of $\hat {\Sigma}(t)$ under the null hypothesis, i.e., $\sigma_H^2(t)$, by only a multiplicative factor $e^{2c\alpha_1}$. For $p\geq 2$, the difference between the limits of $\hat { \Sigma}(t)$ under the null hypothesis and local alternatives can be no longer captured by  a multiplicative factor.\par
  \cref{lm:Sigma_d} and \cref{lm:Sigma_dn} lead to the desired limiting distribution of the bootstrap statistics under the fixed and local alternatives achieving satisfactory power performance in finite samples, see \citep{bai2021} for theoretical justification and numerical evidence. In finite samples, we demonstrate that long memory tests with difference-based long-run covariance matrix estimates can achieve sizes closer to the nominal level and are more powerful than their counterparts using plug-in estimates, see \cref{sec:sim}.
 
\section{Simulation} \label{sec:sim}
\subsection{Setting}
\LJ{ We elaborate the procedure of tuning parameter selection, which is available in  the  R package \texttt{mlrv}, and display the values of the parameters selected in Appendix \ref{sec:sel}. }
%We investigate the finite sample power performance of the tests introduced in \cref{sec:app} and compare the results using difference-based estimator \eqref{eq:diff_correct} with the results equipped with existing long-run variance estimators.
 Let $(\varepsilon_{l})_{l\in Z}, (\zeta_{l})_{l\in Z}, (\eta_{l})_{l\in Z}$ be  $ N(0,1)$, $\vartheta_i = (\eta_i+ \varepsilon_i)/2$,  and consider the filtrations
\begin{align}
    \F_j =( \LJ{\ldots}, \zeta_{j-1}, \zeta_{j}), \quad \mathcal{G}_j = (\LJ{\ldots}, \varepsilon_{j-1}, \varepsilon_{j}),\quad \mathcal{H}_j = (\LJ{\ldots}, \varepsilon_{j}, \eta_{j})\quad (j=-\infty, \LJ{\ldots}, n). \quad \nonumber
\end{align}

\subsection{Testing for structural changes}
We generate the locally stationary process 
 $x_{i,n,1}$ from $G_1(t, \mathcal H_i) = \sum_{j=0}^{\infty} (0.5 - 0.5t)^j \vartheta_{i-j}$, the locally stationary process $x_{i,n,2}$ from $G_2(t, \mathcal H_i) = \sum_{j=0}^{\infty} \{0.25 +0.5(t-0.5)^2\}^j \epsilon_{i-j}$, and the locally stationary process $u_{i,n}$ from $G(t, \F_i) = 0.65 \cos(2\pi t) G(t, \F_{i-1}) + \zeta_i$.
We consider the following heteroscedastic linear regression model:
\begin{align}
    y_{i,n} = 1 + m_{i,n} + x_{i,n,1}+ x_{i,n,2} + e_{i,n},\quad e_{i,n} = (1 + 0.1 x_{i,n,1})u_{i,n}, \quad (i=1,\ldots, n),
\end{align}
where the function $m_{i,n}\ (i=1,\ldots,n)$ includes the following scenarios \par
% \WC{CP4 typo; introduce of indicator function}: \par
CP1: $m_{i,n} = 2\delta \sin(2 \pi t_i)x_{i,n,1} 1(0.5 \leq t_i \leq 1)$.\par
   CP2: $m_{i,n}=\delta \sin(2 \pi t_i) 1(0 \leq t_i \leq 0.4) + \delta x_{i,n,1}  1(0.7 \leq t_i \leq 1)/2$. \par
   CP4: $m_{i,n} = 1.5\delta \sin(2 \pi t_i) 1(0 \leq t_i \leq 0.2~\text{or}~0.4 \leq t_i \leq 0.6~\text{or}~0.8 \leq t_i \leq 1)$.\par
We conduct our simulation with a sample size $300$.
As shown in \cref{fig:cptest}, when there are $4$ change points, the block bootstrap test based on ordinary least squares residuals (\citep{wu2018gradient}) suffers from low and non-monotonic power that can not approach $1$. By contrast, the newly proposed difference-based long-run covariance matrix estimator enhances the simulated power significantly and addresses the non-monotonic power issue by taking the difference. To further illustrate the impact of long-run covariance matrix estimators in bootstrap tests of structural breaks, we investigate the estimation accuracy of both long-run covariance matrix estimators and find that our proposed estimator halves the empirical mean square error in the presence of change points,  see \cref{sec:lrv} of the supplement for extra simulation results \LJ{and sensitivity analysis}. \par
% \LJ{Maybe postpone to the supplement. Estimation accuracy of lrv matrix} \WC{Mention this in the main}
% 
\begin{figure}
    \centering
    \includegraphics[width = 0.7\linewidth]{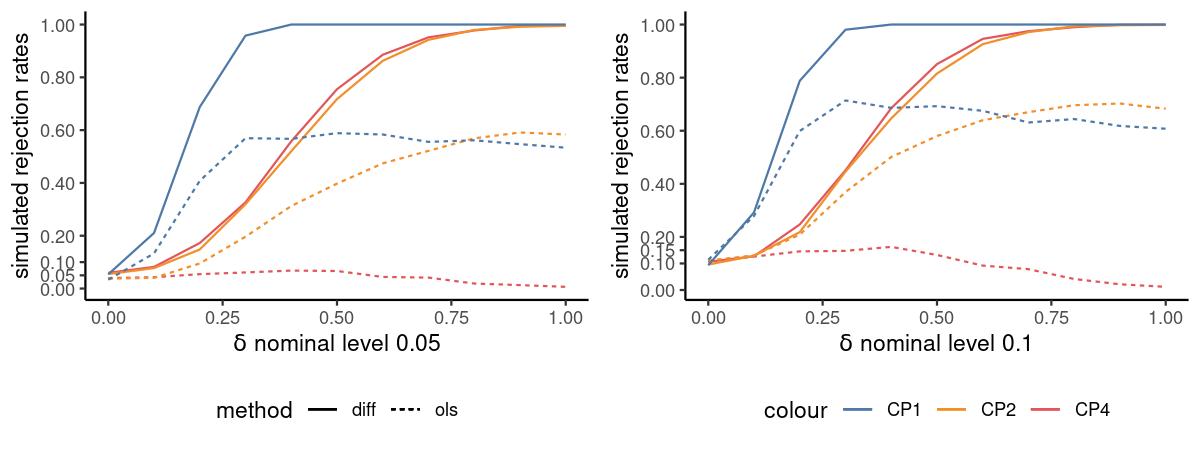}
 \caption{The empirical rejection rates of gradient-based structural change point tests as $\delta$ increases from $0$ to $1$ with sample size $n = 300$ and simulation times $2000$, under three scenarios CP1(blue), CP2(orange), CP4(red), using blocks of ordinary least squares residuals (small-dashes), and difference-based long-run covariance matrix estimator (solid).  Left panel: nominal size 0.05; Right panel: nominal size 0.1.}
 % Dotted curves denote simulated rejection rates of bootstrap tests using blocks of ordinary least squares residuals,  and solid curves denote simulated rejection rates of those equipped with the difference-based long-run covariance matrix estimator. 
    % \caption{Power performance of gradient-based structural change point tests as $\delta$ increases from $0$ to $1$ with sample size $n = 300$ and simulation times $2000$, under three scenarios cp1(circle), cp2(triangle), cp4(square). Dotted curves denote simulated rejection rates of bootstrap tests using blocks of ordinary least squares residuals,  and solid curves denote simulated rejection rates of those equipped difference-based long-run variance estimators.  Left panel: nominal size 0.05; Right panel: nominal size 0.1}
    \label{fig:cptest}
\end{figure}
\subsection{Testing for long-range dependence}
Consider the following heteroscedastic functional linear model, 
% \WC{please rewrite it using the new formulation, e.g., $(1-B)^d e_i$}
\begin{align}
    y_{i, n}=\beta_{1}(t_i)+\beta_{2}(t_i) x_{i, n}+(1-\B)^{-d}e_{i,n}, \quad (i=1, \ldots, n),
\label{M1}
\end{align}
 where $\B$ is the lag operator, $\beta_{1}(t)=4 \sin (\pi t) $, $\beta_{2}(t)=4 \exp \LJ{\{}-2 \left(t-0.5\right)^{2}\LJ{\}}$, $x_{i, n}=W(t_i, \mathcal{F}_{i})\ (i=1, \ldots, n)$, and $ e_{j,n} = H(t_j,\mathcal{F}_j, \mathcal{G}_j)\ (j=1, \ldots, n)$, where $$
    H(t,\mathcal{F}_{i},\mathcal{G}_{i}) =  B\left(t,\mathcal{G}_{i}\right)\{1+W^2(t, \mathcal{F}_{i})\}^{1/2}\quad (i \in Z;\ t \in [0,1]),$$  $
     W\left(t, \mathcal{F}_{i}\right)= \{0.1 + 0.1\cos(2\pi t)\}W(t, \mathcal{F}_{i-1})+ 0.2\zeta_{i} + 0.7(t-0.5)^2,
$ and $B(t, \mathcal{G}_{i})=\{0.3 - 0.4(t-0.5)^2\} B(t, \mathcal{G}_{i-1})+ 0.8\varepsilon_{i}$. 
 % \cref{fig:lrd} compares the empirical rejection rates of the KPSS, V/S, R/S, and K/S tests using difference-based and plug-in long-run covariance matrix estimators for model \eqref{M1} of sample size $1500$ with nominal level $0.05$ and $0.1$  with long memory parameter $d$ increasing from $0$ to $0.5$. 
 As demonstrated by  \cref{fig:lrd}, the difference-based long-run covariance matrix estimator yields uniform improvement for the power of KPSS, V/S, R/S, and K/S tests against $0< d \leq 1/2$ in finite samples. Notably, equipped with the difference-based estimator the simulated power of K/S, R/S, and KPSS tests can reach $1$ with the sample size $1500$ as $d$ increases to $0.5$, while using the plug-in long-run covariance matrix estimator the power is much lower and stays far below $1$, except for the V/S test. \LJ{The corresponding sensitive analysis is in \cref{sec:lrv}.}
\begin{figure}
\centering
     \includegraphics[width = 0.7\linewidth]{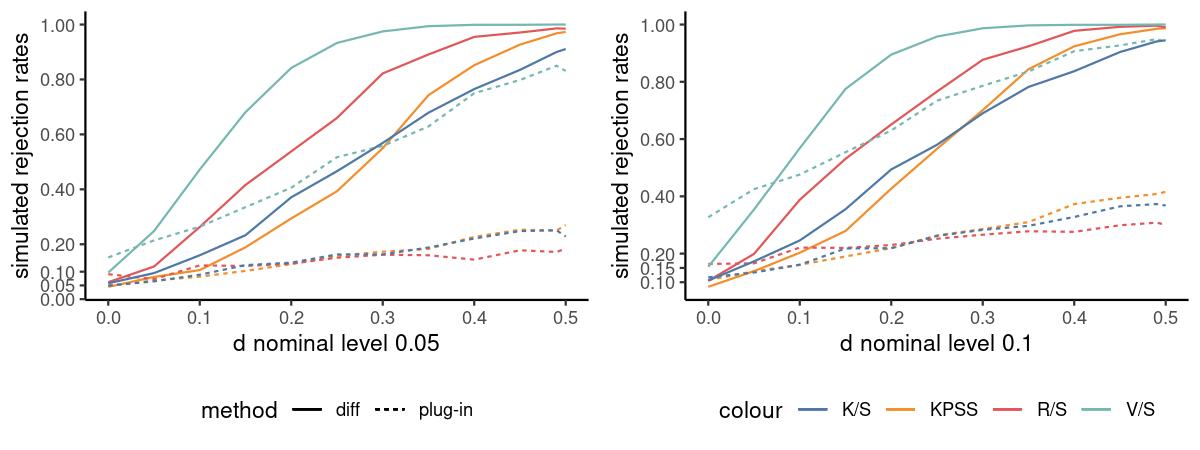}
\caption{Empirical rejection rates of KPSS(orange), K/S(blue), R/S(red), and  V/S(green) tests under different $d$'s with sample size $1500$, using the plug-in method (small-dashes) and the difference-based method (solid). Left panel: nominal size 0.05; Right panel: nominal size 0.1.}
 % Dotted curves denote simulated rejection rates of tests using the plug-in method, while solid curves denote simulated rejection rates of tests using the difference-based method. Left panel: nominal size 0.05; Right panel: nominal size 0.1.
    % \caption{Power performance of KPSS(circle), K/S(triangle), R/S(square), and  V/S(dagger) tests with sample size $1000$. Dotted curves denote simulated rejection rates of tests using plug-in method, while solid curves denote simulated rejection rates of tests using the difference-based method. Left panel: nominal size 0.05; Right panel: nominal size 0.1.}
    \label{fig:lrd}
\end{figure}

% \LJ{Maybe add some sensitivity analysis}
%% use bibfile 
% \appendix
\section{Data analysis}\label{sec:data}
We apply our newly proposed long-run covariance matrix estimator \eqref{eq:diff_correct} to the analysis of Hong Kong hospital data, including structural change detection and tests for long memory. The data set consists of daily hospital admissions in Hong Kong as well as daily measurements of pollutants between January 1, 1994, and December 31, 1995. The sample size is $730$. Consider the functional linear model for this data set, i.e., 
\begin{equation}\label{eq:hk_model}
   y_{i,n}=\beta_{1}(t_i)+\sum_{p=2}^{4} \beta_{p}(t_i) x_{i, p, n}+\varepsilon_{i, n}, \quad (i=1, \ldots, n),
 \end{equation}
 where $(y_{i,n})$ is the series of daily total number of hospital admissions of circulation and respiration and $(x_{i, p, n})$, $p = 2, 3, 4$, are the series of daily levels of $\text{SO}_2$, $\text{NO}_2$ and dust, respectively, in micrograms per cubic meter.  %\par We consider the test for the long memory of $\varepsilon_i$ of the whole time span between January 1, 1994, and December 31, 1995 with sample size $730$.

As illustrated in Section 5 of \citep{wu2018gradient}, it is of practical concern to test whether $\beta(\cdot)=(\beta_p(\cdot),1\leq p\leq 4)^\T$ is a constant vector. The test for structural changes equipped with the difference-based estimator yields $p$-value $0.006$, which rejects the null hypothesis of no structural change, while the test procedure proposed in  \citep{wu2018gradient} based on ordinary least squares residuals yields $p$-value greater than $0.1$. The different testing results can be attributed to the power loss of \citep{wu2018gradient} under structural change, as shown in \cref{fig:cptest}.

% \LJ{Compare \citep{zhou2013heteroscedasticity}}
 We then consider the test for the long memory of $\varepsilon_{i,n}$. 
\citep{bai2021} performs long-memory tests on each covariate process and concludes that they are short-range dependent. Therefore, we could apply the tests introduced in \cref{sec:app2} to this data, and compare the $p$-values 
 of the tests equipped with the difference-based estimator \eqref{eq:diff_correct} and with the plug-in estimator of \citep{zhou2010simultaneous}, respectively.  The bandwidth $b_n$ in \cref{sec:app2} are selected by the GCV  method advocated by \citep{zhou2010simultaneous} and \citep{bai2021}.
 \begin{table}[ht]
\centering
\begin{tabular}{rrrrr|rrrrr}
  \hline
 method& KPSS & R/S & V/S & K/S &method & KPSS & R/S & V/S & K/S \\ 
  \hline
plug & 0.300 & 0.171 & 0.079 & 0.356 &
  diff & 0.810 & 0.888 & 0.835 & 0.907 \\ 
   \hline
\end{tabular}
\caption{$p$-values of tests for long memory}
\label{tb:pvalue}
\end{table}
 As in \cref{tb:pvalue}, 
the $p$-values of four types of tests for long memory based on plug-in estimates are much smaller than those based on  $\hat\Sigma(t)$. The smaller $p$-values might result from the inaccurate size performance associated with the plug-in estimator, which  is evidenced by extra simulation results in \cref{tb:size} of the supplemental material showing that the methods with plug-in estimates tend to over-reject and result in smaller  $p$-values.  A further sensitivity check shows that when using $1.2\times$GCV bandwidths, the R/S test with the plug-in estimator yields $p$-value $0.08$ rejecting the null hypothesis at the significance level of $10\%$, while the $p$-values of tests equipped with the difference-based estimator remain large leading to the same decision of accepting the null for all four tests.
% \WC{make comments on the stability of our method} %The difference are due to the inaccurate size performance of the tests  as shown  in \cref{tb:size} of the simulation.  

\section{Conclusion}
% In this paper, we propose a debiased difference-based long-run covariance matrix estimator for the functional linear model  under complex dynamics, \LJ{which is implemented in the R package \texttt{mlrv}.} We establish the uniform consistency of the proposed estimator and justify theoretically its application to structural change tests and detection of long memory. We conduct extensive simulation studies and achieve significant improvements in the finite samples. In particular, the new difference-based statistic is robust to both smooth and abrupt changes in the regression coefficients, largely eliminating the non-monotonic issue arising in the change point tests. 
\LJ{Additional potential applications can be found in \cref{sec:appl} of the supplement.} 
% Finally, we apply our new estimators to the nonparametric analysis of Hong Kong circulatory and respiratory data, providing new uncertainty on structural changes and long memory. 
The optimal long-run variance for time series with stationary errors has been thoroughly discussed recently by \citep{chan2021optimal}. However, the approach therein is not applicable when non-stationarity is present. We leave the optimal estimation of the long-run covariance matrix under time series non-stationarity as rewarding future work. In addition, the generalization of our method beyond linear models will also be of great importance. 
% \WC{R package}
%\WC{Furthermore, the reuslts yielded by the plugin methods are more  sensitive to the tuning parameter $b_n$ } %The optimality of difference statistics \citep{chan2021optimal} for the random-design locally stationary time series regression model is left for rewarding future work.

\section*{Acknowledgement}
Weichi Wu is the corresponding author and is supported by National Natural Science Foundation of China 12271287. The authors thank the editors, associate editors and referees for constructive comments.
\section*{Supplementary material}
In the supplement, we present implementation details, extra simulation studies, the proofs of the findings in this paper as well as auxiliary technical results. 
% \LJ{The package `mlrv' for the implementation of the tests and estimators is online available.}
\bigskip
\newpage
\begin{center}
   {\Large \bf Supplement to ``Difference-based covariance matrix estimate in time series
nonparametric regression with applications to specification
tests"}
\end{center}

\setcounter{section}{0}

We organize the supplementary material as follows: \cref{sec:bias} gives the intuition of the bias in the difference-based estimator. The implementation details including the procedure of selection of tuning parameters are in \cref{sec:sel}. \cref{sec:high} provides some discussion on the assumptions. \cref{sec:appl} offers other applications of the proposed difference-based estimator. We investigate the sensitivity of finite-sample performance of the tests with respect to the smoothing parameters, extra simulation with smaller sample size, and the performance of different estimates in the presence of change points in \cref{sec:lrv}. \cref{diff} presents proofs of the results in the main paper. \cref{sec:aux} provides auxiliary results which are used in the proofs.     \par 

% Let "$\Rightarrow$" denote weak convergence, and "$\leadsto$" denote the convergence of a process. 
% Let $0 \times \infty = 0$, $a_n \sim b_n$ denote $\lim_{n \to \infty} a_n/b_n = 1$ for real sequences $a_n$ and $b_n$. For a random variable $X$ and a distribution $G$, $X \sim G$ is denoted by $X$ follows the distribution $G$.
% Let $D[0,1]$ be the space of real functions on $[0,1]$ that are right-continuous and have left-hand limits (also named càdlàg functions).
\appendix
\section{Bias in the difference-based estimator}\label{sec:bias}
\WC{To see this, consider the simple case where $x_i$ and $e_i$ are independent and $i.i.d.$ random variables.  Then, for the differenced series we have 
\begin{align}
    &\mathrm{E} \{(e_i + \beta_0(i/n) + \beta(i/n) x_i - e_{i-1} - \beta_0((i-1)/n)- \beta((i-1)/n) x_{i-1})^2\}\\ 
    & = \mathrm{E}\{(e_i - e_{i-1})^2\} +\mathrm{E}\{(\beta_0(i/n) - \beta_0((i-1)/n))^2\}+  \mathrm{E}\{(\beta(i/n) x_i - \beta((i-1)/n) x_{i-1})^2\},
\end{align}
 where the third term is close to $2\beta(i/n)\mathrm{var}(x_i)$. Notice that this term is $O(n^{-2})$ if $x_i$ is deterministic and smooth instead.
In addition, we examine the bias via empirical studies. We consider the following dependent and independent settings and compare the differenced data simulated from models  with the stochastic trend and with only the deterministic smooth trend for a simple illustration:}\par
\WC{A. Independent scenario. }\par
\WC{A.1 Stochastic trend
\begin{align}
    y_i =  4(i/n-0.5)^2 + 0.5x_{1,i}+ 0.4 x_{2,i} + e_i\quad (i= 1,\ldots,n),
\end{align}
}
\WC{where $x_{1,i}$ are independent and identically distributed $N(2,1)$ random variables, $x_{2,i}$ and $e_i$ are independent and identically distributed standard Gaussian variables. Note that this setting is allowed by our Assumptions 1-5.}   \par
\WC{A.2 Deterministic smooth trend
\begin{align}
    y_i = 4(i/n-0.5)^2 + 1 +  e_i\quad (i= 1,\ldots,n),
\end{align}
where $e_i$ are independent and identically distributed standard Gaussian variables.}
\begin{figure}[!ht]
    \centering
        \includegraphics[width = 0.7\linewidth]{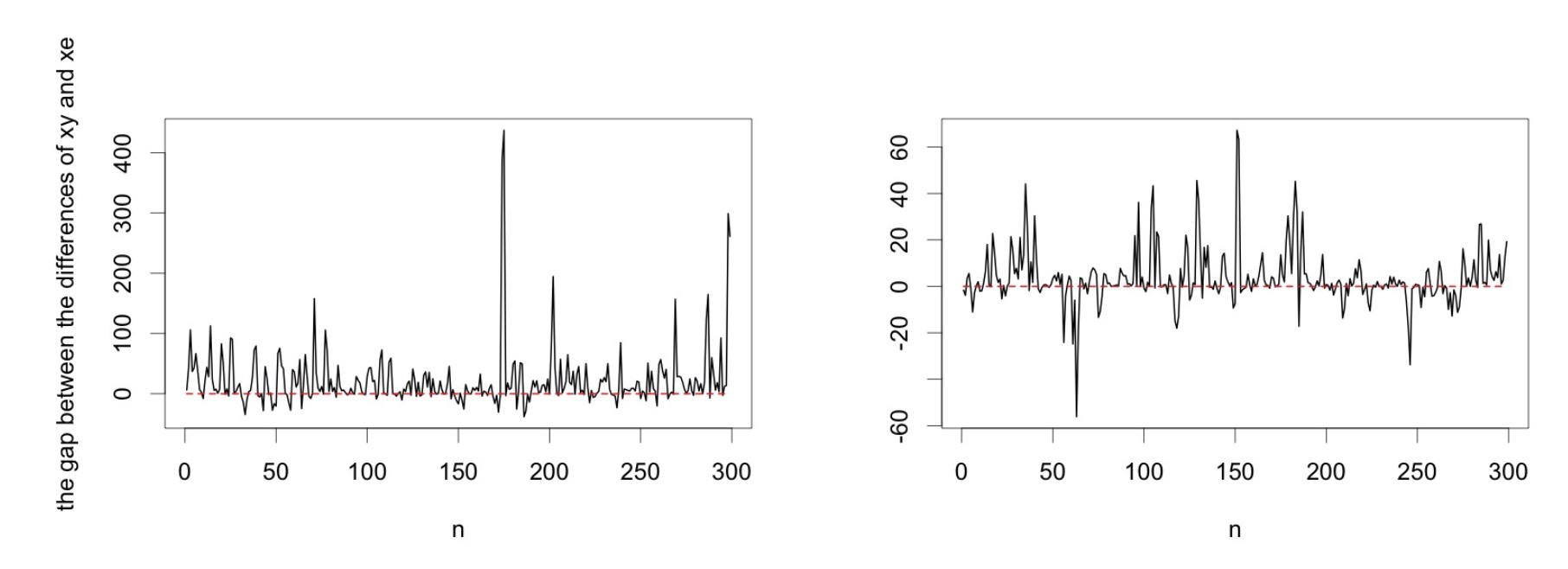}
%  \includegraphics[width=0.45\linewidth]{diff.jpeg}
%  % \figuresize{.15}
% \includegraphics[width=0.45\linewidth]{diff_ar.jpeg}
    \caption{\WC{The sample path of $(x_{i+1}y_{i+1} - x_{i}y_{i})^2 - (x_{i+1}e_{i+1} - x_{i}e_{i})^2$ with stochastic trend (solid lines), and with only deterministic smooth trend (dotted lines). The left panel shows the gaps between A.1 and A.2 and the right panel shows that between B.1 and B.2, respectively. }}
    \label{fig:diffiid}
\end{figure}
\par
\WC{B. Dependent scenario} \par
\WC{B.1 Stochastic trend
\begin{align}
    y_i =  4(i/n-0.5)^2 + \sin (\pi i/n)x_{1,i}+ x_{2,i} + e_i \quad (i= 1,\ldots,n),\label{eq:dep_sto}
\end{align}
where $x_{1,i}$ is an Autoregressive Moving Average process with autoregressive coefficients 0.8897 and -0.4858, moving average coefficients -0.2279 and 0.2488, i.e., 
\begin{align}
    x_{1,i} - 0.8897 x_{1,i-1} + 0.4858  x_{1,i-2} = z_i -0.2279 z_{i-1} + 0.2488 z_{i-2},
\end{align}
where $(z_i)$ are independent and identically distributed standard Gaussian variables, 
$(x_{2,i})$ are independent and identically distributed standard Gaussian variables, and $(e_i)$ is an autoregressive process with coefficient 0.5, i.e., 
\begin{align}
    e_i = 0.5 e_{i-1} + \eta_i, 
\end{align}
with $\eta_i$ being independent and identically distributed standard Gaussian variables. Note that this setting is allowed by our Assumption 1-5. }  
% \WC{Please check: is mean $0$ allowed? Also we emphasize for $p\geq 2$, which is also awared by the reviewer. Therefore it is better to have a $p=2$ case (maybe an intercept?. And maybe some time-varying component?) We should discuss the new model very carefully to avoid possible model identification issue.}\LJ{We don't need to fit a model.}
\par
\WC{B.2. Deterministic smooth trend}\par
\WC{\begin{align}
    y_i = 4(i/n-0.5)^2 + e_i \quad (i= 1,\ldots,n),
\end{align}
where $e_i$ is as defined in \eqref{eq:dep_sto}.}

\WC{\cref{fig:diffiid} displays the sample path of the gap between the differences between $(x_{i+1}y_{i+1} - x_{i}y_{i})^2$ and $(x_{i+1}e_{i+1} - x_{i}e_{i})^2$ under the two scenarios described above. As shown by \cref{fig:diffiid}, the paths of the stochastic trend model are much more jagged than those of the deterministic smooth trend model even in the independent setting, since the deterministic smooth trend is almost eliminated by differencing.  
 Similar gaps can also be observed between the differences of $(x_{i+k}y_{i+k} - x_{i}y_{i})^2$ and $(x_{i+k}e_{i+k} - x_{i}e_{i})^2$ when $k>1$. These illustrate the influence of stochastic covariates on the difference-based statistics in approximating the difference of true errors weighted by covariates, 
which is in fact the source of the non-negligible bias.}

\section{Implementation details}\label{sec:sel}
\subsection{Selection of tuning parameters}
% The choices of $m \asymp n^{4/15}$ and $\tau_n \asymp n^{-2/15}$ can achieve rate optimal for sufficiently large $\kappa$, 
% while 
\WC{For refinement, we recommend the following extended minimum volatility method as proposed in Chapter 9 of \citep{politis1999subsampling} which works quite well in our empirical studies. The extended minimum volatility method has the advantage of robustness under complex dependence structures and does not depend on any parametric assumptions of the time series. To be concrete, we first propose a grid of possible block %\WC{check my rewriting till the end of this answer} 
 sizes and bandwidths $\{m_1, m_2,\cdots, m_{M_1}\}$, $\{\tau_1, \tau_2,\cdots, \tau_{M_2}\}$ from $[\lfloor c_1 n^{4/15} \rfloor, \lfloor c_2 n^{4/15} \rfloor]$ and  $[c_3 n^{-2/15}, c_4 n^{-2/15}]$, respectively, where $c_1, \cdots, c_4$ are constants set as default in the package.
Define $s^2_{m_i,\tau_j}$ as the sample variance of the bootstrap statistics, say $\tilde T_{n, (1)},\ldots, \tilde T_{n,(100)}$ calculated from 100 bootstrap runs 
with parameters $m_i$ and $\tau_j$. The formula of the bootstrap statistics is
determined by the tests. For example, in the structural stability test, we use 
% \WC{Lj why you have $t$ in $s^2_{m_i,\tau_j}(t)$ but there is no $t$ in $\tilde T_{n,(r)}$}\LJ{should be only $s^2_{m_i,\tau_j}$}
\begin{align}
 \tilde T_{n,(r)} =  \max_{m \leq i \leq n-m+1}|\Psi_{i,m}^{(r)} - \hat{ \Lambda}(i/n) \hat{ \Lambda}^{-1}(1) \Psi_{n-m+1,m}^{(r)}|,
\end{align}
where
% \WC{Change all $\hat \Lambda_i$ to $\hat \Lambda(i/n)$. What is $l_n$}
$
    \Psi_{i,m}^{(r)} = n^{-1/2}\sum_{j=1}^i \hat { \Sigma}^{1/2}(t_j) R_j^{(r)}$, $\hat{ \Lambda}(i/n) = \sum_{j=1}^i x_{j,n} x_{j,n}^{\T}/n,
$
  $( R_j^{(r)})_{j=1}^n$ are the independent and identically distributed standard normal random variables independent of data and are independently generated  in the $r$th bootstrap iteration, and $ \hat \Sigma(t) $ is an estimator of $ \Sigma(t)$.  For testing long memory, we can use $
         \widetilde{ \mathrm{K}}^{(r)}_{n}$, $ \widetilde{\mathrm{RS}}^{(r)}_{n} $, $ \widetilde{\mathrm{VS}}^{(r)}_{n}$ and  $ \widetilde{\mathrm{KS}}^{(r)}_{n}$ in Section 5.2 for $\tilde T_{n,(r)}$   to choose smoothing parameters for different tests. 
Then we calculate  
   \begin{align}
    \mathrm{MV}(i,j):=  \mathrm{SE} \lt\{\cup_{r_1=-1}^{1}\{s^2_{m_{i}, \tau_{j+r_1}}\} \cup \cup_{r_2=-1}^{1}\{s^2_{m_{i+r_2}, \tau_{j}}\}\rt\} \label{eq:mvbootstrap1},
    \end{align}
    where SE stands for standard error.
    Finally, we select the pair $(m_{i^*},\tau_{j^*})$ where $(i^*,j^*)$ minimizes $\mathrm{MV}(i,j)$. The extended minimum volatility selection criterion \eqref{eq:mvbootstrap1} is similar in spirit to the classical one  %one which measures the standard deviation of the estimated long-run variance, 
    except that \eqref{eq:mvbootstrap1} is built on the bootstrap test statistics instead of using solely long-run covariance estimators. %the standard deviation of the bootstrap test statistics instead of the estimated long-run variance, 
    Therefore,  
    the extended minimum volatility selection criterion \eqref{eq:mvbootstrap1} is adaptive to various types of hypothesis testing problems. %, such as the detection of long memory. 
    \checkit{
    In our simulation studies, we recommend  $c_1 = 3/7$, $c_2 = 11/7$, 
$M_1$ is the number of the points and  $\tau_n$ from $(2/3) n^{-2/15}$ to $ n^{-2/15}$ with grid 0.05, i.e., $c_3 = 2/3$, $c_4 = 1$, and $M_2 = \lfloor  n^{-2/15}/0.15 \rfloor$ is the number of grid points of $\tau_n$. The choices of the constants $c_1,\ldots, c_4$ and $M_1$ and $M_2$ are constants  replying on the dependence and smoothness of the time series.  In practical, one can also choose the constants according to the prior knowledge.}
    The selection procedures of the tuning parameters
    considered in our paper are implemented in our package, while the package also supports user-specific choices of $m$ and $\tau_n$. \checkit{In practice, we recommend choosing $m$ from 
\begin{align}
    \max(\lfloor (3/7) n^{4/15} \rfloor -1, 1), \max(\lfloor (3/7) n^{4/15} \rfloor -1, 1)+1,
    \ldots, \\\max(\lfloor (11/7) n^{4/15} \rfloor+1, \max(\lfloor (3/7) n^{4/15} \rfloor -1, 1) + 2)),\label{eq:practicalm}
\end{align} to make sure there are enough neighborhood points for extended minimum volatility selection at the rate between  $\lfloor (3/7) n^{4/15} \rfloor$ and $\lfloor (11/7) n^{4/15} \rfloor$}.The terms $\max$, +1, -1, and +2 make the grid appropriate when the sample size is small.}\par
\checkit{The full algorithm including data-driven choices of $m$ and $\tau_n$ is as follows
\begin{enumerate}
    \item First propose a grid of possible block sizes and bandwidths $\{m_1, m_2,\cdots, m_{M_1}\}$, $\{\tau_1, \tau_2,\cdots, \tau_{M_2}\}$ for $m$ and $\tau$, say the grid for $m$ is $\max(\lfloor (3/7) n^{4/15} \rfloor -1, 1), \max(\lfloor (3/7) n^{4/15} \rfloor -1, 1)+1, 
    \ldots, \max(\lfloor (11/7) n^{4/15} \rfloor+1, \max(\lfloor (3/7) n^{4/15} \rfloor -1, 1) + 2))$ and  the grid for $\tau_n$ is $(2/3) n^{-2/15}, (2/3) n^{-2/15} + 0.05, \ldots, n^{-2/15} $. 
\item Compute $s^2_{m_i,\tau_j}$,  the sample variance of the bootstrap statistics, say $\tilde T_{n, (1)},\ldots, \tilde T_{n,(100)}$ calculated from 100 bootstrap runs 
with parameters $m_i$ and $\tau_j$. For example, in the structural stability test, we use 
% \WC{Lj why you have $t$ in $s^2_{m_i,\tau_j}(t)$ but there is no $t$ in $\tilde T_{n,(r)}$}\LJ{should be only $s^2_{m_i,\tau_j}$}
\begin{align}
 \tilde T_{n,(r)} =  \max_{m \leq i \leq n-m+1}|\Psi_{i,m}^{(r)} - \hat{ \Lambda}(i/n) \hat{ \Lambda}^{-1}(1) \Psi_{n-m+1,m}^{(r)}|,
\end{align}
where
% \WC{Change all $\hat \Lambda_i$ to $\hat \Lambda(i/n)$. What is $l_n$}
$
    \Psi_{i,m}^{(r)} = n^{-1/2}\sum_{j=1}^i \hat { \Sigma}^{1/2}(t_j) R_j^{(r)}$, $\hat{ \Lambda}(i/n) = \sum_{j=1}^i x_{j,n} x_{j,n}^{\T}/n,
$
  $( R_j^{(r)})_{j=1}^n$ are the independent and identically distributed standard normal random variables independent of data and are independently generated  in the $r$th bootstrap iteration, and $ \hat \Sigma(t) $ is an estimator of $ \Sigma(t)$. 
\item Calculate  
   \begin{align}
    \mathrm{MV}(i,j):=  \mathrm{SE} \lt\{\cup_{r_1=-1}^{1}\{s^2_{m_{i}, \tau_{j+r_1}}\} \cup \cup_{r_2=-1}^{1}\{s^2_{m_{i+r_2}, \tau_{j}}\}\rt\},
    \end{align}
    where SE stands for standard error.
\item Select the pair $(m_{i^*},\tau_{j^*})$ where $(i^*,j^*)$ minimizes $\mathrm{MV}(i,j)$
\item For $t \in [m/n, 1-m/n]$, compute the estimator using $ m_{i^*}$ for $m$, and $\tau_{j^*}$ for $\tau_n$, 
\begin{align}
         \hat {  \Sigma}(t) = \acute {  \Sigma}(t) - \breve {  \Sigma}(t), \quad \breve {  \Sigma}(t)= \sum_{j=m}^{n-m} \frac{m\hat{  A}_{j,m} \hat{  A}_{j,m}^{\T}}{2}\omega(t, j),
     \end{align}
where 
\begin{align}\hat{  A}_{j, m} = \frac{1}{m}\sum_{i= j-m+1}^{j} \{  x_{i,n}   x_{i,n}^{\T}\breve {  \beta}(t_i)-  x_{i+m, n}   x_{i+m, n}^{\T}\breve {  \beta}(t_{i+m})\},~ \breve {  \beta}(t) =   \Omega^{-1}(t)  \varpi (t),\\
\Omega(t) = \sum_{j=m}^{n-m} \acute {  \Delta}_{j} \omega(t, j)/2,\varpi (t) = \sum_{j=m}^{n-m} \breve {  \Delta}_{j} \tilde \omega(t, j)/2,\\
\tilde \omega(t, i)= K( (t_i-t)/\tau_n^{3/2} ) / [\sum_{i=1}^{n} K \{(t_i-t)/\tau_n^{3/2} \}], \omega(t, i)= K\{ (t_i-t)/\tau_n\} / [\sum_{i=1}^{n} K \{(t_i-t)/\tau_n\}].
\end{align} 
\item For $t \in [0,m/n)$, set $\hat \Sigma(t) = \hat \Sigma(m/n)$. For $t \in (1-m/n, 1]$, set $\hat \Sigma(t) = \hat \Sigma(1-m/n)$. 
\end{enumerate}}
\WC{Using the default choices of the tuning parameters, the summary of selected $m$'s and $\tau_n$'s in the time series regression setting with $d = 0$ for the four types of long memory tests and the regression model with no change points for the structural ability tests are displayed in Table \ref{tb:pars} and Figure \ref{fig:parmbox}, which partly demonstrates that our proposed estimator is not sensitive to the choices of the smoothing parameters.
From the results, we find that our tuning parameter selection approach recommends different but similar tuning parameters for those tests, especially for the long memory tests V/S, R/S, KPSS, and K/S tests which share the same statistical model and the same null hypothesis, implying the stability of our selection procedure.
Together with the simulation studies on the simulated rejection rates, the results indicate that our selection procedure works reasonably well.}\par
\begin{table}[ht]
\centering
\begin{tabular}{rrrrr}
  \hline
  &\multicolumn{2}{c}{$m$}&\multicolumn{2}{c}{$\tau_n$} \\
   \hline
 & Median & Max & Mean & Max \\ 
  % \hline
VS & 8 & 10 & 0.328 & 0.352 \\
  RS & 8 & 10 & 0.329 & 0.352 \\
  KPSS & 7 & 10 & 0.327 & 0.352 \\
  KS & 7 & 10 & 0.328 & 0.352 \\
  CP & 7 & 9 & 0.399 & 0.424 \\
   \hline
\end{tabular}
\caption{
\WC{Selected values of $m$ and $\tau_n$ in the long-memory tests (V/S, R/S, KPSS, K/S) and the test for structural stability (CP).}
}
% \caption{Selected values of $m$ and $\tau_n$.}
\label{tb:pars}
\end{table}
%\LJ{WEichi: quantile/histogram of the parameters}
\begin{figure}
\centering
    \includegraphics[width = 0.7\linewidth]{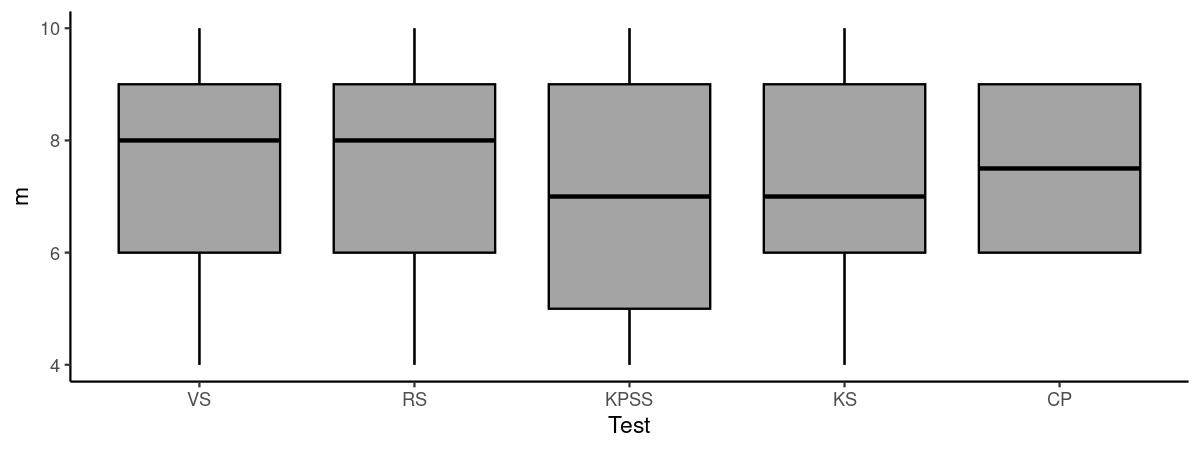}
    \caption{A boxplot of the selected smoothing parameter $m$ of V/S, R/S, KPSS, K/S tests and the test of structural stability.}
    \label{fig:parmbox}
\end{figure}
\subsection{The use of kernels}
\WC{We investigate a group of common kernels that satisfy Assumption 6, including quartic ($15/16(1-u^2)^2$), triweight ($35/32(1-u^2)^3$) with bounded support $|u| \leq 1$, and other kernels which are  differentiable but are not continuously differentiable at some points in $(-1,1)$ 
are  continuously differentiable almost everywhere  in $(-1,1)$ except a few points.}
\WC{ See \cref{fig:kernel} for our simulation results which check the performance of change points detection using our proposed long-run covariance matrix estimator with different kernels. The simulation result shows that the performance is reasonably well using different kernels satisfying Assumption 6 and some kernels  partially fullfill Assumption 6.} \par
 \WC{In the R package \texttt{mlrv}, we also offer the options of employing different kernels in the estimation of long-run covariance matrix, including triangular kernel, Epanechnikov kernel, quartic kernel, triweight kernel and tricube kernel. }
  % \LJ{may add more lines}
    \begin{figure}
        \centering
\includegraphics[width = 0.7\linewidth]{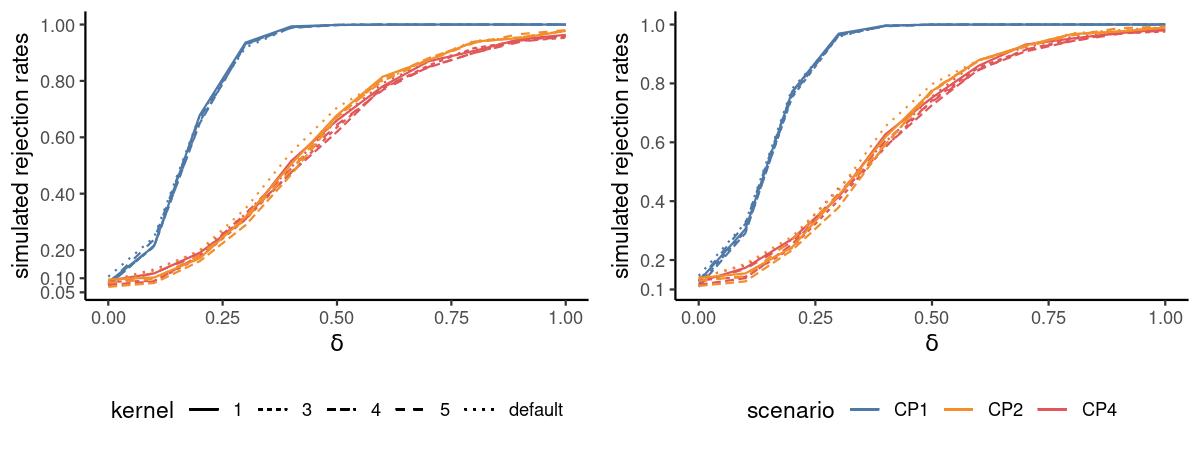}
        \caption{\WC{Empirical rejection rates of the structural stability test with triangular(solid), Epanechnikov (dotted), quartic(small dashes), triweight (dashes)  and tricube (long dashes) kernels, respectively.}}
        \label{fig:kernel}
    \end{figure}
\subsection{Positive definite estimators}
\WC{ In practice when the estimated covariance matrix, which proves to be consistent, is not positive definite, there are two possibilities: collinearity or the small sample size. The former implies that the estimand, namely the covariance matrix, is not positive definite, which is excluded by Assumption 2. In the latter case, the covariance matrix is positive definite, but the estimated covariance matrix can be non-positive definite due to the stochastic variation at a small sample size.} \par 
\WC{ Therefore, the remedy of the non-positive definite estimator when the sample size is small  
is of primary and practical concern. As discussed in the literature, one can use threshold method (\citep{politis2011higher}) or penalization (\citep{rothman2012positive}) for the remedy. %to handle the presence of indefiniteness in finite samples.
 We only discuss the threshold approach employed by \citep{politis2011higher}, \citep{dette2020prediction} among others due to the page limit, and leave further exploration of the modification as a rewarding future work.}\par
 %\LJ{I hope that the following is clear now.} 
\WC{We follow the threshold procedure as discussed in  \citep{politis2011higher}. %\LJ{I hope that the following is clear now.}  
In particular, let $s_n$ denote the stochastic upper bound under possible smooth and abrupt changes or long memory in Theorem 1, Theorem 2 and Theorem 3, respectively. 
 We choose the threshold level for the eigenvalues to be $1/n$,  so that $\rho(\hat \Sigma^{\mathrm{pd}}(t)-\hat \Sigma(t)) = o(s_n)$, where $\hat \Sigma^{\mathrm{pd}}(t)$ denotes the estimator after using threshold for eigenvalues, i.e.,  $\hat \Sigma^{\mathrm{pd}}(t)= U(t) \Lambda^{\mathrm{pd}}(t) U(t)^{\T}$, $\Lambda^{\mathrm{pd}}(t)$ is the diagonal matrix with diagonal elements $\lambda_i^{+}(t) = \max (1/n, \lambda_i(t) )~(i = 1, \ldots, p)$, where  $\lambda_i(t)$ is the ordered eigenvalue of $\hat \Sigma(t)$, and
 $U(t)$ is the matrix consisting of the corresponding eigenvectors of $\hat \Sigma(t)$.  }

\section{Discussion on assumptions}\label{sec:high}
\subsection{high-level assumptions}
\WC{It is possible to formulate high-level assumptions for the pilot estimator so that the debias effect can be achieved at least theoretically. In the {\it absence} of jump points, the high-level assumption of the  pilot estimator that will lead to the consistency of the difference-based long-run covariance matrix estimator is that for the event $G_n=\{
\text{$\breve \beta_0(t)$ is well-defined for $t \in I$}\}$ and $P(G_n) \to 1$,
\begin{align}
\sup_{t \in \I}\| \{\beta(t) - \breve \beta_0(t)\} 1(G_n)\|_{4\kappa} = O\{\tau_n^3+ (n\tau_n^{3/2})^{-1/2}\}.
\end{align} }
\par
 \WC{Let $\mathcal  C_m = \{j:j-m \leq i \leq j + m, \beta^{\prime \prime \prime} (t_i) ~\text{exists and continuous}\}$, $\mathcal C_n = \{j:j-n\tau_n^{3/2} \leq i \leq j + n\tau_n^{3/2}, \beta^{\prime \prime \prime} (t_i) ~\text{exists and continuous}\}$ and the number of $q_n$ jump points should satisfy the condition of Theorem 2.   In the presence of jump points, the high-level assumption will be
\begin{align}
   \sup_{j \in  \mathcal C_m}\|\{ \check \beta_0(t_{j+m}
      ) - \check { \beta}_0(t_j)\}1(G_n) \|_{4\kappa} = O\left\{\frac{mq_n}{n\tau_n^{3/2}} + (n\tau_n^{3/2})^{-1/2}\right\},
\end{align}
and 
\begin{align}
   \sup_{j \in  \mathcal C_n}\|\{ \check \beta_0(t_{j}
      ) - { \beta}_0(t_j)\}1(G_n) \|_{4\kappa} = O\left\{\tau_n^3 + (n\tau_n^{3/2})^{-1/2}\right\}.
\end{align}
For the case of long memory, from Lemma C.1 and Lemma C.2, we can obtain the high-level assumption is 
\begin{align}
    \sup_{t \in \I}\| \{\beta(t) - \breve \beta_0(t)\} 1(G_n)\|_{4\kappa} = O\{\tau_n^3+(n\tau_n^{3/2})^{d-1/2}\}.
\end{align}
Therefore, it is possible to find other estimators that satisfy the high-level assumptions to achieve a similar debias effect. However, in non-parametric estimation and inference, it is more convenient in practice to use tuning parameters as few as possible.  We recommend the statistic $\breve \beta(\cdot)$ in the paper, mainly because it can satisfy all the high-level assumptions without introducing extra smoothing parameters. }

\subsection{Discussion on Assumption 5}
\WC{ Recall our definition of short-range dependence in the paper:
\begin{definition}
The process $G(t, \F_i)$ is of $r$-order short-range dependence on interval $I$ if $\sup_{t \in I} \| G(t, \F_0) \|_{r} <\infty $, $\delta_{r}(H, k, I) = O(\chi^k)$, for some $\chi \in (0,1)$, $r \geq 1$, and $s$-order locally stationary on interval $I$, $s \geq 2$, if
 $G(t, \F_0) \in \mathrm{Lip}_{s}(I)$.
\end{definition}
}
\WC{Therefore, there are two restrictions of Assumption 5, which are the moment constraint and the constraints in order $r$ of the physical dependence measure 
\begin{align}
    \delta_r(L, k, I)=\sup _{t \in I}\left\|L\left(t, \mathcal{F}_k\right)-L\left(t, \mathcal{F}_{k,\{0\}}\right)\right\|_r.
\end{align}
}
\WC{Under the conditions of short-range dependence and no jump points, we allow $\kappa = 1$, i.e., $16$-order moment is required mainly because of the corresponding non-parametric smoothing of time series as well as the use of physical dependence.
% \WC{I don't understand the following}which upper bounds the non-linear dependence. 
In general, the requirement of $16$ is hard to be reduced because we estimate $\check \beta(\cdot)$ using second-order series and formulation of $\breve \Sigma(\cdot)$ involving the square of $x_ix_i^{\T}\breve \beta(t_i)$. Recall that   \begin{align}
   \breve{\Sigma}(t)=\sum_{j=m}^{n-m} \frac{m \hat{A}_{j, m} \hat{A}_{j, m}^{\mathrm{T}}}{2} \omega(t, j),\quad  \hat{A}_{j, m}=\frac{1}{m} \sum_{i=j-m+1}^j\left\{x_{i, n} x_{i, n}^{\mathrm{T}} \breve{\beta}\left(t_i\right)-x_{i+m, n} x_{i+m, n}^{\mathrm{T}} \breve{\beta}\left(t_{i+m}\right)\right\},
    \end{align}
    and 
    \begin{align}
        \breve {  \beta}(t) =   \Omega^{-1}(t)  \varpi (t)\quad (t \in [0,1]),
    \end{align}
    where $ \Omega(t)$ and $  \varpi (t)$ are the smoothed versions of $$\acute{  \Delta}_{j}/2 = \frac{1}{2m}\sum_{i=j-m+1}^{j} \tilde{  X}_{i,m}\tilde{  X}_{i,m}^{\T}~ \text{and}~\breve {  \Delta}_{j}/2 = \frac{1}{2m}\sum_{i=j-m+1}^{j} \tilde{  X}_{i,m}^{\T}\tilde{  Y}_{i,m},$$ where  $
    \tilde{  Y}_{i,m} =   x_{i,n} y_{i,n}-  x_{i+m, n} y_{i+m, n},  \tilde{  X}_{i,m} =   x_{i,n}   x_{i,n}^{\T}-  x_{i+m, n}   x_{i+m, n}^{\T}
    $. }%\WC{LJ: I suggest copy the full formula;} \LJ{done} 
    %\WC{also the next sentence is too strong if we don't make the sharpest argument. Maybe we just say this is the difficulty we have in weakening the seemingly strong technique assumption 5.} 
    %\LJ{Therefore,  the moment condition of order $16$ is difficult to relax, due to the existence of the variance of $\breve \Sigma(\cdot)$.} 
   \WC{ Therefore,  the moment condition of order $16$ is difficult to relax since we require $\breve \Sigma(\cdot)$ to have a finite second moment.}
    
  \WC{ Although Assumption 5 is seemingly strong for technical convenience, empirical studies show that our proposed estimator could still be consistent allowing  $H(\cdot,\mathcal F_i)$ to have a heavier tail than Assumption 5 under considered scenarios, though with possibly slower convergence rate.  %But the performance could be worse than that under Assumption 5. 
    In Figure \ref{fig:innovation}, we display the empirical rejection rates of different types of innovations, i.e., normal, $t(5)$ and $t(6)$ for $\zeta_i$ in  
    \begin{align}
u_{i,n}= G(t_i, \mathcal F_i) = 0.65 \cos(2\pi t_i) G(t_i,  \mathcal F_{i-1}) + \zeta_i,
    \end{align}
    in the model in Section 5.2 of the paper for the tests for structural stability:
\begin{align}
    y_{i,n} = 1 + m_{i,n} + x_{i,n,1}+ x_{i,n,2} + e_{i,n},\quad e_{i,n} = (1 + 0.1 x_{i,n,1})u_{i,n}. 
\end{align}
The results show that the tests for structural stability and long memory can still work reasonably well under some scenarios for example the CP1 case when we relax the moment conditions using $t(5)$ or $t(6)$.}

 \WC{   In Figure \ref{fig:innovationlrd}, 
 the empirical rejection rates of different types of innovations, i.e, normal, $t(4)$ and $t(6)$ for $\varepsilon_{i}$ in
 \begin{align}
     B(t, \mathcal{G}_{i})=\{0.3 - 0.4(t-0.5)^2\} B(t, \mathcal{G}_{i-1})+ 0.8\varepsilon_{i},
 \end{align}
 in the model in Section 5.3 of the paper for the tests for long memory, i.e.,
 \begin{align}
  y_{i, n}=\beta_{1}(t_i)+\beta_{2}(t_i) x_{i, n}+(1-\mathcal B)^{-d}e_{i,n},\quad
  e_{i,n} = H(t_i,\mathcal{F}_{i},\mathcal{G}_{i}) =  B\left(t_i,\mathcal{G}_{i}\right)\{1+W^2(t_i, \mathcal{F}_{i})\}^{1/2},
 \end{align}
 where $W\left(t, \mathcal{F}_{i}\right)= \{0.1 + 0.1\cos(2\pi t)\}W(t, \mathcal{F}_{i-1})+ 0.2\zeta_{i} + 0.7(t-0.5)^2$ for $ i =1\cdots, n$. 
 %\WC{the reviewer might forget the role of innovations. We could copy the model}\LJ{done} 
 The bandwidth selection procedure is identical to the one used in the paper.  %\WC{I see $t(4)$ in the figure but this does not appear in the answer.}\LJ{done}. 
 We can find the influences of moment conditions in both the structural stability test and the long memory detection, i.e., the empirical power of the former is reduced when the tails are heavier, while the empirical sizes of the latter increase with the heavier tails.}
    % \WC{The results show that ? I think we could change 'power performance' to empirical rejection rates under null different alternatives.If you only mention power performance then if then the reader might overlook our performance under null.  }

\begin{figure}
    \centering
    \includegraphics[width = 0.8\linewidth]{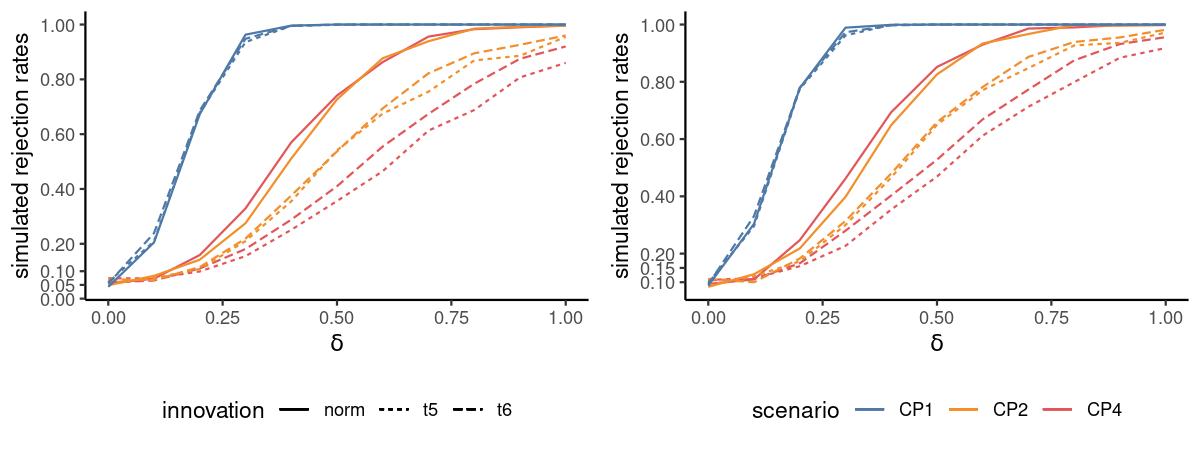}
    \caption{\WC{Empirical rejection rates of tests for structural stability with respect to different innovations: normal(solid), $t(5)$(small dashes) and $t(6)$(dashes) with sample size $n=300$.}}
    \label{fig:innovation}
\end{figure}

\begin{figure}
    \centering
    \includegraphics[width = 0.8\linewidth]{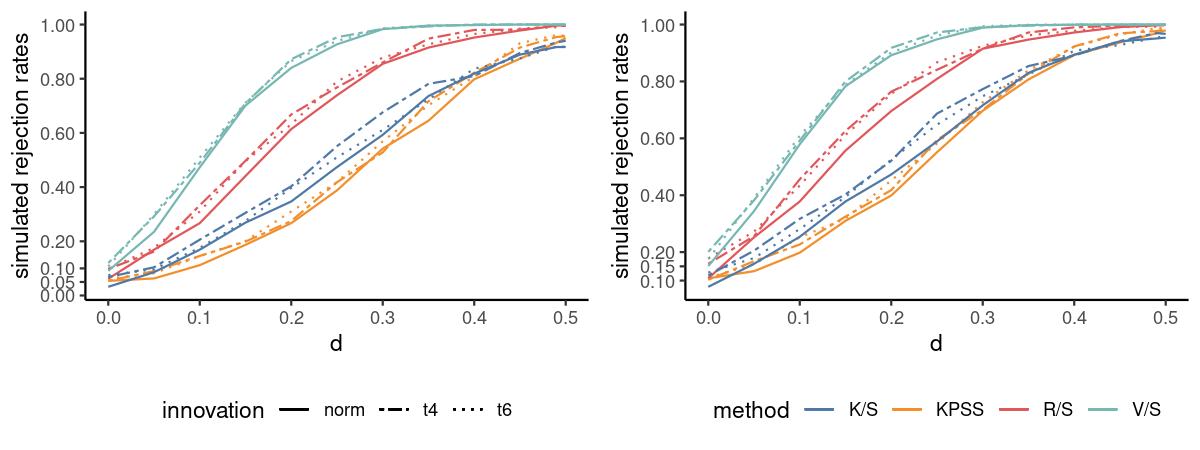}
    \caption{\WC{Empirical rejection rates of tests for long memory with respect to difference innovations: normal (solid), $t(4)$ (dashes) and $t(6)$ (dotted) with sample size $n=1500$.} }
    \label{fig:innovationlrd}
\end{figure}

\WC{Our current setting allows polynomial tailed distribution.  High-order physical dependence measure condition can be omitted when sub-exponential moment condition is assumed, where we only need order-1 physical dependence measure, see the following lemma.}

% \LJ{Definition $C$, $\sigma$.}
\WC{\begin{lemma}\label{lm:delta}
Suppose $\exists t_0 > 0$, $C= \sup_{t \in [0,1]} E \{\exp(t_0|H(t,   \F_0)|)\} <\infty$, $\delta_1(H,l,[0,1]) = O(\chi^l)$ for some $\chi \in (0,1)$. Then, we have the following holds for $q>2$: (i) $\sup_{t\in[0,1]}  E(|H(t, \F_0)|^q) \leq C t_0^{-q} q^q$. (ii) There exists  a positive constant $C^{\prime}$ such that  $\delta_q(H,l,[0,1]) = \sup_{t \in [0,1]} \|H(t, \F_l) - H(t, \F_{l,\{0\}})\|_{q} \leq C^{\prime} q \chi^{l/q}$.
\end{lemma}
\begin{proof}
For $p \geq P$, $P > 0$, elementary calculation gives that $x > p \log (x/p)$, $x > 0$. 
   For a sufficiently large constant $C$, we have 
\begin{align}
   \sup_{t\in[0,1]} E(|H(t, \F_0)|^q)  < \left(\frac{q}{t_0} \right)^q   \sup_{t\in[0,1]} E \{ \exp(t_0|H(t, \F_0)|)\} = t_0^{-q} q^q C. \label{eq:moment}
\end{align}
 By Hölder inequality,
\begin{align}
  \sup_{t \in [0,1]} \|H(t,F_l) - H(t,F_{l,\{0\}})\|_{q}^2 & =  [ E\{|H(t,F_l) - H(t,F_{l,\{0\}})|^{q-1/2}|H(t,F_l) - H(t,F_{l,\{0\}})|^{1/2}\}]^2\\
    & \leq  E\{|H(t,F_l) - H(t,F_{l,\{0\}})|^{2q-1} \}  E\{|H(t,F_l) - H(t,F_{l,\{0\}})|\} \\ 
    & \leq \|H(t,F_l) - H(t,F_{l,\{0\}})\|^{2q-1}_{2q-1}\delta_1(H,l,[0,1]) \\ &\leq C_1 (2q-1)^{2q-1} t_0^{-2q+1} \chi^l,\label{eq:delta}
\end{align}
where $C_1$ is a sufficiently large constant.
Therefore, we have $\delta_q(H,l,[0,1])\leq C^{\prime} \chi^{l/q}q$.  
\end{proof}
}
\section{Other potential applications}\label{sec:appl}
\WC{Under local stationarity, our proposed estimator can be used in many practical scenarios, such as constructing simultaneous confidence bands for time-varying regression coefficients, deriving preliminary estimation and visualization of the long memory parameter $d$ for locally stationary long memory process, as well as many other inference problems that involve the estimation of the long-run covariance matrix, such as testing for white noises, generalized likelihood ratio test and squared integrated tests for time-varying regression coefficient functions, see \citep{zhou2014glr}  for instance. In the following we list several detailed examples.}\par 
\WC{ \textbf{Visualization of long memory.}
The long-run covariance estimator can serve as a simple and heuristic tool for visualizing and assessing the presence of long memory, see Section 1.2 of \citep{BeranLongmemory}.
 By Theorem 3, under the fixed alternative, we have 
$$
\sup_{t \in  I}\left| m^{-2d}\hat{{ \Sigma}}(t) -\kappa_2(d)\sigma_H^2(t)  \mu_W(t) \mu^{\T}_W(t)\right| = o_p(1).\nonumber
$$
 Suppose we have a grid of $m$'s, i.e., $m_1, \ldots, m_M$, and the corresponding long-run covariance estimator $\hat \Sigma_1(\cdot), \ldots, \hat \Sigma_M(\cdot)$ calculated using $m_1,...,m_M$, respectively.}
 % \WC{what do you mean by log(matrix)}
\WC{Taking $x_i = \log m_i$, $y_i = \sum_{j=1}^n\log |\hat \Sigma_i(t_j)|/n,\ (i=1, \ldots, M)$, where $|\cdot|$ denotes the Frobenious norm, and we have
\begin{align}
y_i \approx (2d) x_i + \log \kappa_2(d)  + \sum_{j=1}^n \log \sigma_H^2(t_j)/n + \sum_{j=1}^n \log |\mu_W(t_j)\mu^{\T}_W(t_j)|/n\\
= (2d) x_i + f(d),\label{fittedreg}
\end{align}
where the quantity $f(d)$ is independent of $i$. Therefore, one can visualize $d$ by drawing a regression line $y_i\sim x_i$.  For illustration, we generate a data set from the functional linear model in Section 6.3 of the main paper with $d = 0.2$, i.e., }
\WC{\begin{align}
    y_{i,n} =\beta_1(t_i)+\beta_2(t_i)x_{i,n} +(1 - \mathcal B)^{-d}e_{i,n},
\end{align}
where $\mathcal B$ is the lag operator, $\beta_{1}(t) = 8 \sin (\pi t) $, $\beta_{2}(t)=4 \exp \{-2 \left(t-0.5\right)^{2}\}$, $x_{i, n}=W(t_i, \mathcal{F}_{i})\ (i=1, \ldots, n)$, and $ e_{j,n} = H(t_j,\mathcal{F}_j, \mathcal{G}_j)\ (j=1, \ldots, n)$ with} \WC{$$ H(t,\mathcal{F}_{i},\mathcal{G}_{i}) =  B\left(t,\mathcal{G}_{i}\right)\{1+W^2(t, \mathcal{F}_{i})\}^{1/2}\quad (i \in Z;\ t \in [0,1]),$$  where $
     W\left(t, \mathcal{F}_{i}\right)= 0.1\cos(2\pi t)W(t, \mathcal{F}_{i-1})+ 0.2\zeta_{i} + 0.4(t-0.5)^2,
$ and $B(t, \mathcal{G}_{i})=\{0.3 - 0.4(t-0.5)^2\} B(t, \mathcal{G}_{i-1})+ 0.6\varepsilon_{i}$. }

\WC{\cref{fig:graphical} displays an instance for the visualization. The $y$-axis is the average logarithm of Frobenius norm of the long-run covariance estimator, while the $x$-axis is the logarithm of the parameter $m$. The displayed fitted regression line of \eqref{fittedreg} is $y=-0.82+0.41x$, and the estimated $d$ is close to the half of the slope, i.e. $0.205$. Furthermore, 100 times of simulations yield the average estimated $d$ being $0.204(0.006)$.}
\begin{figure}[!ht]
    \centering
    \includegraphics[width = 0.5\linewidth]{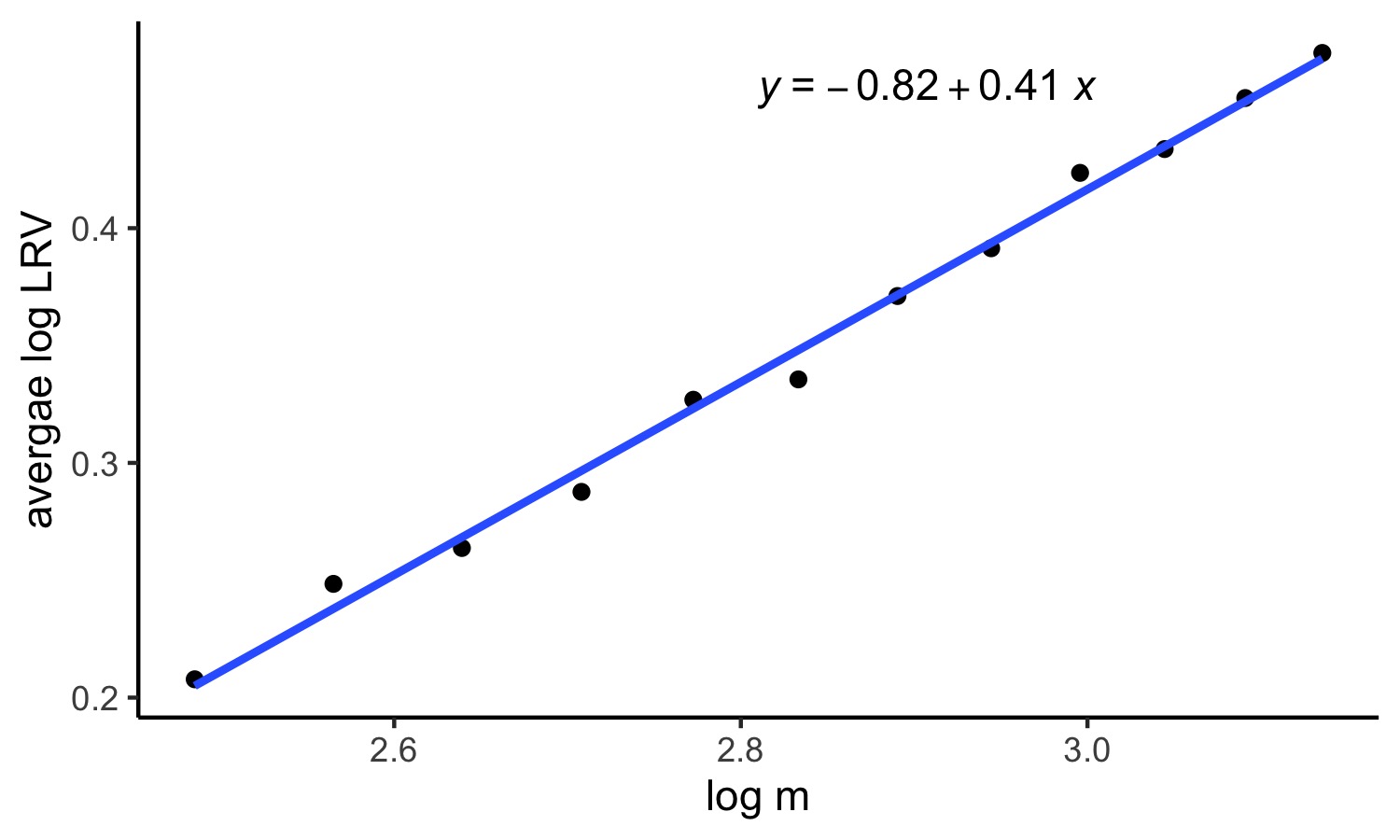}
    \caption{\WC{Regression using data from the long-run covariance matrix estimator. The estimated $d$ is half of the slope.} }
    \label{fig:graphical}
\end{figure} \par
 \WC{\textbf{Simultaneous confidence tubes.} We consider a similar functional linear model as in Section 6.3 of the main article with $d=0$, namely}
% \WC{please rewrite it using the new formulation, e.g., $(1-B)^d e_i$}
\WC{\begin{align}
    y_{i, n}=\beta_{1}(t_i)+\beta_{2}(t_i) x_{i, n}+e_{i,n} \quad (i=1, \ldots, n),
\label{M1}
\end{align}
where $\beta_{1}(t)=4 \sin (\pi t) $, $\beta_{2}(t)=4 \exp \{-2 \left(t-0.5\right)^{2}\} $, $x_{i, n}=W(t_i, \mathcal{F}_{i})$ with $
     W\left(t, \mathcal{F}_{i}\right)= \{0.25 + 0.25\cos(2\pi t)\}W(t, \mathcal{F}_{i-1})+ 0.2\zeta_{i}$, and $ e_{i,n} = H(t_i,\mathcal{F}_i, \mathcal{G}_i)$, where $t_i=i/n$ and $$
    H(t, \mathcal{F}_{i},\mathcal{G}_{i}) =  B\left(t,\mathcal{G}_{i}\right)\{1+W^2(t, \mathcal{F}_{i})\}^{1/2},$$ 
where $B(t, \mathcal{G}_{i})=\{0.2 - 0.4(t-0.5)^2\} B(t, \mathcal{G}_{i-1})+ 0.8\varepsilon_{i}$.  }
\WC{\citep{zhou2010simultaneous} considers the simultaneous confidence tubes for regression coefficient functions using plug-in estimators,  see estimator (17) of their paper, which is also available in our R package. We compare the performance of simultaneous confidence tubes jointly for $\beta_1(t)$ and $\beta_2(t)$, using our estimator and the plug-in estimator advocated by \citep{zhou2010simultaneous} in \cref{tb:1}, which shows the advantage of our estimator. Notice that $H(t, \mathcal F_i, \mathcal G_i)$ depends on the covariates in a nonlinear way, which has not been investigated empirically by \citep{zhou2010simultaneous}.  }
% Coverage rate.
\begin{table}[ht]
\centering
\def~{\hphantom{0}}
\begin{tabular}{rrrrr}
\hline
  &\multicolumn{2}{c}{Ours}  & \multicolumn{2}{c}{Plug-in}\\
  \hline
 b/nominal &95\% &  90\%& 95\%  & 90\% \\ 
  \hline
% 0.2500 & 93.0 & 88.0 & 86.6 & 79.5 \\ 
%   0.2750 & 94.5 & 89.8 & 89.4 & 82.9 \\ 
%   0.2875 & 94.9 & 90.9 & 89.1 & 80.6 \\ 
%   0.3000 & 94.7 & 91.0 & 88.1 & 80.0 \\ 
%   0.3125 & 95.2 & 91.7 & 87.4 & 79.2 \\ 
%   0.3250 & 95.3 & 91.0 & 88.3 & 80.1 \\ 
%   0.3375 & 95.3 & 92.0 & 87.8 & 79.5 \\ 
%   0.3500 & 95.2 & 90.6 & 88.2 & 79.2 \\ 
0.2500 & 89.5 & 83.8 & 87.0 & 78.5 \\ 
  0.2750 & 89.9 & 84.8 & 86.8 & 77.2 \\ 
  0.2875 & 90.9 & 85.1 & 85.7 & 78.5 \\ 
  0.3000 & 92.1 & 86.3 & 87.1 & 77.8 \\ 
  0.3125 & 92.8 & 87.3 & 86.5 & 78.2 \\ 
  0.3250 & 92.2 & 86.5 & 87.4 & 79.8 \\ 
  0.3375 & 93.3 & 87.6 & 88.0 & 78.8 \\ 
  0.3500 & 91.2 & 85.3 & 86.6 & 78.1 \\ 
   \hline
\end{tabular}
\caption{Empirical coverage rates (in \%) via 1000 times of simulations for the simultaneous confidence tubes of $(\beta_1(t), \beta_2(t))$ with sample size $500$  using our estimator and the plug-in estimator advocated by \citep{zhou2010simultaneous}.}
% \caption{Empirical coverage rates (in \%) of 2000 times of simulation of the simultaneous band using our estimator and plug-in estimator as in \citep{zhou2010simultaneous}.}
\label{tb:1}
\end{table}
\section{Extra simulation}\label{sec:lrv}
\subsection{Sensitivity analysis}
\WC{For the sensitivity check, we examine the performance of our method with different choices of $m$'s and $\tau_n$'s under both null and various alternative hypotheses and compare it to that of baseline methods using nonparametric or ordinary least square residuals for estimating the long-run covariance matrix.}  % we consider fixed choices of $m$ and $\tau_n$ and \LJ{power sensitivity with respect to the prespecified constants $c_1, c_2$ of extended minimum volatility selection.}\par 
\begin{itemize}
    \item 
\WC{Figure \ref{fig:diff} demonstrates the power performance of the tests equipped with the debiased difference-based estimator compared with the baseline method based on ordinary least square residuals $\hat e_{i,n}=y_{i,n}-x_{i,n}^{\T}\hat \beta_n$ in the CP1 model in the main paper, where we  compare fixed $m$'s ($10$ and $20$ with $\tau_n = n^{-2/15}$) and fixed $\tau_n$'s ($0.2$, $0.3$, $0.4$ with $m = 10$) as well as $m$ selected via extended minimum volatility over different ranges. In the paper, we choose $m$ from $6$ ($\lfloor 10/7n^{4/15} \rfloor$) to $9$ ($\lfloor 15/7n^{4/15}\rfloor$).
In \cref{fig:diff}, we show the results when  choosing $m$ from $6$ ($\lfloor 10/7n^{4/15} \rfloor$) to $30$ ($\lfloor 50/7n^{4/15} \rfloor$) and from $10$  ($\lfloor 18/7n^{4/15} \rfloor$) to $50$ ($\lfloor 75/7n^{4/15} \rfloor$). Using our difference-based debiased long-run covariance matrix estimator, the test outperforms that based on ordinary least square residuals by a large margin under different choices of tuning parameters. }
\item  \WC{Similar improvement can also be found in \cref{fig:diff4} in the CP4 model.}
\item \WC{In \cref{fig:lrdvs}, the roles of tuning parameters in the detection of long memory are investigated. The bootstrap tests equipped with the debiased difference-based estimator under different smoothing parameters are compared with the baseline method based on plugging in nonparametric residuals in \citep{zhou2010simultaneous} as in the main paper. Long-memory tests with the proposed difference-based estimator  achieve much better trade-offs in type-I and type-II errors under various choices of tuning parameters.}%even if the tuning parameters are perturbed. 
% \citep{wu2018gradient} proposes a general bootstrap statistics relying on the ordinary least square residuals $\hat e_{i,n}=y_{i,n}-x_{i,n}^{\T}\hat \beta_n$.
\end{itemize}

% \WC{what is the selection of power}

\begin{figure}
    \centering
\includegraphics[width=0.8\linewidth]{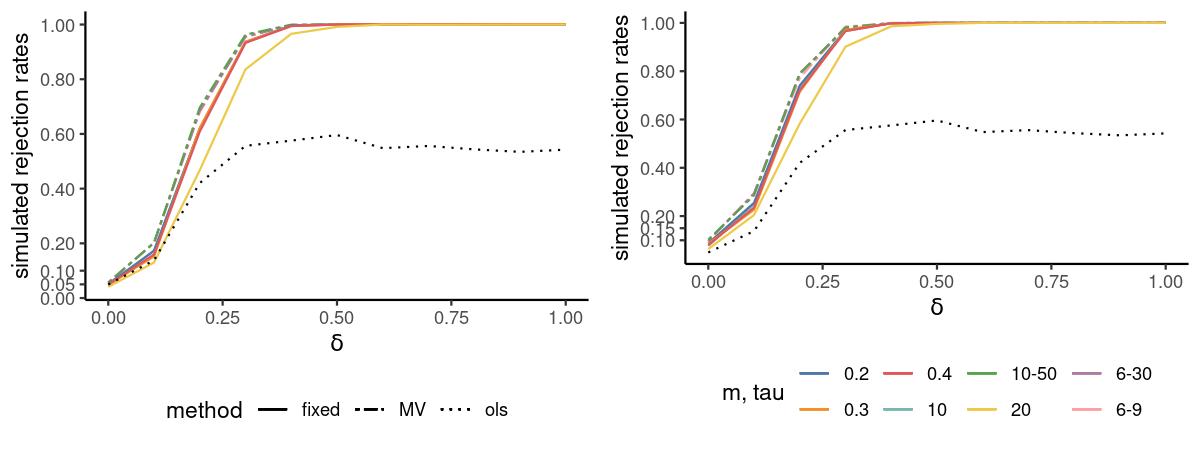}
    \caption{\WC{Comparing structural stability tests of using plug-in estimators of ordinary least square residuals (ols, dotted), the proposed difference-based estimator using extended minimum volatility selection procedure (MV, dashes), choosing $m$ from $6$ ($\lfloor 10/7n^{4/15} \rfloor$) to $9$ ($\lfloor 15/7n^{4/15}\rfloor$) versus choosing $m$ from $6$ ($\lfloor 10/7n^{4/15} \rfloor$) to $30$ ($\lfloor 50/7n^{4/15} \rfloor$), and $10$ ($\lfloor 18/7n^{4/15} \rfloor$) to $50$ ($\lfloor 75/7n^{4/15} \rfloor$), as well as the proposed difference-based estimator with several fixed $m$'s with $\tau_n = n^{-2/15}$ and $\tau_n$'s with $m = 10$ (fixed, solid) when there is one change point (CP1).}}
    \label{fig:diff}
\end{figure}

\begin{figure}
    \centering
\includegraphics[width=0.8\linewidth]{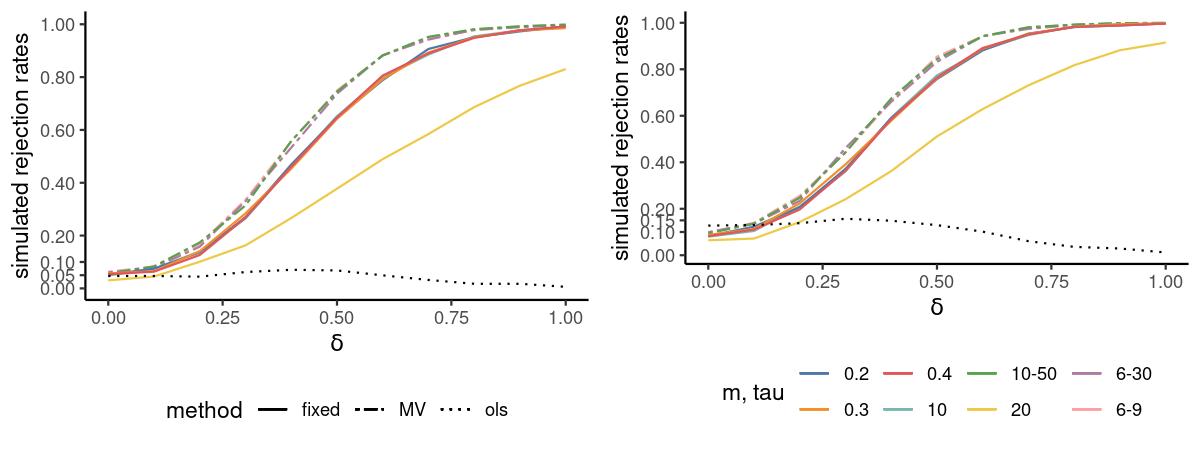}
    \caption{\WC{Comparing structural stability tests of using plug-in estimators of ordinary least square residuals (ols, dotted), the proposed difference-based estimator using extended minimum volatility selection procedure (MV, dashes), choosing $m$ from $6$ ($\lfloor 10/7n^{4/15} \rfloor$) to $9$ ($\lfloor 15/7n^{4/15}\rfloor$) versus choosing $m$ from $6$ ($\lfloor 10/7n^{4/15} \rfloor$) to $30$ ($\lfloor 50/7n^{4/15} \rfloor$), and $10$ ($\lfloor 18/7n^{4/15} \rfloor$) to $50$ ($\lfloor 75/7n^{4/15} \rfloor$) as well as the proposed difference-based estimator with several fixed $m$'s with $\tau_n = n^{-2/15}$ and $\tau_n$'s with $m = 10$ (fixed, solid) when there are 4 change points (CP4).}}
    \label{fig:diff4}
\end{figure}

\begin{figure}
    \centering
\includegraphics[width=0.8\linewidth]{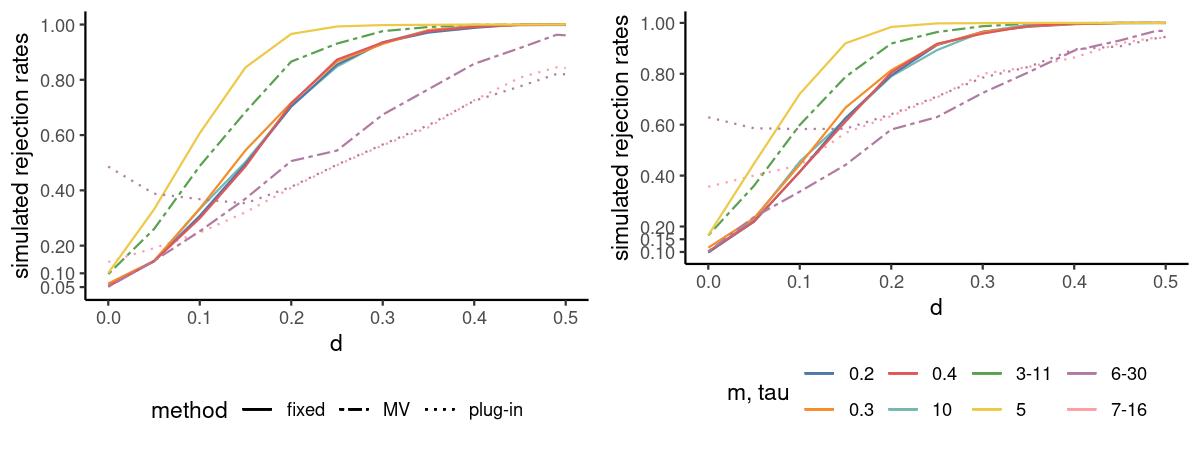}
    \caption{\WC{Comparing V/S-type tests of  using plug-in estimators of nonparametric residuals (plug-in, dashes), the proposed difference-based estimator choosing $m$ from $3$ to $11$  versus from $6$ to $30$, and from $7$ to $16$ (MV, dotted) as well as using several fixed $m$'s with $\tau_n = n^{-2/15}$ and $\tau_n$'s with $m = 10$ (fixed, solid).}}
    \label{fig:lrdvs}
\end{figure}

\subsection{Simulated rejected rates with sample size 750}
The following \cref{tb:size} reports the simulated sizes of KPSS, R/S, V/S and K/S-type tests when $n = 750$. 
\begin{table}[ht]

\centering
\begin{tabular}{rrrrrrrrr}
 \hline
 & \multicolumn{2}{c}{KPSS} & \multicolumn{2}{c}{R/S}   & \multicolumn{2}{c}{V/S}  & \multicolumn{2}{c}{K/S}    \\ 
  \hline
 & $5\%$ & $10\%$ & $5\%$ & $10\%$ & $5\%$ & $10\%$ & $5\%$ & $10\%$ \\ 
  \hline
% plug-in  & 7.60 & 15.00 & 14.90 & 28.60 & 18.00 & 36.20 & 7.70 & 16.50 \\ 
%   diff & 4.80 & 8.90 & 3.70 & 6.20 & 3.60 & 5.90 & 4.50 & 9.40 \\ 
 %   plug-in & 8.50 & 15.70 & 18.10 & 29.20 & 50.70 & 65.20 & 9.30 & 19.40 \\
 % diff & 4.00 & 7.70 & 4.70 & 8.20 & 7.50 & 12.90 & 4.00 & 8.90 \\
 plug-in & 8.40 & 15.00 & 21.00 & 32.80 & 52.00 & 67.10 & 13.30 & 22.30 \\
diff & 4.00 & 7.50 & 6.10 & 11.50 & 10.00 & 15.50 & 3.30 & 7.30 \\
   \hline
\end{tabular}
\caption{
Simulated rejection rates of KPSS, R/S, V/S and K/S-type tests for long memory when $d=0$, $n = 750$.}
\label{tb:size}
\end{table}
\subsection{Long-run covariance estimates with change points}
To further illustrate the role of long-run variance estimators in bootstrap tests, we conduct $1000$ times simulation, where we find the MSE of long-run variance estimators using ordinary least square residuals is $974.12$, more than twice the magnitude of the MSE of  that of difference-based method $275.94$ when there are four change points in the time series structure (scenario CP4 with $
\delta = 1$). Moreover, two-sample t test also shows that under the alternative with $\delta = 1$, the difference-based method yields long-run variance estimate with much smaller MSE than the  ordinary least square method with $p$-value smaller than $0.01\%$.
\section{Proof}\label{sec:proof}\label{diff}
\subsection{Notation}
We first introduce some notation that will be frequently used in the mathematical argument of this section. In the following proofs, we will omit the index $n$ in $e_{i,n}, x_{i,n}, y_{i,n}$ for simplicity. Define filtration $\FF_i=(\varepsilon_{-\infty},...,\varepsilon_i)$ for independent and identically distributed random variables $(\varepsilon_{i})_{i\in   Z}$. For a random vector $(v_i)_{i=1}^n \in \F_s$, let $v_{i,\{s\}}$ denote the series replacing the $\varepsilon_{s}$ with its independent and identically distributed copy.  For a random matrix $(A_i)_{i=1}^n \in \F_s$, 
define $A_{j,\{s\}}$ as the random matrix replacing  $\varepsilon_{s}$ in $A_{j}$ with its independent and identically distributed copy. Recall that $ e_{i,n}^{(d)} = \sum_{j=0}^{\infty}\psi_j(d) e_{i-j, n}$, $ e_{i,n}^{(d_n)} = \sum_{j=0}^{\infty}\psi_j(d_n) e_{i-j, n}$. For the sake of simplicity, we use $\psi_j$ to represent $\psi_j(d)$ when we discuss the fixed alternatives and $\psi_j(d_n)$ for the theory of the local alternatives. Recall $t_i = i/n$, and $K^*(x)$ denotes the jackknife equivalent kernel $2 \surd{2} K(\surd{2}x) - K(x)$. Let $0 \times \infty = 0$, $a_n \sim b_n$ denote $\lim_{n \to \infty} a_n/b_n = 1$ for real sequences $a_n$ and $b_n$.  Let $I = [\gamma_n, 1-\gamma_n] \subset (0,1)$, $\gamma_{n}=\tau_{n}+(m+1) / n.$ Recall  $$A_{j,m} = \frac{1}{m} \sum_{i=j-m+1}^j\{x_i x_i^{\T} { \beta}(t_i) - x_{i+m} x_{i+m}^{\T} { \beta}(t_{i+m})\},\quad { \Sigma}^A (t) = \sum_{j=m}^{n-m} \frac{m A_{j,m} A_{j,m}^{\T}}{2}\omega(t, j),$$
$$\hat{A}_{j,m} = \frac{1}{m} \sum_{i=j-m+1}^j\{x_i x_i^{\T} \breve{ \beta}(t_i) - x_{i+m} x_{i+m}^{\T} \breve{ \beta}(t_{i+m})\},\quad \breve{ \Sigma} (t) = \sum_{j=m}^{n-m} \frac{m \hat{A}_{j,m} \hat{A}_{j,m}^{\T}}{2}\omega(t, j),$$
where $ \omega(t, i)= K_{\tau_n}\left(t_i-t\right) / \sum_{i=1}^{n} K_{\tau_{n}} \left(t_i-t\right)$.
Let 
\begin{align}
{\grave A}_{j,m} = \frac{1}{m}\sum_{i=j-m+1}^{j} (x_i x_i^{\T}-x_{i+m} x_{i+m}^{\T}) \{x_i x_i^{\T} { \beta}(t_i) -x_{i+m} x_{i+m}^{\T} { \beta}(t_{i+m}) \},\label{eq:graveA}
\end{align}
and
\begin{align}
  \grave { \Delta}_{j} = \frac{1}{m}\sum_{i=j-m+1}^{j} (x_i x_i^{\T}-x_{i+m} x_{i+m}^{\T}) (x_i e_i-x_{i+m} e_{i+m}).\label{eq:graveD}
\end{align}
 
 We define below the counterparts of $Q_{k,m}$, ${ \Delta}_j$ and $\hat{{ \Sigma}}(\cdot)$ in \cref{sec:diff}. Define for $m \geq 2$, $t \in [m/n, 1 - m/n]$,
    \begin{align}
      \tilde{Q}_{k,m} = \sum_{i=k}^{k+m-1} x_ie_i,\quad
      \tilde{ \Delta}_{j}=\frac{\tilde{Q}_{j-m+1, m}- \tilde{Q}_{j+1, m}}{m},\quad \tilde{{ \Sigma}}(t)=\sum_{j=m}^{n-m} \frac{m \tilde{ \Delta}_{j}\tilde{ \Delta}_{j}^{\T}}{2}\omega(t, j).\label{eq:tildeSigma}
    \end{align}
    % where $ \omega(t, i)= K_{\tau_n}\left(t_i-t\right) / \sum_{i=1}^{n} K_{\tau_{n}} \left(t_i-t\right)$. 
\subsection{Proof of \texorpdfstring{ \cref{thm:lrv_diff}}{Proof of Theorem 1} }\label{proof:thm7.1}

The following lemma presents the consistency of bias correction for the difference-based estimator with time series covariates, which is crucial to establish the consistency result of \cref{thm:lrv_diff}.

\begin{lemma}
% \WC{Checked, please change green to dark}
Under the condition of \cref{thm:lrv_diff}, we have 
\begin{align}
\sup_{t \in \I}\lt| \breve { \Sigma}(t) - { \Sigma}^A (t)\rt|= \Op\left( \surd{m}\tau_n^{3-1/\kappa} + \surd{\frac{m}{n\tau_n^{3/2+2\kappa}}} \right).
\end{align}
\label{lm:bias}\label{lm:correction}
\end{lemma}

\begin{proof}[Proof of \cref{lm:bias}]
Let $l_n$ be a sequence of real numbers so that $l_n \to \infty$ arbitrarily slow.
Define $A_n = \{ \sup_{t \in  \I} |{ \Omega}(t) - M^+(t)| \leq l_n \{(m\tau_n^{3/2})^{-1/2} + m/(n\tau_n^{3/4}) + \tau_n^{3/4}\}\}$. By \eqref{eq:remark5.2}, $\lim_{n\to\infty} \pp(A_n) = 1.$ Since $0 \leq \sum_{j=m}^{n-m} \omega(t,j) \leq 1$, we have
  \begin{align}
    &\sup_{t \in \I} \lt\| (\breve { \Sigma}(t) - { \Sigma}^A (t))1(A_n)\rt\|_{\kappa}\\ 
    % &\leq \sup_{t \in \I} \sum_{j=m}^{n-m}\frac{m\omega(t,j)}{2} \lt\|(\hat{A}_{j,m}\hat{A}_{j,m}^{\T} -  A_{j,m}A_{j,m}^{\T})1(A_n)\rt\|_{\kappa} \\
    &\leq \sup_{t \in \I} \sum_{j=m}^{n-m}\frac{m\omega(t,j)}{2} \lt\|(\hat{A}_{j,m} -  A_{j,m})1(A_n)\rt\|_{2\kappa}\left(\lt\|(\hat{A}_{j,m} -  A_{j,m})1(A_n)\rt\|_{2\kappa} + 2\|A_{j,m}\|_{2\kappa}\right) \\ 
    &\leq m \max_{m \leq j \leq n-m}  \left\|(\hat{A}_{j,m} -  A_{j,m})1(A_n)\right\|_{2\kappa} \\ &\times \lt(\max_{m \leq j \leq n-m}  \lt\|(\hat{A}_{j,m} -  A_{j,m})1(A_n)\rt\|_{2\kappa} + 2\max_{m \leq j \leq n-m}\|A_{j,m}\|_{2\kappa}\rt).\label{eq:fixtA}
  \end{align}
  \par
 First, we shall show that 
 \begin{align}
  \max_{m \leq j \leq n-m}\|A_{j,m}\|_{2\kappa} = O(m^{-1/2} + m/n)  = O(m^{-1/2}).\label{eq:Ajm}
 \end{align}
 Define $B_{j,m} = \frac{1}{m}\sum_{i=j-m+1}^j x_i x_i^{\T}.$
    Notice that 
    \begin{align}
      \| A_{j,m}\|_{2\kappa} &\leq \sup_{m \leq i \leq n - m}||{ \beta}(t_i)-{ \beta}(t_{i+m})| \|B_{j,m}\|_{2\kappa} + \sup_{t\in [0,1]}|{ \beta}(t)|\| B_{j,m} - B_{j+m,m}\|_{2\kappa} = A_1 + A_2.\label{eq:Astar}
    \end{align}
    Since ${ \beta}(t)$ is Lipschitz continuous,
    and under Assumption \ref{E:WW}, $\max_{m \leq j \leq n}\|B_{j,m}\|_{2\kappa}$ is bounded, we have 
    \begin{align}\label{eq:A1}
      A_1 = O(m/n).
    \end{align}
    For the calculation of $A_2$, notice that
    \begin{align}
      \| B_{j,m} - B_{j+m,m}\|_{2\kappa} & \leq \|B_{j,m}- \E(B_{j,m})\|_{2\kappa} +  \|B_{j+m,m}-\E(B_{j+m,m})\|_{2\kappa}\\ &+\|\E(B_{j,m})- M(t_j)\|_{2\kappa}+\|\E(B_{j+m,m})- M(t_j)\|_{2\kappa}.
    \end{align}
    Similar to  Lemma 6 in \citep{zhou2010simultaneous}, using rectangular kernel with bandwidth $m/n$, under Assumption \ref{E:WW}, we have
    \begin{align}
      \sup_{m \leq j \leq n}\|B_{j,m}- \E(B_{j,m})\|_{2\kappa} = O(m^{-1/2}).
    \end{align}
    Since $\E(B_{j,m}) = \frac{1}{m}\sum_{i=j-m+1}^j M(t_j)$ and $M(t)$ is Lipschitz continuous, 
   $
      \|\E(B_{j,m})- M(t_j)\|_{2\kappa}  = O(m/n).
   $
    Finally, by the boundedness of $\sup_{t \in [0,1]} |{ \beta}(t)|$, we have 
    \begin{align}\label{eq:A2}
      A_2 = O(m^{-1/2 }+ m/n).
    \end{align}
    
    Therefore, by \eqref{eq:A1} and \eqref{eq:A2}, we  have shown \eqref{eq:Ajm}.

Second, by triangle inequality, we have
\begin{align} 
\max_{m \leq j \leq n-m}  \lt\|(\hat{A}_{j,m} -  A_{j,m})1(A_n) \rt\|_{2\kappa} 
& \leq 2\max_{1 \leq i \leq n} \lt\| x_i x_i^{\T}\rt\|_{4\kappa}\|\{{ \beta}(t_i) - \breve{ \beta}(t_i)\}1(A_n) \|_{4\kappa}.
\end{align}
% \begin{align}
% \max_{m \leq j \leq n-m}  \lt\|(\hat{A}_{j,m} -  A_{j,m})1(A_n) \rt\|_{2\kappa} &\leq
%     \sup_{m \leq i \leq n - m} \|({\breve \beta}(t_i)-{\breve \beta}(t_{i+m}))-({ \beta}(t_i)-{ \beta}(t_{i+m}))\|_{4\kappa} \|B_{j,m}\|_{4\kappa}\\ & + \sup_{t\in [0,1]}\|{ \breve \beta}(t)-{ \beta}(t)\|_{4\kappa} \| B_{j,m} - B_{j+m,m}\|_{4\kappa} 
% \end{align}
Since under Assumption \ref{E:WW}, $\max_{1 \leq i \leq n} \lt\| x_i x_i^{\T} \rt\|_{4\kappa}  = O(1)$, we shall show that
\begin{align}
\sup_{t \in \I} \|\{ { \beta}(t) - \breve{ \beta}(t)\}1(A_n) \|_{4\kappa} = O\left\{\tau_n^3 + (n\tau_n^{3/2})^{-1/2}\right\}.\label{eq:betabreve}
\end{align}

Let  $M^+(t) =\E\{\bar{J}(t, \F_0)\bar{J}^{\T}(t, \F_0)\}$. 
Following similar arguments in Lemma 6 of \citep{zhou2010simultaneous} and Theorem 5.2 of \citep{dette2019detecting}, under Assumptions \ref{E:WW} and \ref{Ass-E}, we have 
\begin{align}
   \sup_{m \leq j\leq n} \|\acute { \Delta}_j/2 - M^+(t_j)\| = O(m^{-1/2} + m/n).
\end{align}
Then, it follows that 
\begin{align}
      \sup_{t \in \I } \|{ \Omega}(t) - M^+(t)\| = O(m^{-1/2} + m/n + \tau_n^{3/2}).
\end{align}
By the chaining argument in Proposition B.1 in Section B.2 in \citep{dette2018change}, we have
\begin{align}
      \sup_{t \in \I} |{ \Omega}(t) - M^+(t)| = \Op\{(m\tau_n^{3/2})^{-1/2} + m/(n\tau_n^{3/4}) + \tau_n^{3/4}\}.\label{eq:remark5.2}
\end{align}

Note that ${ \Omega}(t)$ is invertible on $A_n$. 
Then, for a sufficiently large constant $C$, we have
\begin{align}
    \|\{\breve { \beta}(t) - { \beta}(t)\} 1(A_n)\|_{4\kappa} & = \| { \Omega}^{-1}(t)\{{ \varpi}(t) - { \Omega}(t){ \beta}(t)\}1(A_n)\|_{4\kappa} \\ 
    &\leq \| \rho(  { \Omega}^{-1}(t)) | { \varpi}(t) - { \Omega}(t){ \beta}(t)| 1(A_n)\|_{4\kappa} \leq C \| { \varpi}(t) - { \Omega}(t){ \beta}(t)\|_{4\kappa}.
    \label{eq:betabreveA}
\end{align}
Then, it's sufficient to show that
\begin{align}
\sup_{t \in \I} \|{ \varpi}(t) - { \Omega}(t) { \beta}(t)\|_{4\kappa} = O\left\{\tau_n^3 + (n\tau_n^{3/2})^{-1/2}\right\}.
\end{align}

Recall the definition of 
$\grave { \Delta}_{j}$ and $\grave{A}_{j,m}$ in \eqref{eq:graveD} and \eqref{eq:graveA} respectively.
Observe that 
\begin{align}
{ \varpi}(t)& = \sum_{j=m}^{n-m} \frac{\tilde \omega(t, j)}{2} \grave{A}_{j,m} + \sum_{j=m}^{n-m} \frac{\tilde \omega(t, j)}{2} \grave { \Delta}_j : =  W_1(t) + W_2(t),\label{eq:varpi}
\end{align}
where $W_1(t), W_2(t)$ are defined in the obvious way. Recall that 
$\acute { \Delta}_{j} = \frac{1}{m}\sum_{i=j-m+1}^{j} (x_{i,n} x_{i,n}^{\T}-x_{i+m, n} x_{i+m, n}^{\T})^2$.
By triangle inequality, we have 
\begin{align}
\sup_{t \in \I} \left\| W_1(t) - { \Omega}(t) { \beta}(t) \right\|_{4\kappa} 
&\leq \sup_{t \in \I} \left\|W_1(t) - \sum_{j=m}^{n-m} \frac{ \acute { \Delta}_j  \tilde \omega(t, j)}{2} { \beta}(t_j) \right\|_{4\kappa}  +  \sup_{t \in \I} \left\| \sum_{j=m}^{n-m} \frac{\acute { \Delta}_j  \tilde \omega(t, j)}{2} \{{ \beta}(t_j)-{ \beta}(t)\} \right\|_{4\kappa}\\ 
& = W_{11} + W_{12}.
\end{align}
Again, by triangle inequality, under Assumptions \ref{A:beta} and \ref{E:WW}, we obtain
\begin{align}
W_{11} 
% &\leq \sup_{t \in \I} \left\|\sum_{j=m}^{n-m} \frac{\omega(t, j)}{2} (\grave {A}_{j,m}-  \acute { \Delta}_j { \beta}(t_j))\right\|_{4\kappa}
&\leq \max_{m \leq j \leq n-m}  \|\grave{A}_{j,m} - \acute { \Delta}_j { \beta}(t_j)\|_{4\kappa}, \\
&\leq \frac{1}{m}\max_{m \leq j \leq n-m} \left(\sum_{i = j-m+1}^{j+m} \left\|x_i x_i^{\T} - x_{i+m} x_{i+m}^{\T}\right\|_{8\kappa}\|x_i x_i^{\T}\|_{8\kappa} | { \beta}(t_i) - { \beta}(t_j)| \right) \\ &= O(m/n). \label{eq:W11}
\end{align}
% where under \cref{A:beta} and Assumption  \ref{E:WW}, 
% \begin{align}
% \max_{m \leq j \leq n-m} \|\grave{A}_{j,m} - \acute { \Delta}_j { \beta}(t_j)\|_{4\kappa}
%  &\leq \frac{1}{m}\max_{m \leq j \leq n-m} \left\{\sum_{i = j-m+1}^{j+m} \left\|x_i x_i^{\T} - x_{i+m} x_{i+m}^{\T}\right\|_{8\kappa}\|x_i x_i^{\T}\|_{8\kappa} | { \beta}(t_i) - { \beta}(t_j)| \right\}\\ 
%  &= O(m/n).
% \end{align}
% Therefore, we obtain from \eqref{eq:W11}, 
% \begin{align}
% W_{11} = O(m/n).
% \end{align}
Under Assumption \ref{E:WW}, we have $\max_{m \leq j \leq n-m} \|\acute { \Delta}_j\|_{4\kappa} = O(1).$ Then, by similar arguments in Lemma 3 of \citep{zhou2010simultaneous} and the continuity of $M^{+}(t)$, $m\tau_n^{3/2} \to \infty$, we obtain  
\begin{align}
 W_{12} 
% & \leq \sup_{t\in \I}\sum_{j=m}^{n-m} \frac{ \omega(t, j)|{ \beta}(t_j) -{ \beta}(t)|}{2}\max_{m \leq j \leq n-m}\|\acute { \Delta}_j\|_{4\kappa} = O(\tau_n).
& = \sup_{t \in \I} \left\|\sum_{j=m}^{n-m} \frac{\acute { \Delta}_j  \tilde \omega(t, j)}{2} \{{ \beta}^{\prime}(t)(t_j - t) + O(\tau_n^3)\}  \right\|_{4\kappa}\\ 
& \leq \sup_{t \in \I} \left\|\sum_{j=m}^{n-m} \{\acute { \Delta}_j/2 - M^{+}(t_j)\} \tilde \omega(t, j) { \beta}^{\prime}(t)(t_j - t) \right\|_{4\kappa}\\ &+ \sup_{t \in \I} \left\|\sum_{j=m}^{n-m}  M^{+}(t_j)\tilde \omega(t, j) { \beta}^{\prime}(t)(t_j - t)  \right\|_{4\kappa} + O(\tau_n^3) \\ 
& = \sup_{t \in \I} \left\|\sum_{j=m}^{n-m}  \tilde \omega(t, j) { \beta}^{\prime}(t)M^{+}(t)(t_j - t)  \right\|_{4\kappa} + O\{\tau_n^{3/2}(m^{-1/2} + m/n) + \tau_n^3\} \\ 
&= O\left(\tau_n^3\right).
\label{eq:W12}
\end{align}

Therefore, combining \eqref{eq:W11}  and \eqref{eq:W12}, since $n\tau_n^3 \to \infty$, $m/(n\tau_n^3) \to 0$, we have
\begin{align}
\sup_{t \in \I} \| W_1(t) - \Omega(t) { \beta}(t)\|_{4\kappa}  = O\left(\tau_n^3 \right).\label{eq:W1}
\end{align}

To proceed, we shall show that 
\begin{align}
\sup_{t \in \I} \|W_2(t)\|_{4\kappa} = O\{(n \tau_n^{3/2})^{-1/2} + \chi^m\}.\label{eq:W2}
\end{align} 
Under Assumption \ref{E:WW}, we have
% following similar arguments in \cref{Propo-4-24-1},
\begin{align}
    \lt|\E(\grave { \Delta}_j)\rt| 
   & =  \frac{1}{m}\left|\sum_{i=j-m+1}^{j} \E\lt(x_i x_i^{\T}x_{i+m} e_{i+m}\rt)+ \E\lt(x_{i+m} x_{i+m}^{\T}x_i e_i\rt) \right| 
     = O(\chi^m).\label{eq:W2E}
\end{align}
Then, by Burkholder's inequality, for a sufficiently large $C$, we have
\begin{align}
\left\|  \sum_{j=m}^{n-m} \frac{\tilde \omega(t, j)}{2} \lt\{\grave { \Delta}_j -  \E( \grave { \Delta}_j) \rt\} \right\|_{4\kappa} &= \left\| \sum_{s = -m}^{\infty}\sum_{j=m}^{n-m} \frac{\tilde \omega(t, j)}{2} \proj_{j-s} \grave{ \Delta}_j \right\|_{4\kappa}\\
 &\leq C  \sum_{s = -m}^{\infty} \left\{ \sum_{j=m}^{n-m} \frac{ \tilde \omega^2(t, j)}{4} \lt\|\proj_{j-s} \grave{ \Delta}_j \rt\|_{4\kappa}^2\right\}^{1/2}.\label{eq:W2-1}
\end{align}
Under Assumptions \ref{E:WW} and \ref{B:H_delta}, using similar techniques in Lemma 3 of \citep{zhou2010simultaneous}, we have
\begin{align}
\lt\|\proj_{j-s} \grave{ \Delta}_j \rt \|_{4\kappa}&\leq \lt\|\grave{ \Delta}_j  - \grave{ \Delta}_{j,\{j-s\}} \rt\|_{4\kappa} 
 \leq \frac{1}{m}\sum_{i=j-m+1}^{j+m} \{\delta_{8\kappa}(J, i-j+s) +\delta_{8\kappa}(U, i-j+s) \}.
\label{eq:grave_delta}
\end{align}
Then, \eqref{eq:W2} follows from \eqref{eq:W2E}, \eqref{eq:W2-1} and \eqref{eq:grave_delta}.
% that 
% \begin{align}
% \sup_{t \in \I}\| W_{2} \|_{4\kappa}  = O((n \tau_n)^{-1/2}). \label{eq:W2}
% \end{align}
Combining \eqref{eq:W1} and \eqref{eq:W2}, since $m = O(n^{1/3})$, we have \eqref{eq:betabreve}.
Hence, by \eqref{eq:fixtA}, under conditions $m = O(n^{1/3})$, we have
\begin{align}
 \sup_{t \in \I}\left\| \{\breve { \Sigma}(t) - { \Sigma}^A (t)\}1(A_n) \right\|_{\kappa}= O\left(\surd{m}\tau_n^3+ \surd{\frac{m}{n\tau_n^{3/2}}}\right).
\end{align}
%  By the chaining argument in Proposition B.1 in Section B.2 in \citep{dette2018change}, we have
%  \begin{align}
%   \left \| \sup_{t \in \I}| (\breve { \Sigma}(t) - { \Sigma}^A (t))1(A_n)|  \right\|_{\kappa}= O\left( \surd{m\tau_n^{2-2/\kappa}}\right).
% \end{align}
% Finally, by Proposition A.1 in \citep{wu2018gradient}, we have 
% \begin{align}
%      \sup_{t \in \I} \lt| \breve { \Sigma}(t) - { \Sigma}^A (t)\rt|= \Op\left( \surd{m\tau_n^{2-2/\kappa}}\right).
% \end{align}
Finally, the result follows from the chaining argument in Proposition B.1 in Section B.2 in \citep{dette2018change} and Proposition A.1 in \citep{wu2018gradient}.
\end{proof}

% \subsubsection{Proof of \cref{thm:lrv_diff}}
\begin{proof}[ \cref{thm:lrv_diff}]
% \WC{Checked, please change green to dark}
% \begin{proof}
Recall the definition of $\tilde{ \Sigma}(t)$ in \eqref{eq:tildeSigma}. Observe that
    \begin{align}
      \hat{{ \Sigma}}(t) - \tilde{{ \Sigma}}(t)  
      & = \sum_{j=m}^{n-m} \frac{m \lt({ \Delta}_{j}{ \Delta}_{j}^{\T} - \tilde{ \Delta}_{j}\tilde{ \Delta}_{j}^{\T}\rt) }{2}\omega(t, j)-\breve { \Sigma}(t)\\ 
    %   & = \frac{m}{2}\sum_{j=m}^{n-m}\lt\{(A_{j,m} + \tilde { \Delta}_{j})(A_{j,m} + \tilde { \Delta}_{j})^{\T}- \tilde{ \Delta}_{j} \tilde{ \Delta}_{j}^{\T}\rt\} \omega(t, j)-\breve { \Sigma}(t)\\ 
      & =  { \Sigma}^A(t) -  \breve { \Sigma}(t) + \frac{m}{2}\sum_{j=m}^{n-m} \lt(A_{j,m}\tilde { \Delta}_{j}^{\T} + \tilde { \Delta}_{j}A_{j,m}^{\T}\rt)\omega(t, j).
    \end{align}
    Then, we have 
    \begin{align}
      \sup_{t \in \I} \lt|\hat{{ \Sigma}}(t) - \tilde{{ \Sigma}}(t)\rt| \leq   \sup_{t \in \I} \lt|{ \Sigma}^A(t) - \breve { \Sigma}(t)\rt| + m \sup_{t \in \I} \left|\sum_{j=m}^{n-m} \omega(t, j)\tilde { \Delta}_{j}A_{j,m}^{\T} \right| .\label{eq:Sigma_decomp}
    \end{align}
    Take note that $\Sigma^A(t)$ is the leading term of bias, so we introduce the correction.
    By \cref{lm:correction}, 
    we have
    \begin{align}
      \sup_{t \in \I}\lt|{ \Sigma}^A (t) -  \breve { \Sigma}(t)\rt| = \Op\left(\surd{m}\tau_n^{3-1/\kappa} + \surd{\frac{m}{n\tau_n^{3/2+2/\kappa}}}\right).\label{eq:SigmaA}
  \end{align}
  To proceed, define 
    \begin{align}
      h_{s}(t) = \sum_{j=m}^{n-m}\omega(t,j)\proj_{j-s} (\tilde { \Delta}_{j}A_{j,m}^{\T})=\sum_{j=m}^{n-m} h_{s,j}(t).
    \end{align}
    Under Assumption \ref{E:HW}, it's straightforward that $\E(\tilde { \Delta}_{j}A_{j,m}^{\T}) = 0$, for $m \leq j \leq n-m $. Then, we can write $ \sum_{j=m}^{n-m} \omega(t, j)\tilde { \Delta}_{j}A_{j,m}^{\T}$ as a summation of martingale differences, i.e.
    \begin{align}
    \sum_{j=m}^{n-m} \omega(t, j)\tilde { \Delta}_{j}A_{j,m}^{\T} = \sum_{s = -m}^{\infty}h_{s}(t),\label{eq:sumh}
    \end{align}

  Next, we shall show that $\| h_{s,j}(t)\|_{\kappa} = O\{(1/n + m^{-3/2})\min(\chi^{s-m},1)\}$.
%   Define 
%   $$\delta_4(\tilde{ \Delta}(m),k) = \sup_{1 \leq j \leq n} \| \tilde{ \Delta}_j - \tilde{ \Delta}_{j,\{j-k\}}\|_4.$$ 
  Under Assumptions \ref{Ass-E} and \ref{B:H_delta}, similar to (36) in Lemma 3 in \citep{zhou2010simultaneous}, we have
   \begin{align}
    \delta_{2\kappa}(\tilde{ \Delta}(m),k): = \sup_{1 \leq j \leq n} \| \tilde{ \Delta}_j - \tilde{ \Delta}_{j,\{j-k\}}\|_{2\kappa} &= O\left\{ \frac{1}{m}\sum_{i=-m+1}^m\delta_{2\kappa}(U,k+i)\right\}\\ &= O\{\min(\chi^{k-m}, 1)/m\}.\label{eq:lrv2}
  \end{align}

   Using similar arguments in \eqref{eq:lrv2}, by the boundedness of ${ \beta}(\cdot)$, we have
  \begin{align}
    \sup_{1 \leq j \leq n} \left\|A_{j,m} - A_{j,m,\{j-k\}} \right\|_{2\kappa}
    %   &\leq \frac{1}{m}\sup_{1 \leq j \leq n} \left\|\sum_{i=j-m+1}^{j+m} x_i x_i^{\T} { \beta}(t_i) - x_{i,\{j-k\}} x_{i,\{j-k\}}^{\T} { \beta}(t_i)  \right\| _4\\ 
    & \leq \frac{M }{m}\sum_{i=-m+1}^m\delta_{4\kappa}(W,k+i) = O\{\min(\chi^{k-m}, 1)/m\}, \label{eq:lrv3}
  \end{align}
  where $M$ is a sufficiently large positive constant.
   Following similar arguments in Theorem 1 of \citep{wu2007strong}, $\sup_j \|\tilde{ \Delta}_j\|_{2\kappa} = O(m^{-1/2})$. Similar to \eqref{eq:Ajm}, we have uniformly for $m \leq j\leq n- m$, 
   \begin{align}
   \| A_{j,m}\|_{2\kappa} = O(m/n + m^{-1/2}).\label{eq:lrv1}
   \end{align}
  Under Assumption \ref{Ass-U}, from \eqref{eq:lrv2}, \eqref{eq:lrv3} and \eqref{eq:lrv1}, we obtain uniformly for $m \leq j \leq n-m$, $s \geq -m$,
  \begin{align}
    \|h_{s,j}\|_{\kappa}/\omega(t,j) & = \| \proj_{j-s} \{ \tilde { \Delta}_{j}A_{j,m}^{\T} \}\|_{\kappa} \notag\\ 
     &\leq \|A_{j,m,\{j-s\}} - A_{j,m}\|_{2\kappa}\lt\|\tilde{ \Delta}_j^{\T}\rt\|_{2\kappa} + \|A_{j,m,\{j-s\}}\|_{2\kappa} \lt\|\tilde{ \Delta}^{\T}_j - \tilde{ \Delta}^{\T}_{j,\{j-s\}}\rt\|_{2\kappa}\\
    & = O\{(1/n + m^{-3/2})\min(\chi^{s-m},1)\}.\label{eq:h_null}
  \end{align}
  Since $h_{s,j}(t)$ are martingale differences with respect to $j$, we have for $t \in [m/n, 1 - m/n]$,
  \begin{align}
    \|h_s(t)\|_{\kappa}^2 =\sum_{j=m}^{n-m}\|h_{s,j}(t)\|_{\kappa}^2 
    = O \left\{(n \tau_n)^{-1}(1/n^2 + m^{-3})\min(\chi^{2s-2m},1)\right\}.\label{eq:hsj}
  \end{align}
  By \eqref{eq:sumh} and \eqref{eq:hsj}, we obtain 
  \begin{align}
    &\left\|  m \sum_{j=m}^{n-m} \omega(t, j)\tilde { \Delta}_{j}A_{j,m}^{\T} \right\|_{\kappa}
    \leq m\sum_{s=-m}^{m}\| h_{s}(t)\|_{\kappa} + m\sum_{s=m+1}^{\infty} \| h_{s}(t)\|_{\kappa} 
    = O\left(\surd{\frac{m}{n \tau_n}}\right).
    \label{eq:Sigmat0}
  \end{align}
  By the chaining argument in Proposition B.1 in Section B.2 in \citep{dette2018change}, we have
  \begin{equation}
    \sup _{t \in I} \left| m \sum_{j=m}^{n-m} \omega(t, j)\tilde { \Delta}_{j}A_{j,m}^{\T} \right| = \Op\left(\surd{\frac{m}{n \tau_{n}^{1+2/\kappa}}}\right).\label{eq:Sigmat}
   \end{equation} 
  Combining \eqref{eq:Sigma_decomp},  \eqref{eq:SigmaA} and \eqref{eq:Sigmat}, we obtain
  \begin{equation}
    \sup _{t \in I} \lt| \hat{{ \Sigma}}(t)-\tilde{{ \Sigma}}(t) \rt| = \Op\left(\surd{m}\tau_n^{3-1/\kappa} +  \surd{\frac{m}{n\tau_n^{3/2+2/\kappa}}}\right).
    \label{eq:Sigmatsup}
  \end{equation}
 By Lemma 3 in \citep{zhou2010simultaneous}, Proposition B.1 in Section B.2 in \citep{dette2018change}, we have 
  \begin{equation}
   \sup _{t \in I} \lt| \tilde{{ \Sigma}}(t)- {E}\tilde{{ \Sigma}}(t)\rt|=\Op\left(\surd{\frac{m}{n\tau_n^{1+2/\kappa}}}\right). 
    \label{eq:Sigmat0sup}
  \end{equation}
  Using similar techniques in Lemma 4 and Lemma 5 of \citep{zhou2010simultaneous}, which hold uniformly for $t\in(0,1)$,  \eqref{eq:Sigmat0sup} leads to 
  \begin{equation}
    \sup _{t \in I} \lt| \tilde{{ \Sigma}}(t) -{ \Sigma}(t)\rt|  = 
    \Op\left(\surd{\frac{m}{n\tau_n^{3/2+2/\kappa}}} + \frac{1}{m}+\surd{m}\tau_n^{3-1/\kappa}\right),
    \label{eq:Sigmatsup0}
  \end{equation}
  With \eqref{eq:Sigmatsup} and \eqref{eq:Sigmatsup0}, the supreme bound is thus proved. \end{proof}
  \subsection{Proof of \texorpdfstring{\cref{thm:cp}}{Proof of Theorem 2}}
  Let $\mathcal C_n = \{j:j- n \tau_n^{3/2}  \leq i \leq j + n \tau_n^{3/2}, \beta^{\prime\prime\prime}(t_i) ~\text{exists and continuous} \}$, $\mathcal C_m = \{j:j- m \leq i \leq j + m, \beta^{\prime\prime\prime}(t_i) ~\text{exists and continuous} \}$. Recall in \cref{lm:bias} $A_n = \{ \sup_{t \in  \I} |{ \Omega}(t) - M^+(t)| \leq l_n \{(m\tau_n^{3/2})^{-1/2} + m/(n\tau_n^{3/4}) + \tau_n^{3/4}\}\}$, and $\lim_{n\to\infty} \pp(A_n) = 1.$
\begin{lemma}\label{lm:beta_cp}
      Under the conditions of \cref{thm:cp},  we have 
      \begin{align}
           \sup_{j \in  \mathcal C_m}\sup_{m_1 = 1,\ldots, m} \|\{ \check \beta(t_{j+m}
      ) - \check { \beta}(t_j)\}1(A_n) \|_{4\kappa} = O\left\{ \frac{mq_n}{n \tau_n^{3/2}} + (n \tau_n^{3/2})^{-1/2}\right\}.
      \end{align}
  \end{lemma}
 \begin{proof}[Proof of \cref{lm:beta_cp}]
  Following the arguments in \cref{lm:correction}, since $\Omega(t)$ is invertible and continuous on $A_n$, and $\sup_{j \in  \mathcal C_m}\|\varpi(t_{j})\|_{8\kappa} = O(1)$, we have 
  \begin{align}
      &\sup_{j \in  \mathcal C_m} \sup_{m_1 = 1,\ldots, m} \|\{ \check \beta(t_{j+m_1}
      ) - \check { \beta}(t_j)\}1(A_n) \|_{4\kappa}\\ &= \sup_{j \in  \mathcal C_m} \sup_{m_1 = 1,\ldots, m} \|\{\Omega^{-1}(t_{j+m_1}) \varpi(t_{j+m_1}) - \Omega^{-1}(t_{j}) \varpi(t_{j}) 
    \}1(A_n) \|_{4\kappa} \\ 
      % & \leq \sup_{j \in  \mathcal C_m} \|\Omega^{-1}(t_{j+m}) (\varpi(t_{j+m}) -  \varpi(t_{j}) )1(A_n) \|_{4\kappa} \\ 
      % & + \sup_{j \in  \mathcal C_m} \|\Omega^{-1}(t_{j+m})(\Omega(t_{j+m}) - \Omega(t_{j})) \Omega^{-1}(t_{j}))  \varpi(t_{j}) 1(A_n) \|_{4\kappa}\\ 
      & \leq \sup_{j \in  \mathcal C_m} \sup_{m_1 = 1,\ldots, m}\sqrt{p}\|\rho\{\Omega^{-1}(t_{j+m_1})\} |\varpi(t_{j+m_1}) -  \varpi(t_{j}) |1(A_n) \|_{4\kappa} \\ 
      & + \sup_{j \in  \mathcal C_m} \sup_{m_1 = 1,\ldots, m}\sqrt{p}\|\rho\{\Omega^{-1}(t_{j+m_1})\} |\Omega(t_{j+m_1}) - \Omega(t_{j})| \rho\{\Omega^{-1}(t_{j}) \} \varpi(t_{j}) 1(A_n) \|_{4\kappa}\\ 
      & \leq C_1 \sup_{j \in  \mathcal C_m}\sup_{m_1 = 1,\ldots, m} \| |\varpi(t_{j+m_1}) -  \varpi(t_{j}) |1(A_n) \|_{4\kappa} + C_2 m/n,\label{eq:newdiff}
  \end{align}
  where $C_1$ and $C_2$ are sufficiently large constants.\par
By \eqref{eq:newdiff}, it is sufficient to show 
  \begin{align}
      \sup_{r \in  \mathcal C_m}\sup_{m_1 = 1,\ldots, m} \| \{\varpi(t_{r+m_1}) -  \varpi(t_{r}) \}1(A_n) \|_{4\kappa} = O\{(mq_n)/(n \tau_n^{3/2}) + (n\tau_n^{3/2})^{-1/2} \}.\label{eq:varpi_cp}
  \end{align}
By triangle inequality, we have
  \begin{align}
        &\sup_{r \in  \mathcal C_m}\sup_{m_1 = 1,\ldots, m} \| \{\varpi(t_{r+m_1}) -  \varpi(t_{r}) \}1(A_n) \|_{4\kappa} \\ 
         & \leq \sup_{r \in  \mathcal C_m}\sup_{m_1 = 1,\ldots, m} \left\| \left\{{ \varpi}(t_{r + m_1}) - \sum_{j=m}^{n-m} \frac{\acute \Delta_j \tilde \omega(t_{r + m_1},j)}{2} \beta(t_j) \right\} 1(A_n) \right\|_{4\kappa}\\
         &+ \sup_{r \in  \mathcal C_m}\sup_{m_1 = 1,\ldots, m} \left\| \left\{{ \varpi}(t_{r}) - \sum_{j=m}^{n-m} \frac{\acute \Delta_j \tilde 
 \omega(t_{r},j)}{2} \beta(t_j) \right\} 1(A_n) \right\|_{4\kappa}\\
        & +  \sup_{r \in \mathcal C_m} \sup_{m_1 = 1,\ldots, m} \left\|\left[\sum_{j=m}^{n-m} \acute \Delta_j \{\tilde  \omega(t_r,j) - \tilde 
 \omega(t_{r+m_1},j)\}  \beta(t_j)/2 \right]1(A_n)\right\|_{4\kappa}\\
        & \leq 2\sup_{t \in I} \left\| \left\{{ \varpi}(t) - \sum_{j=m}^{n-m} \frac{\acute \Delta_j \tilde  \omega(t,j)}{2} \beta(t_j) \right\} 1(A_n) \right\|_{4\kappa}\\
        & +  \sup_{r \in \mathcal C_m} \sup_{m_1 = 1,\ldots, m} \left\|\left[\sum_{j=m}^{n-m} \acute \Delta_j \{\tilde  \omega(t_r,j) - \tilde \omega(t_{r+m_1},j)\}  \beta(t_j)/2 \right] 1(A_n)\right\|_{4\kappa}.\label{eq:decom_varpi}
  \end{align}
  Based on \eqref{eq:decom_varpi}, we break the proof for \eqref{eq:varpi_cp} into two parts. First, we investigate the first term in \eqref{eq:decom_varpi}. Recall the definition of 
$\grave { \Delta}_{j}$ and $\grave{A}_{j,m}$ in \eqref{eq:graveD} and \eqref{eq:graveA} and we have from \eqref{eq:varpi} that
\begin{align}
{ \varpi}(t) & = \sum_{j=m}^{n-m} \frac{ \tilde \omega(t, j)}{2} \grave{A}_{j,m} + \sum_{j=m}^{n-m} \frac{\tilde \omega(t, j)}{2} \grave { \Delta}_j =  W_1(t) + W_2(t),
\end{align}
where we directly have from \eqref{eq:W2} that 
\begin{align}
  \sup_{t \in \I} \|W_2(t)\|_{4\kappa} = O\{(n \tau_n^{3/2})^{-1/2}\}.
\end{align}
Therefore, we have
  \begin{align}
     & \sup_{t \in I} \left\|\left\{{ \varpi}(t) - \sum_{j=m}^{n-m} \frac{\acute \Delta_j \tilde  \omega(t,j)}{2} \beta(t_j) \right\} 1(A_n)\right\|_{4\kappa}\\ 
     & \leq  \sup_{t \in I} \sum_{j=m}^{n-m} \frac{ \tilde  \omega(t, j)}{2m}\left(\sum_{i = j-m+1}^{j} \left\|x_i x_i^{\T} - x_{i+m} x_{i+m}^{\T}\right\|_{8\kappa}\right.\\ & 
\times \left. \|x_i x_i^{\T}({\beta}(t_i) - { \beta}(t_j)) - x_{i+m} x_{i+m}^{\T}({\beta}(t_{i+m}) - { \beta}(t_j))\|_{8\kappa} \right) +\sup_{t \in \I} \|W_2(t)\|_{4\kappa}\\
& \leq \sup_{t \in I} \left(\sum_{j \in  \mathcal C_m, m \leq j \leq n-m} + \sum_{j \notin  \mathcal C_m, m \leq j \leq n-m}\right)\frac{\tilde  \omega(t, j)}{2m} \left(\sum_{i = j-m+1}^{j} \left\|x_i x_i^{\T} - x_{i+m} x_{i+m}^{\T}\right\|_{8\kappa}\right. \\ &  \left.
\times  \|x_i x_i^{\T}({\beta}(t_i) - { \beta}(t_j)) - x_{i+m} x_{i+m}^{\T}({\beta}(t_{i+m}) - { \beta}(t_j))\|_{8\kappa}\right)+\sup_{t \in \I} \|W_2(t)\|_{4\kappa}\\ 
     &= O\left\{\frac{m}{n} + \frac{mq_n}{n \tau_n^{3/2}} + (n \tau_n^{3/2})^{-1/2}+\chi^m\right\}.\label{eq:varpi_1}
  \end{align}
  Secondly, by the continuity of the kernel function $K(\cdot)$ and since $\sup_{j= m, \ldots, n-m}\| \acute \Delta_j\|=O(1)$, 
we have 
  \begin{align}
    \sup_{r \in \mathcal C_m} \sup_{m_1 = 1,\ldots, m} \left\|\left[\sum_{j=m}^{n-m} \acute \Delta_j \{\tilde 
 \omega(t_r,j) - \tilde 
 \omega(t_{r+m_1},j)\} \beta(t_j)/2 \right] 1(A_n)\right\|_{4\kappa} = O\{m/(n\tau_n^{3/2})\}.\label{eq:varpi_2}
  \end{align}
Then, \eqref{eq:varpi_cp} follows from \eqref{eq:varpi_1} and \eqref{eq:varpi_2}. The lemma follows from \eqref{eq:newdiff} and \eqref{eq:varpi_cp}.
\end{proof}

  \begin{proof}[Proof of Theorem 2]
For sufficiently large constants $C_1$ and $C_2$, we have
  \begin{align}
    & \sup_{t \in \I} \left|\{\breve { \Sigma}(t) - { \Sigma}^A (t)\}1(A_n)\right|\\
    % &\leq \sup_{t \in \I} \sum_{j=m}^{n-m}\frac{m\omega(t,j)}{2} \lt\|(\hat{A}_{j,m} -  A_{j,m})1(A_n)\rt\|_{2\kappa}\left(\lt\|(\hat{A}_{j,m} -  A_{j,m})1(A_n)\rt\|_{2\kappa} + 2\|A_{j,m}\|_{2\kappa}\right) \\ 
    % & \leq \sup_{t \in \I} \left(\sum_{j\in \mathcal C_m, m \leq j \leq n-m} + \sum_{j \notin \mathcal C_m, m \leq j \leq n-m}\right) \frac{m\omega(t,j)}{2} \lt\|(\hat{A}_{j,m} -  A_{j,m})1(A_n)\rt\|_{2\kappa} \\ & \times \left(\lt\|(\hat{A}_{j,m} -  A_{j,m})1(A_n)\rt\|_{2\kappa} + 2\|A_{j,m}\|_{2\kappa}\right)\\
    % & \leq \sup_{t \in \I} \left(\sum_{j\in \mathcal C_n, m \leq j \leq n-m} + \sum_{j\notin \mathcal C_n,  j\in \mathcal C_m m \leq j \leq n-m} + \sum_{j \notin \mathcal C_m, m \leq j \leq n-m}\right) \frac{m\omega(t,j)}{2} \lt\|(\hat{A}_{j,m} -  A_{j,m})1(A_n)\rt\|_{2\kappa} \\ & \times \left(\lt\|(\hat{A}_{j,m} -  A_{j,m})1(A_n)\rt\|_{2\kappa} + 2\|A_{j,m}\|_{2\kappa}\right).
    &\leq \sup_{t \in \I} \left| \sum_{j\in \mathcal C_n, m \leq j \leq n-m}\frac{m\omega(t,j)}{2} (\hat{A}_{j,m}\hat{A}_{j,m}^{\T} - A_{j,m} A_{j,m}^{\T}) 1(A_n) \right| \\ &+ C_1 \sum_{j\notin \mathcal C_n,  j\in \mathcal C_m m \leq j \leq n-m}\frac{m}{n\tau_n} | (\hat{A}_{j,m}\hat{A}_{j,m}^{\T} - A_{j,m} A_{j,m}^{\T})1(A_n)| \\ 
    &+ C_2 \sum_{j \notin \mathcal C_m, m \leq j \leq n-m} \frac{m}{n\tau_n} | (\hat{A}_{j,m}\hat{A}_{j,m}^{\T} - A_{j,m} A_{j,m}^{\T})1(A_n)|.
    \label{eq:fixAcp}
  \end{align}
 For  $j \in \mathcal C_m$, by \eqref{eq:lrv1}, we have $\| A_{j,m} \|_{2\kappa} = O(m^{-1/2}+m/n)$. 
We shall  show that
\begin{align}
 \sup_{j \in  \mathcal C_m} \lt\|(\hat{A}_{j,m} -  A_{j,m})1(A_n)\rt\|_{2\kappa} = O(m^{-1/2}),\label{eq:jcm0}
 \end{align}
 and
 \begin{align}
    \sup_{j \in  \mathcal C_n}\|(\hat{  A}_{j,m} -   A_{j,m})1(A_n) \|_{2\kappa} = O\{(mq_n)/(n \tau_n^{3/2}) + (n \tau_n^{3/2})^{-1/2} + m^{-1/2}\tau_n^3\}.\label{eq:jcm1}
\end{align}
By \cref{lm:beta_cp}, we have 
  \begin{align}
      & \sup_{j \in  \mathcal C_m}\|(\hat{  A}_{j,m} -   A_{j,m})1(A_n) \|_{2\kappa}\\  
      & \leq \sup_{j \in  \mathcal C_m}\left\|m^{-1}\sum_{i=j-m+1}^j x_i x_i^{\T}\left[\{\check \beta(t_i)-\check \beta(t_{i+m})\} - \{\beta(t_i)- \beta(t_{i+m})\} \right] 1(A_n) \right\|_{2\kappa} \\ 
      & +  \sup_{j \in  \mathcal C_m}\left\|m^{-1}\sum_{i=j-m+1}^j ( x_i x_i^{\T} - x_{i+m} x_{i+m}^{\T}) \left\{\check \beta(t_{i+m}) -  \beta(t_{i+m}) \right\} 1(A_n) \right\|_{2\kappa} \\ 
      % & \leq C_1 \tau_n^2 + C_2 m/n + \sup_{j \in  \mathcal C_m}\left\|m^{-1}\sum_{i=j-m+1}^j ( x_i x_i^{\T} - x_{i+m} x_{i+m}^{\T}) \left(\check \beta(t_{j}) -  \beta(t_{j}) \right) 1(A_n) \right\|_{2\kappa}\\ 
      & \leq C_1 \left\{ \frac{mq_n}{n \tau_n^{3/2}}  + (n \tau_n^{3/2})^{-1/2}\right\}+ C_2 m/n + \sup_{j \in  \mathcal C_m}\left\|m^{-1}\sum_{i=j-m+1}^j ( x_i x_i^{\T} - x_{i+m} x_{i+m}^{\T})\right\|_{4\kappa} \left\| \left\{\check \beta(t_{j}) -  \beta(t_{j}) \right\} 1(A_n) \right\|_{2\kappa} \\ 
       &= O(m^{-1/2}), \label{eq:cmdiffA}
      % & = O(m^{-1/2} +m/n + \tau_n^2  + (n\tau_n)^{-1/2} ) = O(\tau_n^2+m^{-1/2} + (n\tau_n)^{-1/2}),
  \end{align}
  where the last line follows from similar arguments in proving \eqref{eq:A2}. Therefore, we have shown \eqref{eq:jcm0}.
  % Similar to the proof of \eqref{eq:Ajm}, we also have 
  % \begin{align}
  %    \sup_{j \in \mathcal C_m} \| A_{j,m} \|_{2\kappa} = O(m^{-1/2} + m/n ).\label{eq:jcm2}
  % \end{align}
Following the arguments in \eqref{eq:cmdiffA}, for the indices $j \in \mathcal C_n \subset \mathcal C_m$, by \eqref{eq:betabreve}, we have 
\begin{align}
     \sup_{j \in  \mathcal C_n}\|(\hat{  A}_{j,m} -   A_{j,m})1(A_n) \|_{2\kappa} = O\{(mq_n)/(n \tau_n^{3/2}) + (n \tau_n^{3/2})^{-1/2} + m^{-1/2} \tau_n^3 \}.
\end{align}
For the first term of \eqref{eq:fixAcp}, note that  
\begin{align}
    &\left\| \sum_{j\in \mathcal C_n, m \leq j \leq n-m}\frac{m\omega(t,j)}{2} (\hat{A}_{j,m}\hat{A}_{j,m}^{\T} - A_{j,m} A_{j,m}^{\T}) 1(A_n) \right\|_{\kappa}\\ &\leq \sum_{j\in \mathcal C_n, m \leq j \leq n-m} \frac{m\omega(t,j)}{2} \lt\|(\hat{A}_{j,m} -  A_{j,m})1(A_n)\rt\|_{2\kappa} \\ & \times \left(\lt\|(\hat{A}_{j,m} -  A_{j,m})1(A_n)\rt\|_{2\kappa} + 2\|A_{j,m}\|_{2\kappa}\right) \\ 
    & = O\left(\frac{m^{3/2}q_n}{n \tau_n^{3/2}} + \tau_n^3 + \surd\frac{m}{n\tau_n^{3/2}}\right).\label{eq:firstterm}
\end{align}
% Then, by \eqref{eq:fixAcp}, \eqref{eq:jcm0}, and 
%  \eqref{eq:jcm1},  we have 
%  \begin{align}
%      \sup_{t \in \I} \lt\| (\breve { \Sigma}(t) - { \Sigma}^A (t))1(A_n)\rt\|_{\kappa} &= O\left(\frac{m^2 q_n}{n \tau_n} + \tau_n^{1/2} q_n + \frac{m^{3/2}q_n}{n \tau_n^{3/2}} + \tau_n^3 + \surd\frac{m}{n\tau_n^{3/2}} \right) \\ 
%     & = O\left(\frac{m^2 q_n}{n \tau_n} + \tau_n^{1/2} q_n  + \surd\frac{m}{n\tau_n^{3/2}}+ \frac{m^{3/2}q_n}{n \tau_n^{3/2}} \right) .
%  \end{align}
By the chaining argument of Proposition B.1 of \citep{dette2018change}, by we have 
\begin{align}
    \sup_{t \in \I} \left| \sum_{j\in \mathcal C_n, m \leq j \leq n-m}\frac{m\omega(t,j)}{2} (\hat{A}_{j,m}\hat{A}_{j,m}^{\T} - A_{j,m} A_{j,m}^{\T}) 1(A_n) \right| = \Op\left(\surd\frac{m}{n\tau_n^{3/2 + 2/\kappa}}  + \tau_n^{3-1/\kappa} + \frac{m^{3/2}q_n}{n \tau_n^{3/2+1/\kappa}}\right).
\end{align}
For the second term of \eqref{eq:fixAcp}, by \eqref{eq:jcm0}, we have 
\begin{align}
    &\sum_{j\notin \mathcal C_n,  j\in \mathcal C_m m \leq j \leq n-m}\frac{m}{n\tau_n} \E | (\hat{A}_{j,m}\hat{A}_{j,m}^{\T} - A_{j,m} A_{j,m}^{\T})1(A_n)|\\ & \leq \sum_{j\notin \mathcal C_n,  j\in \mathcal C_m m \leq j \leq n-m}\frac{m}{n\tau_n} \lt\|(\hat{A}_{j,m} -  A_{j,m})1(A_n)\rt\|  \left(\lt\|(\hat{A}_{j,m} -  A_{j,m})1(A_n)\rt\| + 2\|A_{j,m}\|\right)\\ 
    & = O(q_n \tau_n^{1/2}). \label{eq:secondterm}
\end{align}
Since for $j \notin \mathcal C_m$,  $\| A_{j,m} \|_{2\kappa} = O(1)$ and $\lt\|(\hat{A}_{j,m} -  A_{j,m})1(A_n)\rt\|_{2\kappa} = O(1)$, for the third term of \eqref{eq:fixAcp}, we have 
\begin{align}
    \sum_{j \notin \mathcal C_m, m \leq j \leq n-m} \frac{m}{n\tau_n} \E | (\hat{A}_{j,m}\hat{A}_{j,m}^{\T} - A_{j,m} A_{j,m}^{\T})1(A_n)| =  O\left( \frac{m^2q_n}{n\tau}\right).\label{eq:thirdterm}
\end{align}
Combining \eqref{eq:fixAcp}, \eqref{eq:firstterm}, \eqref{eq:secondterm} and \eqref{eq:thirdterm}, 
since $\kappa \geq 1$, $m\tau_n^3 \to \infty$, 
% $m\tau_n^{2/3} \to \infty$, $m = O(n^{1/3})$, $m/(n\tau_n^3) \to \infty$,  
we have 
\begin{align}
\sup_{t \in \I}\lt| \breve { \Sigma}(t) - { \Sigma}^A (t)\rt| 
% & = \Op\left( \surd{m}\tau_n^{2-1/\kappa} + \surd{\frac{m}{n\tau_n^{1+2/\kappa}}} + \frac{m^2q_n}{n\tau^{1+1/\kappa}} \right) \\ 
& = \Op\left(q_n\tau_n^{1/2} + \surd\frac{m}{n\tau_n^{3/2 + 2/\kappa}}  + \frac{m^2q_n}{n\tau_n} + \tau_n^{3-1/\kappa} + \frac{m^{3/2}q_n}{n \tau_n^{3/2+1/\kappa}}\right)\\ 
& = \Op(q_n\tau_n^{1/2} + \surd\frac{m}{n\tau_n^{7/2}}  + \frac{m^2q_n}{n\tau_n} + \tau_n^{2} ) \label{eq:cpcorrect}
% \\ & = \Op \left(q_n\tau_n^{1/2}  + \surd\frac{m}{n\tau_n^{3}} \right). 
% \\
% & = \Op\left( \surd{m}\tau_n^{2-1/\kappa} + n^{-1/3} \tau_n^{-1/2-1/\kappa} + \frac{m q_n}{n\tau_n^{3}} \right). \label{eq:cpcorrect}
\end{align}
We proceed to investigate \cref{thm:lrv_diff} in the presence of change points in $\beta(t)$. Recall that from \eqref{eq:Sigma_decomp},
    \begin{align}
      \sup_{t \in \I} \lt|\hat{{ \Sigma}}(t) - \tilde{{ \Sigma}}(t)\rt| \leq   \sup_{t \in \I} \lt|{ \Sigma}^A(t) - \breve { \Sigma}(t)\rt| + m \sup_{t \in \I} \left|\sum_{j=m}^{n-m} \omega(t, j)\tilde { \Delta}_{j}A_{j,m}^{\T} \right|. \label{eq:cpdecomp}
    \end{align}
    
Inspecting \eqref{eq:hsj}, we have 
\begin{align}
    \|h_s(t)\|_{\kappa}^2 &= \sum_{j \in \mathcal C_m, m \leq j \leq n-m} \| h_{s,j}(t)\|_{\kappa}^2  + \sum_{j \notin \mathcal C_m, m \leq j \leq n-m} \| h_{s,j}(t)\|_{\kappa}^2\\  & = O\left\{(n\tau_n)^{-1}(1/n^2 + m^{-3})\min (\chi^{2s-2m}, 1) + \frac{q_n}{m(n\tau_n)^2}\min (\chi^{2s-2m}, 1) \right\}.
\end{align}
Similar to  \eqref{eq:Sigmat0}, we have 
\begin{align}
    \left\|  m \sum_{j=m}^{n-m} \omega(t, j)\tilde { \Delta}_{j}A_{j,m}^{\T} \right\|_{\kappa}
    &\leq m\sum_{s=-m}^{m}\| h_{s}(t)\|_{\kappa} + m\sum_{s=m+1}^{\infty} \| h_{s}(t)\|_{\kappa} 
    \\ &= O\left\{\surd{\frac{m}{n \tau_n}} +  \surd{\frac{q_n m^3}{(n \tau_n)^2}}\right\} = O\left(\surd{\frac{m}{n \tau_n}}\right).
  \end{align}
   By the chaining argument in Proposition B.1 in Section B.2 in \citep{dette2018change}, similar to \eqref{eq:Sigmat}, we have
  \begin{equation}
    \sup _{t \in I} \left| m \sum_{j=m}^{n-m} \omega(t, j)\tilde { \Delta}_{j}A_{j,m}^{\T} \right| = \Op\left(\surd{\frac{m}{n \tau_{n}^{1+2/\kappa}}}\right). \label{eq:cpmain}
   \end{equation} 
  Finally, combining \eqref{eq:cpcorrect}, \eqref{eq:cpdecomp}, and \eqref{eq:cpmain}, following similar argument of \cref{thm:lrv_diff}, we have 
  shown the first part of the results. The second part follows from \cref{lm:davis} and Theorem 3.4 of \citep{wu2018gradient}. \end{proof}

\subsection{Proof of \texorpdfstring{\cref{prop:stat_alt}}{Proof of Proposition 1}}
\begin{proof}[Proof of \cref{prop:stat_alt}]
\cref{prop:stat_alt} follows from similar but easier arguments of Theorem 3.1 and  Theorem 3.2 of \citep{wu2018gradient}.
\end{proof}

\subsection{Proof of \texorpdfstring{\cref{lm:Sigma_d}}{Proof of Theorem 3}}\label{subsec:lrv_fix}
For  the clarity of presentation, write $\hat{ \Sigma}_d(t)$ for $\hat{ \Sigma}(t)$ defined in \eqref{eq:diff_correct} under the fixed alternative.  Define $${ \Sigma}_d(t) = \kappa_2(d)\sigma_H^2(t)  \mu_W(t) \mu^{\T}_W(t),\quad t \in [0,1].$$
 For the quantities under the fixed alternatives, let $\tilde{Q}^{(d)}_{k,m} = \sum_{i=k}^{k+m-1} x_ie_i^{(d)}$,  
    $$
            \tilde{ \Delta}_{j}^{(d)}=\frac{\tilde{Q}^{(d)}_{j-m+1, m}- \tilde{Q}^{(d)}_{j+1, m}}{m}
 , \quad \tilde{{ \Sigma}}_d(t)=\sum_{j=m}^{n-m} \frac{m \tilde { \Delta}_{j}^{(d)}\tilde { \Delta}^{(d), \T}_{j}}{2}\omega(t, j), $$
  and $$  \grave { \Delta}_{j}^{(d)} = \frac{1}{m}\sum_{i=j-m+1}^{j} (x_i x_i^{\T}-x_{i+m} x_{i+m}^{\T}) (x_i e_i^{(d)}-x_{i+m} e_{i+m}^{(d)})
 .$$ Let $\hat{A}^{(d)}$, $ \varpi^{(d)}(\cdot)$, $\breve { \beta}^{(d)}(\cdot)$,  $\breve { \Sigma}_d(\cdot)$ denote the counterparts of $\hat{A}$, $ \varpi(\cdot)$, $\breve { \beta}(\cdot)$ and $\breve { \Sigma}(\cdot)$  in  \eqref{eq:diff_correct} under the fixed alternatives. Define 
 $ H^{(d)} (t_i, \mathcal{F}_i) = \sum_{k=0}^{\infty} \psi_k H(t_{i-k},  \mathcal{F}_{i-k}) $, 
  $U^{(d)}(t,\FF_i) = W(t,\FF_i) H^{(d)}(t,\FF_i)$. 
The following proposition is from Theorem 4.2 of \citep{bai2021}.
 \begin{proposition}\label{prop:5.2}
      %H_lrv is not necessary
       Under Assumptions \ref{Ass-W} and \ref{assumptionHp},  we have
       % for some $0 <\gamma < 1$, 
       \begin{align}
         \max_{\lf nb_n \rf + 1 \leq r \leq n- \lf  nb_n \rf }\left|\sum_{ i=\lf nb_n \rf+1}^r x_{i,n} e_{i,n}^{(d)}-\sum_{ i=\lf nb_n \rf+1}^r  \mu_W(t_i) e_{i,n}^{(d)}\right| 
         &= \Op (\surd{n}(\log n)^d).
         \end{align}
     \end{proposition}
   
\begin{corollary}\label{cor:5.2}
Under the conditions of \cref{prop:5.2}, assume $l /\log n \to \infty$, $l/n \to 0$, we have
\begin{align}
  \max_{1 \leq k \leq n - l + 1}\left\|\sum_{i=k}^{k+l-1} (x_i - \mu_W(t_i)) e_i^{(d)}\right\| = O(\surd{l}(\log n)^d).
\end{align}
\end{corollary}
\begin{proof}[Proof of \cref{cor:5.2}]
  The corollary follows from a careful check of the proof of Theorem 4.2 of \citep{bai2021}.
\end{proof}

% \subsubsection{Asymptotic behavior of the bias correction term under the fixed alternatives}
  \begin{lemma}\label{lm:correctiond}
%Under Assumptions \ref{A:beta}, \ref{Ass-W}, \ref{Ass-E}, suppose $\kappa  \geq \max \{4, 2/(1-2d)\}$, $m/(n\tau_n^2) \to 0$, $m = O(n^{1/3})$,  $\surd m \tau_n^{2-1/\kappa} \to 0$, $m \to \infty$, $n \tau_n^3 \to \infty$, we have 
Under the conditions of  \cref{lm:Sigma_d}, we have
\begin{align}
\sup_{t \in \I}\lt| \breve { \Sigma}_d(t) - { \Sigma}^A (t)\rt|= \Op\left[\surd{m} \{(n\tau_n^{3/2})^{d-1/2} + \tau_n^3\}\tau_n^{-1/\kappa} +  m \{(n\tau_n^{3/2})^{2d-1}+\tau_n^6\}\tau_n^{-1/\kappa}\right] = \op(m^{2d}).
\end{align}
\end{lemma}

    \begin{proof}[Proof of \cref{lm:correctiond}]
After a careful check of the proof of \cref{lm:correction}, the behavior of $W_1(t)$ is unchanged under the fixed alternatives and  it's sufficient to show that $W_2^{(d)}(t)$, $W_2(t)$ under the fixed alternatives, s.t.
    \begin{align}
    \sup_{t \in \I} \lt\|W_2^{(d)}(t)\rt\|_{4\kappa} = O\{(n\tau_n^{3/2})^{d-1/2}\}.\label{eq:W2d}
    \end{align} 
Then the lemma will follow from the similar steps in \cref{lm:correction}.
Similar to \eqref{eq:grave_delta}, under Assumptions \ref{Ass-W} and \ref{B:H_delta}, following similar arguments in \cref{lm:delta_xed}, we have
\begin{align}
\lt\|\proj_{j-s} \grave{ \Delta}_j^{(d)} \rt\|_{4\kappa}
 &\leq \frac{1}{m}\sum_{i=j-m+1}^{j+m} \lt\{\delta_{8\kappa}(J, i-j+s)+\delta_{8\kappa}(U^{(d)}, i-j+s) \rt\} = O\left\{\frac{1}{m}\sum_{i=-m+1}^m (i+s)^{d-1} \right\}.
\label{eq:grave_deltad}
\end{align}
Let $N = N_n =\lf  n \tau_n^{3/2}\rf$. Note that under Assumption \ref{E:HW}, we can write 
\begin{align}
    W_2^{(d)}(t) &= \sum_{s=0}^N \sum_{j=m}^{n-m}\frac{\tilde \omega(t,j)}{2} \proj_{j-s} \grave { \Delta}_j^{(d)} + \sum_{j=m}^{n-m} \frac{\tilde \omega(t,j)}{2} \sum_{s=N+1}^{\infty} \proj_{j-s} \grave { \Delta}_j^{(d)} = W_{21}^{(d)}(t) + W_{22}^{(d)}(t).
\end{align}
Since $\proj_{j-s}\grave { \Delta}_j^{(d)}$ are martingale differences with respect to $j$, for $0 \leq s\leq N$, we have
\begin{align}
     \sup_{t \in \I} \left\|\sum_{j=m}^{n-m}\frac{\tilde \omega(t,j)}{2} \proj_{j-s} \grave { \Delta}_j^{(d)} \right\|^2_{4\kappa} &= \sum_{j=m}^{n-m}\frac{\tilde \omega^2(t,j)}{4} \lt\|\proj_{j-s} \grave { \Delta}_j^{(d)} \rt\|^2_{4\kappa}\\ & = O\left[\frac{1}{n\tau_n^{3/2} m^2}\left\{\sum_{i=-m+1}^{m} (i+s)^{d-1}\right\}^2 \right].\label{eq:W21d1}
\end{align}
Therefore, by \eqref{eq:W21d1}, we have 
\begin{align}
     \sup_{t \in \I} \left\| W_{21}^{(d)} (t) \right\|_{4\kappa} = O(N^{d}/\surd{n\tau_n^{3/2}}) = O\{(n\tau_n^{3/2})^{d-1/2}\}.\label{eq:W21d}
\end{align}
 Since $\proj_{j-s} \grave { \Delta}_j^{(d)}$ are martingale differences with respect to $s$, elementary calculations shows
\begin{align}
    \left\|\sum_{s=N+1}^{\infty} \proj_{j-s} \grave { \Delta}_j^{(d)} \right\|_{4\kappa}^2  = O\left[m^{-2} \sum_{s=N+1}^{\infty}\left\{\sum_{i=-m+1}^{m} (i+s)^{d-1}\right\}^2 \right] = O(N^{2d-1}).\label{eq:W22d1}
\end{align}
Therefore, by  \eqref{eq:W22d1} and triangle inequality we have 
\begin{align}
     \sup_{t \in \I} \lt\|W_{22}^{(d)}(t) \rt\|_{4\kappa} = O(N^{d-1/2}) = O\{(n\tau_n^{3/2})^{d-1/2}\}.\label{eq:W22d}
\end{align}
Finally, \eqref{eq:W2d} follows from \eqref{eq:W21d} and \eqref{eq:W22d}.
% Combining   \eqref{eq:W1d} and \eqref{eq:W2d}, since $m/(n\tau_n) \to 0$, we have 
% shown \eqref{eq:betaAd1}.
% Then, by \eqref{eq:fixAd}, \eqref{eq:hatAd} and \eqref{eq:betabreved}, under bandwidth condition $m/(n\tau_n) \to 0$, we have 
% \begin{align}
% \sup_{t \in \I}\lt\|(\breve { \Sigma}_d(t) - { \Sigma}^A (t))1 (A_n)\rt\|_{\kappa} = O(m(n\tau_n)^{2d-1} + \surd{m\tau_n^2} + \surd{m} (n\tau_n)^{d-1/2}).
% \end{align}
% Finally, by Proposition B.1 in  \citep{dette2018change} and Proposition A.1 in \citep{wu2018gradient}, the lemma is proved.

\end{proof}

% \subsubsection{Proof of \texorpdfstring{\cref{lm:Sigma_d}}{Theorem 3}}
\begin{proof}[Proof of \cref{lm:Sigma_d}]
  Recall $\tilde{{Q}}_{k,m}^{(d)} = \sum_{i=k}^{k+m-1}x_i e_i^{(d)}$, and $$\tilde { \Delta}_{j}^{(d)}=\frac{\tilde{{Q}}^{(d)}_{j-m+1, m}- \tilde{{Q}}^{(d)}_{j+1, m}}{m},\quad \tilde{{ \Sigma}}_d(t)=\sum_{j=m}^{n-m} \frac{m \tilde { \Delta}_{j}^{(d)}\{\tilde { \Delta}^{(d)}_{j}\}^{\T}}{2}\omega(t, j).$$
  We break the proof into 6 steps. 

\textbf{Step 1}: 
We shall prove that under the bandwidth conditions $\kappa \geq 4/(1/2-d)$, $m/(n\tau_n^3) \to 0$, $\surd{m}\tau_n^{3-1/\kappa} \to 0$, $m = O(n^{1/3})$,  $m \to \infty$, $n \tau_n^3 \to \infty$,
\begin{align}
 \sup_{t \in [0,1]} \lt| \hat{{ \Sigma}}_d(t) - \tilde{{ \Sigma}}_d(t) \rt| &= \Op\left[ \surd{m} \{(n\tau_n^{3/2})^{d-1/2} + \tau_n^3\}\tau_n^{-1/\kappa} +  m \{(n\tau_n^{3/2})^{2d-1}+\tau_n^6\}\tau_n^{-1/\kappa}\right]\\ &= \op(m^{2d}).\label{eq:Sigmad_step1}
\end{align}
  
Recall that ${ \Sigma}^A(t) = \sum_{j=m}^{n-m} \frac{m\omega(t, j)}{2} A_{j,m}A_{j,m}^{\T}$, and similar to \eqref{eq:Sigma_decomp}, we have
  \begin{align}
    \sup_{t \in \I}\lt| \hat{{ \Sigma}}_d(t) - \tilde{{ \Sigma}}_d(t) \rt| \leq \sup_{t \in \I}\lt|\breve { \Sigma}_d(t) - { \Sigma}^A(t) \rt| + \sup_{t \in \I} \left|\sum_{j=m}^{n-m} m\omega(t, j) \tilde { \Delta}^{(d)}_{j}A_{j,m}^{\T}\right|,
    \label{eq:correctiond_sup}
  \end{align}
  where the first term has been  investigated in \cref{lm:correctiond}.

  Define $
      h_{s,j}^{(d)}(t) = \proj_{j-s}(\tilde { \Delta}_{j}^{(d)}A_{j,m}^{\T}).
 $
   Let $N = N_n = \lfloor n\tau_n \rfloor$. 
  Observe that under Assumption  \ref{E:HW},
  \begin{align}
  \sum_{j=m}^{n-m} \omega(t, j) \tilde { \Delta}^{(d)}_{j}A_{j,m}^{\T} = \sum_{s = 0}^N \sum_{j=m}^{n-m} \omega(t, j) h_{s,j}^{(d)} + \sum_{j=m}^{n-m} \omega(t, j) \sum_{s = N+1}^{\infty} h_{s,j}^{(d)}= S_1 + S_2,\label{eq:Sd}
  \end{align}
  where $S_1$ and $S_2$ are defined in the obvious way. To proceed, we first calculate $\| h_{s,j}^{(d)}(t)\|_{\kappa}. $
  
  By \cref{lm:psum_dn}, $\sup_j\|\tilde{ \Delta}^{d}_j\|_{2\kappa} = O(m^{d-1/2})$.
  Then,
  % similar to \eqref{eq:h_null},
  we have
  \begin{align}
    \|h_{s,j}^{(d)}(t) \|_{\kappa}
    &= O\left\{m^{d-3/2} \min(chi^{s-m},1) + (1/n + m^{-3/2})\sum_{i=-m+1}^m \psi_{s + i}\right\}.\label{eq:hd}
  \end{align}
  Since $h_{s,h}^{(d)}(t)$ are martingale differences with respect to $j$, since $m/(n\tau_n) \to 0$, we have for $t \in \I$,
  \begin{align}
  \| S_1 \|_{\kappa} &\leq \sum_{s = 0}^{N}  \left\| \sum_{j=m}^{n-m} \omega(t,j) h_{s,j}^{(d)} \right\|_{\kappa}= O\left[ \sum_{s = 0}^{N} \left\{ \sum_{j=m}^{n-m} \omega^2(t,j) \|h_{s,j}^{(d)} \|_{\kappa}^2 \right\}^{1/2} \right]
   = O\{m^{-1/2} (n \tau_n)^{d-1/2}\}.\label{eq:S1d}
  \end{align}
  By \eqref{eq:hd} and triangle inequality, elementary calculation shows that 
  \begin{align}
  \|S_2\|_{\kappa}   &\leq \sum_{j=m}^{n-m} \omega(t,j) \left\| \sum_{s = N+1}^{\infty} h_{s,j}^{(d)} \right\|_{\kappa}  = O\left\{ \left(\sum_{s = N+1}^{\infty} \max_{m \leq j \leq n - m}\| h_{s,j}^{(d)}\|_{\kappa}^2\right)^{1/2} \right\}
   = O\{m^{-1/2}(n\tau_n)^{d-1/2}\}.  \label{eq:S2d}
  \end{align}
  where the first big $O$ follows from the fact that $h_{s,j}^{(d)}$ are martingale differences with respect to $s$.
  Combining \eqref{eq:S1d} and \eqref{eq:S2d}, 
  by chaining argument in Proposition B.1 in  \citep{dette2018change},  we have 
   \begin{align}
  \sup_{t \in \I} \left |\sum_{j=m}^{n-m} m\omega(t, j) \tilde { \Delta}^{(d)}_{j}A_{j,m}^{\T} \right| = \Op\left\{\surd{\frac{m}{(n\tau_n)^{1-2d}}}\tau_n^{-1/\kappa}\right\}.
  \label{eq:supdn2}
  \end{align}
  Combining \cref{lm:correctionlocal} and \eqref{eq:supdn2}, we have shown \eqref{eq:Sigmad_step1}.
\par
\par
\textbf{Step 2}:
Define $\bar{{Q}}_{k,m}^{(d)} = \sum_{i=k}^{k+m-1}{ \mu}_W(t_i) e_i^{(d)} (k = 1,\ldots, n - m + 1)$, $$\bar { \Delta}_{j}^{(d)} = \frac{\bar{{Q}}^{(d)}_{j-m+1, m}- \bar{{Q}}^{(d)}_{j+1, m}}{m}, \quad \bar{{ \Sigma}}_d(t)=\sum_{j=m}^{n-m} \frac{m \bar { \Delta}_{j}^{(d)}(\bar { \Delta}^{(d)}_{j})^{\T}}{2} \omega(t, j).$$ We shall show that under the bandwidth condition $m \tau_n^{3/2}/\log n\to \infty$, $\kappa \geq 2/(3d)$, \begin{align}
  \sup_{t \in \I}\lt|\tilde{{ \Sigma}}_d(t) -\bar{{ \Sigma}}_d(t) \rt| = \Op\{m^{d}(\log n)^d\tau_n^{-1/\kappa}\}= \op(m^{2d}).\label{eq:Sigmad_step2}
\end{align}
Following similar arguments in \cref{cor:5.2}, we have 
\begin{align}
  \max_{1 \leq k \leq n - m + 1}\lt\|\bar{{Q}}_{k,m}^{(d)} - \tilde{{Q}}_{k,m}^{(d)} \rt\|_{2\kappa} = O(\surd{m} (\log n)^d).
  \label{eq:cor_prop5.2}
\end{align}

Using \eqref{eq:cor_prop5.2}, and the fact $\sup_j \|\tilde{ \Delta}_{j}^{(d)}\|_{2\kappa} = O(m^{d-1/2})$, $\sup_j \|\bar { \Delta}_{j}^{(d)}\|_{2\kappa} = O(m^{d-1/2})$, we have 
\begin{align}
  \left\| \tilde { \Delta}_{j}^{(d)}\tilde { \Delta}_{j}^{(d), \T} -  \bar { \Delta}_{j}^{(d)}\bar { \Delta}_{j}^{(d), \T} \right\|_{\kappa} &\leq \left\| \tilde { \Delta}_{j}^{(d)} -  \bar { \Delta}_{j}^{(d)}\right\|_{2\kappa} \left\|\bar { \Delta}_{j}^{(d), \T} \right\|_{2\kappa}\\ &  + \left\| \tilde { \Delta}_{j}^{(d)} \right\|_{2\kappa} \left\|\lt(\tilde { \Delta}_{j}^{(d)} - \bar { \Delta}_{j}^{(d)}\rt)^{\T} \right\|_{2\kappa}\\ &= O\{m^{d-1}(\log n)^d\}.
\end{align}
Since $m \tau_n^{3/2}/\log n \to \infty$, \eqref{eq:Sigmad_step2} follows from triangle inequality and Proposition B.1 in  \citep{dette2018change}.
\par
\par
\textbf{Step 3}:
Let $\zeta_j = \sum_{i=j}^{\infty} \proj_j e_{i}$, $\zeta_j^{\circ} = \zeta_j (t_j) = \sum_{i=j}^{\infty} \proj_j H(t_j,\FF_i)$. Define $\overline{Z}_{k,m} = \sum_{j=0}^L \psi_j  \sum_{i=k}^{k+m-1}{ \mu}_W(t_i) \zeta^{\circ}_{i-j}$,
\begin{align}
  { \Delta}_j^{(d), \circ} = \frac{\overline{{Z}}_{j-m+1,m} -\overline{{Z}}_{j+1,m}}{m}, \quad { \Sigma}_d^{\circ}(t) = \sum_{j=m}^{n-m} \frac{m { \Delta}_j^{(d), \circ} ({ \Delta}_j^{(d), \circ })^{\T}}{2}\omega(t,j).
\end{align}
 Let $L= M m^{1 + \frac{1}{2d+1}} \tau_n^{1/2}$, where $M$ is a sufficiently large constant. We will show that
\begin{align}
  \sup_{t \in \I}|\bar{{ \Sigma}}_d(t) - { \Sigma}_d^{\circ}(t)|  = \Op\left\{m^{2d}\lt(m^{- \frac{1/2-d}{2d+1}}  \tau_n^{d/2-1/4-1/\kappa}\rt) \right\} = \op(m^{2d}).
  \label{eq:Sigmad_step3}
\end{align}
 Since $m\tau_n \to \infty$, $m^2\tau_n^{1/2}/n  = O(n^{-1/3}\tau_n^{1/2}) = o(1)$, then  $L/m \to \infty$, $L/m^2 \to 0$, $m^{1+1/(2d)}/L \to \infty$, $L(\log n )^2/n \to 0$.
 Observe that 
\begin{align}
  &\lt\|\bar{{ \Sigma}}_d(t) - { \Sigma}_d^{\circ}(t)\rt\|_{\kappa}\\
  &\leq \sum_{j=m}^{n-m} \frac{m\omega(t,j)}{2}\left\|\bar { \Delta}_{j}^{(d)}\lt(\bar { \Delta}_{j}^{(d)}\rt)^{\T} -  { \Delta}_j^{(d), \circ} \lt({ \Delta}_j^{(d), \circ }\rt)^{\T}\right\|_{\kappa} \\ 
  & \leq m \max_{m \leq j \leq n-m} \left(\lt\|\bar { \Delta}_{j}^{(d)}\rt\|_{2\kappa}\lt\|\lt(\bar { \Delta}_{j}^{(d)}\rt)^{\T} - \lt({ \Delta}_j^{(d), \circ }\rt)^{\T}\rt\|_{2\kappa} + \lt\|\bar { \Delta}_{j}^{(d)} -  { \Delta}_j^{(d), \circ}\rt\|_{2\kappa} \lt\|\lt({ \Delta}_j^{(d), \circ }\rt)^{\T} \rt\|_{2\kappa} \right), \label{eq:step3a}
\end{align}
where 
$$
 \lt\|\bar { \Delta}_{j}^{(d)} -  { \Delta}_j^{(d), \circ}\rt\|_{2\kappa} \leq \frac{1}{m}\left(\lt\|\bar{{Q}}_{j-m+1,m}^{(d)} - \overline{{Z}}_{j-m+1,m} \rt\|_{2\kappa} + \lt\|\bar{{Q}}_{j+1,m}^{(d)} -\overline{{Z}}_{j+1,m} \rt\|_{2\kappa}\right).
% \label{eq:step3b}
$$

Define $$\overline{W}_{k,m} = \sum_{j=0}^L \psi_j \sum_{i=k}^{k+m-1}{ \mu}_W(t_i) e_{i-j},\quad  1 \leq k \leq n - m + 1.$$Then, we have for $1 \leq k \leq n - m + 1$ that 
\begin{align}
    \lt\|\bar{{Q}}_{k-m+1,m}^{(d)}  - \overline{Z}_{k-m+1,m}\rt\|_{2\kappa}  &\leq \lt\|\bar{{Q}}_{k-m+1,m}^{(d)} -  \overline{W}_{k-m+1,m}\rt\|_{2\kappa} +  \lt\|\overline{W}_{k-m+1,m} - \overline{Z}_{k-m+1,m}\rt\|_{2\kappa}\\
    &= C_{2\kappa , 1} + C_{2\kappa , 2},\label{eq:step3c}
\end{align}
where $C_{2\kappa , 1}$ and $C_{2\kappa , 2}$ are defined in the obvious way. 

 Under conditions \ref{B:H_delta} and \ref{Ass-W}, by Burkholder's inequality, we have for $k \geq m$,
\begin{align}
 C_{2\kappa , 1}&= \left \|\sum_{j= L + 1}^{\infty} \left\{\sum_{i=0}^{\min\{m, j - L\} - 1} \psi_{j-i}{ \mu}_W\left(\frac{k-i}{n}\right)\right\} e_{k-j}\right\|_{2\kappa} \\
  & \leq \sum_{t = 0}^{\infty} \left[\sum_{j= L + 1}^{\infty} \left\{\sum_{i=0}^{\min\{m, j - L\} - 1} \psi_{j-i}{ \mu}_W\left(\frac{k-i}{n}\right)\right\}^2 \|\proj_{k-j-t} e_{k-j}\|_{2\kappa}^2 \right]^{1/2}\\ 
  & \leq \left[\sum_{j = L + 1}^{\infty} \left\{\sum_{i=0}^m \psi_{j-i}{ \mu}_W\left(\frac{k-i}{n}\right)\right\}^2 \right]^{1/2}  \sum_{t = 0}^{\infty} \delta_{2\kappa}(H, t, (-\infty, 1]) \\ 
  & = O(L^{d-1/2} m).\label{eq:step3_1}
\end{align}

Then, we consider the upper bound of $C_{2\kappa , 2}$. Let $p_{j,k,m} = \sum_{i=(j-L)_+}^{(m-1) \wedge j} \psi_{j-i} { \mu}_W\left(\frac{k-i}{n}\right)$, $m \leq k \leq n-m$.
Then, for $m \leq k \leq n-m$, we can write 
$$
  \overline{W}_{k-m+1,m} - \overline{Z}_{k-m+1,m} = \sum_{j= 0}^{L+m-1} p_{j,k,m} (e_{k-j} - \zeta^{\circ}_{k-j}).
$$

After a careful check on Lemma 2 in \citep{wu2011gaussian}, we have
\begin{align}
  \max_{0 \leq l \leq L+m-1} \left\| \sum_{j=0}^{l} (e_{k-j} - \zeta_{k-j}) \right\|_{2\kappa}^2 \leq M \sum_{i=1}^{L+m} \left\{\sum_{j=i}^{\infty} \delta_{2\kappa}(H,j,(-\infty, 1])\right\}^2 = O(1),\label{eq:u_zeta}
\end{align}

where $M$ is a sufficiently large constant.
Following the proof of  Corollary 2  in \citep{wu2011gaussian}, under Assumption  \ref{A:H_long}, we obtain
\begin{align}
  \|\zeta_i - \zeta_i^{\circ}\|_{2\kappa}
  & = O\{(\log n)^2/n\}.\label{eq:zeta_zeta}
\end{align}
Observe that
\begin{align}
  \sum_{l=0}^{L+m-1} |p_{l,k,m}- p_{l-1,k,m}|
= O(L^d).\label{eq:p_dif}
\end{align}
Then, by the summation-by-parts formula, combining \eqref{eq:u_zeta}, \eqref{eq:zeta_zeta} and \eqref{eq:p_dif},  since $L(\log n)^2/n \to 0$, we have 
\begin{align}
  \|\overline{W}_{k-m+1,m} - \overline{Z}_{k-m+1,m}\|_{2\kappa} &= \left\|\sum_{j= 0}^{L+m-1} p_{j,k,m} (e_{k-j} - \zeta^{\circ}_{k-j}) \right\|_{2\kappa} 
  = O(L^d). \label{eq:step3_2}
\end{align}

Since $\sup_j \|\bar { \Delta}_{j}^{(d)}\|_{2\kappa} = O(m^{d-1/2})$, $\sup_j \|{ \Delta}_{j}^{(d),\circ}\|_{2\kappa} = O(m^{d-1/2})$, by \eqref{eq:step3a},
% \eqref{eq:step3b}, 
 \eqref{eq:step3c}, \eqref{eq:step3_1} and \eqref{eq:step3_2},  since $L / m^2 \to 0$, we obtain 
\begin{align}
  \|\bar{{ \Sigma}}_d(t) - { \Sigma}_d^{\circ}(t)\|_{\kappa}   = O( L^{d-1/2} m^{d+1/2}).
\end{align}
Under the  conditions $m\tau_n^{3/2} \to \infty$ and $\kappa \geq 4/(1/2 - d)$, 
\eqref{eq:Sigmad_step3} then follows from  Proposition B.1 in \citep{dette2018change}.
\par
\par
\textbf{Step 4}: 
We shall show that under condition $m \tau_n^{3/2} \to \infty$,
\begin{align}
  \sup_{t \in \I} \left| { \Sigma}_d^{\circ}(t) - \E { \Sigma}_d^{\circ}(t)\right| = \Op\left\{m^{2d} (m\tau_n^{3/2})^{-1/2} \right\} = \op(m^{2d}).\label{eq:Sigmad_step4}
\end{align}
 Following similar arguments in the proof of Theorem 3.1 \citep{wu2006invariance}, 
for $k =0,\ldots, \lf \frac{n - m}{2L}\rf$,  let
$D_{k,i} =  \Delta_{2kL+i}^{(d), \circ} (\Delta_{2kL+i}^{(d), \circ })^{\T}- \E\{ \Delta_{2kL+i}^{(d), \circ} ( \Delta_{2kL+i}^{(d), \circ })^{\T}|\F_{2kL+i - 2L})\}$,$(i = 0,\ldots, 2L-1)$, and 
$$
E_{h} = \E\{\Delta_{h}^{(d), \circ} (\Delta_{h}^{(d), \circ })^{\T}|\F_{h- 2L}\} - \E\{\Delta_{h}^{(d), \circ} ( \Delta_{h}^{(d), \circ })^{\T}\} ,\quad (h = m,,\ldots,  n-m).$$
 Let $D_{k,i} = 0$, if $2kL+i < m$ or $2kL+i  > n-m$.
Then, we have
\begin{align}
   { \Sigma}_d^{\circ}(t) - \E { \Sigma}_d^{\circ}(t) 
  & =  \sum_{h=m}^{n-m} \frac{m \omega(t,h)}{2} E_h + \sum_{i = 0}^{2L-1}\sum_{k=0}^{\lf n/(2L)\rf} \frac{m \omega(t, 2kL +i)}{2}  D_{k,i}.\label{eq:Sigmao}
\end{align}

Recall that $ \Delta_h^{(d), \circ} = \frac{\overline{{Z}}_{h-m+1,m} -\overline{{Z}}_{h+1,m}}{m}$, in which $$\overline{Z}_{h,m} = \sum_{j=0}^L \psi_j  \sum_{i=h}^{h+m-1}{ \mu}_W(t_i) \zeta^{\circ}_{i-j} = \sum_{j=0}^{L+m-1} p_{j,h + m -1,m} \zeta^{\circ}_{h + m -1 -j},$$
 where $p_{j, h + m - 1,m} = \sum_{i=(j-L)_+}^{(m-1) \wedge j} \psi_{j-i}{ \mu}_W\left(\frac{h + m -1 -i}{n}\right)$,  $\{\zeta^{\circ}_j\}$ are martingale differences. \par Under the geometric measure contraction condition, for $j = 0,\ldots, L$, we have
 \begin{align}
  \|\E\{(\zeta^{\circ}_{r- j})^2 | \FF_{r - 2L}\} - \E\{(\zeta^{\circ}_{r - j})^2)\} \| = O(\chi^L).\label{eq:zetaGMC}
 \end{align}
By \cref{lm:basic} and elementary calculation, we have  
\begin{align}
  \sum_{j=0}^{L+m-1} | p_{j,s,m}p^{\T}_{j,s,m} | = O(m^{2d+1}),\quad  \sum_{j=0}^{L-1} | p_{j,s,m}p^{\T}_{j+m,s,m}| = O(m^{2d+1}).\label{eq:pp}
\end{align}
Therefore, combining \eqref{eq:zetaGMC} and \eqref{eq:pp},  we derive
\begin{align}
\|E_h\| = O(m^{2d-1}\chi^L).\label{eq:Eki}
\end{align}

By Burkholder's inequality, uniformly for all $i = 0,\ldots, 2L-1$,
\begin{align}
  \left\|\sum_{k=1}^{\lf n/(2L)\rf}\omega(t, 2kL + i) D_{k,i} \right\|^2 & \leq C \sum_{k=1}^{\lf n/(2L)\rf}\omega^2(t, 2kL + i) \| D_{k,i}\|^2\\  
   &\leq 2C \sum_{k \in \{r: |2rL + i - nt| \leq n\tau_n \}}  (\| \Delta_{2kL+i}^{(d), \circ }\|_4 \|( \Delta_{2kL+i}^{(d), \circ })^{\T}\|_4)^2 /(n\tau_n)^2\\ 
  &= O\{(n\tau_n)^{-1}L^{-1} m^{4d-2}\},\label{eq:Dki}
\end{align}
where $C$ is a sufficiently large constant. Therefore, since $L /(n\tau_n) = m^{1 + \frac{1}{2d+1}} /(n \tau_n^{1/2}) = O\{1/(m\tau_n^{1/2})\}$,  by \eqref{eq:Sigmao}, \eqref{eq:Eki} and \eqref{eq:Dki},  and Proposition B.1 in  \citep{dette2018change},  we have shown \eqref{eq:Sigmad_step4}.

\par
\par
\textbf{Step 5}:
Recall that $L= M m^{1 + \frac{1}{2d+1}} \tau_n^{1/2}$, $m\to \infty$, $m = O(n^{1/3})$. It follows that $m^{1+\frac{1}{d+1}}/L  \to \infty$, $L^2/(mn) =O(m^3\tau_n/n) = o(1)$.
We shall show that uniformly for $s \in \I$, 
\begin{align}
  m^{-2d} \E { \Sigma}_d^{\circ}(s) = \kappa_2(d) { \mu}_W(s){ \mu}_W^{\T}(s){\sigma}^2_H(s)+ O(f_n),\label{eq:Sigmad_step5}
\end{align}
where $\kappa_2(d) = \Gamma^{-2}(d+1)\int_{0}^{\infty}\{t^d - (t-1)_+^d\}\{2t^d - (t-1)_+^d - (t+1)^d\} dt$,
and $f_n = m^{-d} + \tau_n^2 + L^{d+1}/m^{d+2} + L^2/(mn) +  (L/m^2)^{-\frac{1}{d-2}} = o(1).$
Note that $\overline{Z}_{k-m+1,m} = \sum_{j= 0}^{L+m-1} p_{j,k,m}  \zeta^{\circ}_{k-j}. $ Recall 
\begin{align}
  p_{j,k,m} &= \sum_{i=(j-L)_+}^{(m-1) \wedge j} \psi_{j-i}{ \mu}_W\left(\frac{k-i}{n}\right) = 
  \begin{cases}
    \sum_{i=0}^{j-1} \psi_{j-i}{ \mu}_W\left(\frac{k-i}{n}\right) + { \mu}_W(\frac{k-j}{n}), & j = 0,\ldots, m - 1\\ 
    \sum_{i=0}^{m-1} \psi_{j-i}{ \mu}_W\left(\frac{k-i}{n}\right), &j =  m,\ldots,  L\\ 
    \sum_{i=j-L}^{m-1} \psi_{j-i}{ \mu}_W\left(\frac{k-i}{n}\right)= O(mL^{d-1}), & j= L+1,\ldots,L+m.\\ 
  \end{cases}
\end{align}

Then, approximate $p_{j,k,m}$ by integrals. When $j = 0,\ldots, m - 1$, by the continuity of ${ \mu}_W$ and \cref{lm:basic},
\begin{align}
  m^{-d} \Gamma(d) p_{j,k,m}  = d^{-1}{ \mu}_W\left(k/n\right)   (j/m)^d + O(m^{-d} + m/n).
\end{align}
When $j =  m,\ldots,  L$,
\begin{align}
  m^{-d}  \Gamma(d)p_{j,k,m} = d^{-1}{ \mu}_W\left(k/n\right)  \{(j/m)^d - ((j+1)/m-1)^d\} + O\{m/n + m^{-1}(j/m-1)^{d-1} \}.
\end{align}
Since $(\zeta_i^{\circ})$ are martingale differences, and $\sigma_H(t_j) = \| \sum_{i=j}^{\infty} \proj_j H(t_j,\FF_i)\| = \| \zeta^{\circ}_j \|$, \eqref{eq:Sigmad_step5} then follows from elementary calculation.
\par
\par
\textbf{Step 6}:
Let $g_{\kappa,n} = m^{-d} (\log n)^{d} \tau_n^{-1/\kappa} +\surd{m}\{ (n\tau_n^{3/2})^{d-1/2} + \tau_n^3\}\tau_n^{-1/\kappa} +  m \{(n\tau_n^{3/2})^{2d-1} + \tau_n^6\}\tau_n^{-1/\kappa} + m^{- \frac{1/2-d}{2d+1}}  \tau_n^{d/2-1/4-1/\kappa} +  (m\tau_n^{3/2})^{-1/2} + f_n$.
Summarizing Step 1-5, we have
\begin{align}
  \sup_{t \in \I}\left| m^{-2d}\hat{{ \Sigma}}_d(t) - \kappa_2(d)\sigma_H^2(t) { \mu}_W(t) \mu^{\T}_W(t)\right|   = \Op (g_{\kappa, n} ) = 
  \op(1).\label{eq:fixdrate}
\end{align}
% \end{proof}

% \subsubsection{A corollary of \texorpdfstring{\cref{lm:Sigma_d}}{Sigmad}}
% Similar to the case with covariates, the following \cref{lm:lrv_fix_trend} suggests in the time-varying trend model \eqref{eq:long_memory1}, the uniform convergence of the difference-based estimator \eqref{eq:sigmaHhat} divided by $m^{2d}$.
% \begin{corollary}\label{lm:lrv_fix_trend}
% Let Assumptions  \ref{assumptionHp}, \ref{A:K}, and  \ref{A:beta} be satisfied with $\kappa \geq \frac{4}{1/2-d}$. 
%   Assuming $m\tau_n^{3/2} \to \infty$, $m /(n\tau_n^2) \to 0, \tau_n \to 0$, $n\tau_n^3 \to \infty$, $m = O(n^{1/3})$,  then under the fixed alternatives, we have 
%  \begin{align}
%   \sup_{t \in \I}\left| m^{-2d}\hat{\sigma}^2_d(t) - \kappa_2(d)\sigma_H^2(t)\right|  = \op(1).
%  \end{align}
% %  where $I = [\gamma_n, 1-\gamma_n] \subset (0,1)$, $\gamma_{n}=\tau_{n}+(m+1) / n$.
% \end{corollary}
% \cref{lm:lrv_fix_trend} follows from  Step 1 (with no correction procedure for bias), Step 3, Step 4 and Step 5 of \cref{lm:Sigma_d}.
\end{proof}
\subsection{Proof of \texorpdfstring{\cref{lm:Sigma_dn}}{Proof of Theorem 4}}

  We define the notation under the local alternatives by replacing $d$ with $d_n$.
% \subsubsection{Technical results for the proof of \texorpdfstring{\cref{lm:Sigma_dn}}{Sigmadn}} 
In the following, 
\cref{lm:correctionlocal} studies the asymptotic behavior of the bias correction term under the local alternatives. \cref{lm:psum_dn} investigates the physical dependence of $U^{(d_n)}(t, \FF_i)$ as well as the order of its partial sum process under the local alternatives. 
\begin{lemma}\label{lm:correctionlocal}
%   Recall that $$A_{j,m} = \frac{1}{m} \sum_{i=j-m+1}^j\{x_i x_i^{\T} { \beta}(t_i) - x_{i+m} x_{i+m}^{\T} { \beta}(t_{i+m})\},\quad { \Sigma}^A (t) = \sum_{j=m}^{n-m} \frac{m A_{j,m} A_{j,m}^{\T}}{2}\omega(t, j),$$
%   and
%   $$ \hat{A}_{j,m}^{(d_n)} = \frac{1}{m}\sum_{i= j-m+1}^{j} (x_i x_i^{\T}\breve { \beta}^{(d_n)}(t_i)-x_{i+m} x_{i+m}^{\T}\breve { \beta}^{(d_n)}(t_{i+m})), \quad \breve { \Sigma}^{(d_n)}(t) = \sum_{j=m}^{n-m} \frac{m\hat{A}^{(d_n)}_{j} \hat{A}_{j,m}^{(d_n),\T}}{2}\omega(t, j). $$
%   Under Assumptions \ref{A:K}, \ref{A:beta}, \ref{Ass-W} and \ref{Ass-E} under the bandwidth conditions $m = O(n^{1/3})$, $m / (n \tau^2_n) \to 0$, $n\tau_n^3 \to \infty$, and $\surd{m} \tau_n^{2-1/\kappa} \to 0$, for $d_n = c/\log n$, we have 
Under the conditions of \cref{lm:Sigma_dn}, we have
  \begin{align}
  \sup_{t \in \I}\lt| \breve { \Sigma}_{d_n}(t) - { \Sigma}^A (t)\rt|= \Op\left( \surd{m}\tau_n^{3-1/\kappa} + \surd{\frac{m}{n\tau_n^{3/2+2\kappa}}}\right)=\op(1).
  \end{align}
  \end{lemma}
  \begin{proof}[ of \cref{lm:correctionlocal}]

  Letting $d_n = c/\log n$, the proof follows from similar steps in \cref{lm:correctiond}.
  \end{proof}

  \begin{lemma}\label{lm:psum_dn}
    Under Assumptions \ref{Ass-W} and \ref{assumptionHp}, $m \to \infty$, $m = O(n)$, we have
      \begin{align}
      \sup_{1 \leq k \leq n - m + 1}\left\|\sum_{i=k}^{k+m-1} x_i e_i^{(d_n)}\right\|_4 = O(\surd{m}).
      \end{align}
    \end{lemma}
    \begin{proof}[Proof of \cref{lm:psum_dn}]
    Define $\tilde {{Q}}_{k,m}^{(d_n)} = \sum_{i=k}^{k+m-1} x_i e_i^{(d_n)}$.
    By similar arguments in \cref{lm:delta_xed}, we obtain 
    \begin{align}
    %   \delta_4(U^{(d_n)}, k)  = O(d_n (1+k)^{d_n-1}),\quad k > 0,
    \delta_4(U^{(d_n)}, k)  = O\{\psi_k(d_n)\},\quad k \geq 0,
      \label{eq:delta_xedn}
    \end{align}
    and $\delta_4(U^{(d_n)}, k) = 0$, for $k < 0$.
    % For the simplicity of calculation, we write $$\delta_4(U^{(d_n)}, k)  = O(\psi_k),\quad k \geq 0.$$
    Then, uniformly for $1 \leq k \leq n - m + 1$, by Burkholder's inequality we have 
    \begin{align}
      \lt\|\tilde {{Q}}_{k,m}^{(d_n)}\rt\|_4^2 \leq B_4^2 \left\|\sum_{l = -\infty}^{\infty} \lt|\proj_l \tilde{{Q}}_{k,m}^{(d_n)}\rt|^2 \right\| \leq B_4^2 \sum_{l = -\infty}^{k+m-1}\lt\| \proj_l \tilde{{Q}}_{k,m}^{(d_n)}\rt\|_4^2 \leq B_4^2  \sum_{l = -\infty}^{k+m-1}\left\{\sum_{i = k-l}^{k+m-1-l}\delta_4(U^{(d_n)}, i)\right\}^2, \label{eq:Qkm_dn}
    \end{align}
    where $B_4$ is a constant.
    Therefore, combining \eqref{eq:delta_xedn} and \eqref{eq:Qkm_dn}, it follows  from  \cref{cor:kar_dn} that
   \begin{align}
   \max_{1 \leq k \leq n-m+1} \lt\|\tilde{{Q}}_{k,m}^{(d_n)}\rt\|_4^2 &= O \left[ \sum_{l = -\infty}^{k}\left\{(k+m-l)^{d_n} - (k-l+1)^{d_n}\right\}^2 + \sum_{l = k+1}^{k+m-1}(k+m-l)^{2d_n} \right]
    = O(m). \label{eq:sum_delta_dn}
   \end{align}
   
   \end{proof}

% \subsubsection{Proof of \texorpdfstring{\cref{lm:Sigma_dn}}{Theorem 4}}
\begin{proof}[Proof of \cref{lm:Sigma_dn}]
    Recall that
    $$\tilde{{ \Sigma}}_{d_n}(t)=\sum_{j=m}^{n-m} \frac{m \tilde { \Delta}_{j}^{(d_n)}\tilde{ \Delta}^{(d_n),\T}_{j}}{2}\omega(t, j),\quad \tilde { \Delta}_{j}^{(d_n)}=\frac{\tilde{Q}^{(d_n)}_{j-m+1, m}- \tilde{Q}^{(d_n)}_{j+1, m}}{m}$$ where $\tilde{Q}^{(d_n)}_{k, m}=\sum_{i=k}^{k+m-1} x_{i}e_i^{(d_n)}$.\par We break the proof in the following 8 steps.
    
    \textbf{Step 1}: Following the proof of  \cref{lm:Sigma_d} by replacing $d$ by $d_n$, since $\kappa \geq 4$,  we have \begin{align}
   \sup_{t \in \I} \lt| \hat{{ \Sigma}}_{d_n}(t) - \tilde{{ \Sigma}}_{d_n}(t) \rt| = \Op\left( \surd{m}\tau_n^{3-1/\kappa} + \surd{\frac{m}{n\tau_n^{3/2+2\kappa}}} \right) = \op(1).\label{eq:local_step1}
   \end{align}
  \par
  \textbf{Step 2:}
  Let $L = m^2 \tau_n^{1/2}$, $\check e_{i,L}^{(d_n)} = \sum_{j=0}^{L} \psi_j e_{i-j}.$
    Define 
   $$\check{Q}^{(d_n)}_{k, m} = : \sum_{i=k}^{k+m-1} x_{i}\check e_{i, L}^{(d_n)},\quad \check { \Delta}_{j}^{(d_n)}=\frac{\check{{Q}}^{(d_n)}_{j-m+1, m}- \check{{Q}}^{(d_n)}_{j+1, m}}{m}, \quad \check{{ \Sigma}}_{d_n}(t)=\sum_{j=m}^{n-m} \frac{m \check { \Delta}_{j}^{(d_n)}\check { \Delta}^{(d_n),\T}_{j}}{2}\omega(t, j).$$
  
  In this step, we shall show that under bandwidth condition $m \tau_n^{3/2} \to \infty$,
   \begin{align}
   \sup_{t \in \I }\lt|\tilde { \Sigma}_{d_n}(t) - \check { \Sigma}_{d_n}(t)\rt| = \Op(m^{-1/2}\tau_n^{-3/4}) = \op(1).\label{eq:local_step2}
   \end{align}
   Observe that 
   \begin{align}
   e_i^{(d_n)} = \check{e}^{(d_n)}_{i,L}  + \tilde{e}^{(d_n)}_{i,L},~\text{where}~ \check e_{i,L}^{(d_n)} = \sum_{j=0}^{L} \psi_j e_{i-j}, ~ \tilde{e}^{(d_n)}_{i,L} = \sum_{j=L+1}^{\infty} \psi_j e_{i-j}.
   \end{align}
  % Similar to \eqref{eq:eim},
  By \cref{physical}, we have
  $
  \| \tilde{e}^{(d_n)}_{i,L}\|^2_4
   = O\{\sum_{s=L+1}^{\infty}(s+1)^{2d_n-2}\} = O(L^{-1}).
  $
  Then, under Assumption  \ref{Ass-W}, we have uniformly for $1 \leq k \leq n - m + 1$,
  \begin{align}
  \lt\| \check{Q}^{(d_n)}_{k, m} - \tilde{Q}^{(d_n)}_{k, m}\rt\|  \leq  m \max_{1 \leq i \leq n}\lt\| x_i \tilde{e}^{(d_n)}_{i,L}\rt\| \leq m \max_{1 \leq i \leq n} \| x_i \|_4 \lt\| \tilde{e}^{(d_n)}_{i,L} \rt\|_4  = O(m/\surd{L})\label{eq:diffQdn}.
  \end{align}
By \cref{lm:psum_dn} and \eqref{eq:diffQdn}, we  have  
  \begin{align}
  \left\| \check { \Sigma}_{d_n} (t) - \tilde { \Sigma}_{d_n} (t)\right\|\leq m \max_{m\leq j \leq n-m}\left\| \tilde  { \Delta}_{j}^{(d_n)} -  \check{ \Delta}_{j}^{(d_n)}\right\|_4 \left( \left\| \tilde  { \Delta}_{j}^{(d_n)} -  \check{ \Delta}_{j}^{(d_n)}\right\|_4  + 2\left\| \tilde{ \Delta}_{j}^{(d_n)}\right\|_4\right) = O\left(\surd{m/L}\right).
  \end{align}
  By Proposition B.1 in  \citep{dette2018change}, since $m/(L\tau_n)  = m^{-1}\tau_n^{-3/2} \to 0$, \eqref{eq:local_step2} is proved.
  \par
  \textbf{Step 3}	: 
    Define 
   $\bar{Q}^{(d_n)}_{k, m} = : \sum_{i=k}^{k+m-1}\{x_{i}e_i  + { \mu}_W(t_i) (\check e_{i,L}^{(d_n)} - e_i)\}$,$$\bar { \Delta}_{j}^{(d_n)}=\frac{\bar{{Q}}^{(d_n)}_{j-m+1, m}- \bar{{Q}}^{(d_n)}_{j+1, m}}{m}, \quad \bar{{ \Sigma}}_{d_n}(t)=\sum_{j=m}^{n-m} \frac{m \bar { \Delta}_{j}^{(d_n)}\bar { \Delta}^{(d_n),\T}_{j}}{2}\omega(t, j). $$
  We shall show that 
   \begin{align}
   \sup_{t \in \I }\left|\check { \Sigma}_{d_n}(t) - \bar { \Sigma}_{d_n}(t)\right|  = \Op\left(\surd{\frac{m}{n\tau_n^{1+2\kappa}}} + d_n\right) = \op(1).\label{eq:local_step3}
   \end{align}
  Observe that 
  \begin{align}
  \check {Q}^{(d_n)}_{k, m}
   - \bar{Q}^{(d_n)}_{k, m} = \sum_{j=1}^L \sum_{i=k}^{k+m-1} (x_i -  { \mu}_W(t_i)) \psi_j e_{i-j} = \sum_{j=1}^L \sum_{i=k}^{k+m-1}  \psi_j  \bar{x}_i e_{i-j}, 
    \end{align}
    where $\bar{x}_i = x_i - { \mu}_W(t_i)$. Let 
    \begin{align}
    { \vartheta}_{k,m} = \frac{1}{m} \sum_{j=1}^L \sum_{i=k-m+1}^{k}  \psi_j  (\bar{x}_i e_{i-j} - \bar{x}_{i+m} e_{i+m-j}).
    \end{align}
    Then, it follows that 
   \begin{align}
       \check { \Sigma}_{d_n}(t) - \bar { \Sigma}_{d_n}(t) =  \sum_{j=m}^{n-m} \frac{m\omega(t, j)}{2}( \check { \Delta}_{j}^{(d_n)}\check { \Delta}^{(d_n), \T}_{j} - \bar { \Delta}_{j}^{(d_n)}\bar { \Delta}^{(d_n), \T}_{j}),\label{eq:varstep3}
   \end{align}
   and \begin{align}
    \check { \Delta}_{j}^{(d_n)}\check { \Delta}^{(d_n), \T}_{j} - \bar { \Delta}_{j}^{(d_n)}\bar { \Delta}^{(d_n),\T}_{j}  = { \vartheta}_{k,m} { \vartheta}_{k,m}^{\T} +  { \vartheta}_{k,m}\bar { \Delta}^{(d_n),\T}_{j} + \bar { \Delta}^{(d_n)}_{j}{ \vartheta}_{k,m}^{\T}.\label{eq:vartheta}
   \end{align}
   
   \textbf{Step 3.1}
   We first show that
   \begin{align}
  \sup_{t \in \I } \left|\E\{\check { \Sigma}_{d_n}(t) - \bar { \Sigma}_{d_n}(t)\}\right| = O(d_n) = o(1).\label{eq:local_step31}
  \end{align}
  
  Observe that $(\bar{x}_i)_{i=1}^n$, $(e_i)_{i=-\infty}^n$ are two centered sequences.
  %%%%%This is very important%%%%%%%%
  %One may ask why we don't just trim $x_i e_{i-j}$ for $j > 0$. Actually, the cumulation of $x_i e_{i-j}$, $j>0$, is of the same order of $x_i e_i$. This step shows $\mu_i e_{i-j}$ is the leading term. This part is a counterpart of  proposition 5.3. By step 3.1 and step 3.2 leads to a more precise result
  % Using similar techniques in proposition 5.3(just as corollary of propoisition 5.2)  cannot yield o(1) after chaining argument.
  % In the step 2 of Lemma 5(\cref{lm:Sigma_d}), we actually do not tackle problem. We assume the existence $kappa$ moment.
    Under Assumptions \ref{Ass-W} and  \ref{assumptionHp}, by Lemma 7 in \citep{zhou2014vstat} we have for $l, j>0$,
    \begin{align}
    |\E(\bar{x}_i e_{i-j}\bar{x}^{\T}_{i+k} e_{i+k-l} )|
    = O(\chi^{\rho^*}),
    \end{align}
    where
  following the lines in the proof of Theorem 2 in \citep{zhou2014vstat}, we have 
  \begin{align}
  \rho^* \geq \frac{1}{2} \min \{\max (|k|,|k-l+j|), \max (|k-l|,|k+j|)\},
  \end{align}
  where we define the right hand side as $\rho_{k,l,j}$.
 Then, we are able to bound the expectation of \eqref{eq:vartheta}, for $ 1 \leq p \leq L$, $1 \leq q \leq L$,
    \begin{align}
    &|\E({ \vartheta}_{k,m} { \vartheta}_{k,m}^{\T} )|\\ &\leq \frac{1}{m^2}  \sum_{p,q = 1}^L \psi_p \psi_q \left| \E \left[ \left\{ \sum_{i=k-m+1}^k(\bar{x}_i e_{i-p} -\bar{x}_{i+m} e_{i+m-p})\right\}\left\{ \sum_{j=k-m+1}^k(\bar{x}_j e_{j-q} -\bar{x}_{j+m} e_{j+m-q})\right\}^{\T} \right] \right|\\
  & = O\left(\frac{1}{m^2}  \sum_{p,q = 1}^L  \psi_p \psi_q  \sum_{i,j =k-m+1}^{k+m} \chi^{\rho_{j-i, q, p}} \right).\label{eq:Vdn}
    \end{align}
  Consider $q \leq p$,  since when $q \geq 1$, $\psi_q = O(d_n(1+q)^{d_n-1})$,  we have 
   \begin{align}
       &\frac{1}{m^2}  \sum_{q = 1}^L \sum_{p = q}^L  \psi_p \psi_q  \sum_{i,j =k-m+1}^{k+m} \chi^{\rho_{j-i, q, p}}\\ &= \frac{1}{m^2}  \sum_{q = 1}^L \sum_{p = q}^L  \psi_p \psi_q  \sum_{i =k-m+1}^{k+m} \left(\sum_{j> i+ (q-p)/2} \chi^{(j-i-q+p)/2} + \sum_{j\leq i+ (q-p)/2} \chi^{(i-j)/2}\right) \\
       & = O \left( \sum_{q = 1}^{L}  \psi_q^2/m  \right)= O(d_n/m).
   \end{align}
   Similarly, 
   \begin{align}
       \frac{1}{m^2}  \sum_{q = 1}^L \sum_{p = 1}^{q-1}  \psi_p \psi_q  \sum_{i,j =k-m+1}^{k+m} \chi^{\rho_{j-i, q, p}}= O(d_n/m).
   \end{align}
   
  % %%%%%%%%%%%%%% details %%%%%%%%%%%%%%
  %  Consider $q > p$, we have 
  %  \begin{align}
  %      &\frac{1}{m^2}  \sum_{q = 1}^L \sum_{p = 1}^{q-1}  \psi_p \psi_q  \sum_{i,j =k-m+1}^{k+m} \chi^{\rho_{j-i, q, p}} \\
  %      &= \frac{1}{m^2}  \sum_{q = 1}^L \sum_{p = 1}^{q-1}  \psi_p \psi_q  \sum_{i =k-m+1}^{k+m}\left(\sum_{j> i+ (q-p)/2} \chi^{(j-i)/2} + \sum_{j\leq i+ (q-p)/2} \chi^{(i-j+q-p)/2}\right) \\
  %       &= O\left\{\frac{1}{m}  \sum_{q = 1}^L \left(\sum_{ p  = \max\{1, q-4m\} }^{q-1} \psi_p \psi_q \chi^{(q-p)/4} +\sum_{p = 1}^{\max\{1, q-4m\} -1} \psi_p \psi_q \chi^{(-2m+q-p)/2}
  % \right) \right\}\\
  %      & = \frac{1}{m} O\left( \sum_{q = 1}^{\lf \log n \rf} \sum_{p = \max\{1, q-4m\}}^{q-1}  \psi_p \psi_q\chi^{(q-p)/4} + \sum_{q = \lf \log n \rf+1}^L \psi_q\sum_{p = \max\{1, q-4m\}}^{q-1}  \psi_p \chi^{(q-p)/4} +  \sum_{q = 1}^L \psi_q \psi_1 \chi^{m}\right)\\ 
  %       & = O( \sum_{q = 1}^{\lf \log n \rf} d_n^2 / m + \sum_{q = \lf \log n \rf+1}^L q^{2d_n-2}/m +\chi^{m}/m ) = O(d_n/m),
  %  \end{align}
  % where the last equality follows from a careful check of Lemma 3.2 in \citep{KOKOSZKA199519}.
  % %%%%%%%%%%%%%%%%%%%%%%%%%%%%%%%%%%%%%%%
  
  Then, we have
    \begin{align}
    |\E({ \vartheta}_{k,m} { \vartheta}_{k,m}^{\T} )|  = O(d_n/m).\label{eq:Vdn1}
    \end{align}
 
    Following similar arguments in \cref{lm:psum_dn}, $\| \bar { \Delta}_j^{(d_n)}\| = O(m^{-1/2})$.  Then, it follows that
    \begin{align}
    |\E(\bar { \Delta}_j^{(d_n)}{ \vartheta}_{k,m} ^{\T})| \leq \| \bar { \Delta}_j^{(d_n)}\|  \| { \vartheta}_{k,m} ^{\T}\| = O(d_n^{1/2}/m).\label{eq:Vdn2}
    \end{align}
  Therefore, by \eqref{eq:vartheta}, \eqref{eq:varstep3}, \eqref{eq:Vdn1}, and \eqref{eq:Vdn2}, we have \eqref{eq:local_step31}.
  
  \par
  \textbf{Step 3.2}
   We proceed to show that
   \begin{align}
  \sup_{t \in \I }\left|\check { \Sigma}_{d_n}(t) - \bar { \Sigma}_{d_n}(t) - \E(\check { \Sigma}_{d_n}(t) - \bar { \Sigma}_{d_n}(t))\right| =\Op\left(\surd{\frac{m}{n\tau_n^{1+2\kappa}}}\right)= \op(1).
  \label{eq:local_step32}
  \end{align}
  Notice that $\check e_{i,L}^{(d_n)}-e_i$ has summable physical dependence. Specifically, 
  \begin{align}
  \| { \vartheta}_{k,m} - { \vartheta}_{k,m,\{k-s\}}\|_{2\kappa} &= O\left[\frac{1}{m} \sum_{j=1}^L \psi_j \sum_{i=k-m+1}^{k+m} \{\delta_{4\kappa}(W, i-k+s) + \delta_{4\kappa}(H, i-j-k+s) \}\right]\\ 
%   &= O\left(\frac{1}{m} \sum_{i=-m+1}^{m} \sum_{j=1}^L \psi_j \delta_8(H, i - j +s) \right)\\
  & = O\left[\frac{d_n}{m} \sum_{i=-m+1}^{m}\min\left\{ \chi^{i-L+s}L^{d_n-1}, (i+s)^{d_n-1}\right\}1(i+s > 0 ) \right],
  \end{align}
  where in the last equality, we use the fact $ \delta_{4\kappa}(H, k)  =0$, if $k \leq  0$ and $\sum_{j=1}^L \psi_j \chi^{L-j} = O(\psi_L).$
  From \eqref{eq:Vdn1}, we have $\sup_j \|{ \vartheta}_{j,m} \|_{2\kappa} = O(m^{-1/2}) $. \par  For simplicity, write $r_{i,s, n} =\frac{d_n}{m}\min\left\{ \chi^{i-L+s}L^{d_n-1}, (i+s)^{d_n-1}\right\} 1(i+s > 0 )$. 
  Then,  we obtain for $m \leq j \leq n$, $s \geq 0$,
  \begin{align}
  \left\| \proj_{j-s} { \vartheta}_{j,m} { \vartheta}_{j,m}^{\T} \right\|_{\kappa} &\leq \| { \vartheta}_{j,m}\|_{2\kappa} \| { \vartheta}_{j,m}^{\T} - { \vartheta}_{j,m,\{j-s\}}^{\T}\|_{2\kappa}  + \| { \vartheta}_{j,m} - { \vartheta}_{j,m,\{j-s\}}\|_{2\kappa}  \| { \vartheta}_{j,m,\{j-s\}}^{\T} \|_{2\kappa}
  \\ 
%   & = O\left(m^{-3/2} d_n  \sum_{i=-m+1}^{m}\min\left\{ \chi^{i-L+s}L^{d_n-1}, (i+s)^{d_n-1}\right\} 1(i+s > 0 )\right).
 & = O\left(m^{-1/2} r_{i,s,n}\right).
  \label{eq:Pvv}
  \end{align}
   Under Assumption \ref{Ass-E}, similar to (36) in Lemma 3 in \citep{zhou2010simultaneous}, we have
   \begin{align}
    \delta_{2\kappa}(\tilde{ \Delta}(m),k): = \sup_{1 \leq j \leq n} \| \tilde{ \Delta}_j - \tilde{ \Delta}_{j,\{j-k\}}\|_{2\kappa} &= O\left\{ \frac{1}{m}\sum_{i=-m+1}^m\delta_{2\kappa}(U,k+i)\right\}\\ &= O\{\min(\chi^{k-m}, 1)/m\}.
  \end{align}
  Similar to \eqref{eq:Pvv}, and by \eqref{eq:lrv2}, we have
  \begin{align}
  \left\|\bar { \Delta}_{k}^{(d_n)}  - \bar { \Delta}^{(d_n)} _{k,\{k-s\}}\right\|_{2\kappa} &\leq  \left\|\tilde  { \Delta}_{k} -  \tilde { \Delta}_{k,\{k-s\}} \right\|_{2\kappa} +  \frac{1}{m}\sum_{i=k-m+1}^{k+m} \sum_{j=1}^L \psi_j \left\|{ \mu}_W(t_i)(e_{i-j} - e_{i-j,\{k-s\}}) \right\|_{2\kappa}\\ 
%   & = O\left(\min\{\chi^{s-m},1\}/m +\frac{d_n}{m}\sum_{i=-m+1}^{m}\min\{ \chi^{i-L+s}L^{d_n-1}, (i+s)^{d_n-1}\}1(i+s > 0 ) \right).
  & =O\left\{\min(\chi^{s-m},1)/m +\sum_{i=-m+1}^{m}r_{i,s, n} \right\}. \label{eq:deltalocal}
  \end{align}
  
  Since $\sup_j \|\bar { \Delta}_{j}^{(d_n)} \|_{\kappa}= O(m^{-1/2}) $, by \eqref{eq:deltalocal}, similar to \eqref{eq:Pvv}, we obtain
  \begin{align}
  \left\| \proj_{j-s}  \bar { \Delta}_{j}^{(d_n)} { \vartheta}_{j,m}^{\T} \right\|_{\kappa}
%   &\leq \left\| \bar { \Delta}_{j}^{(d_n)} \right\|_4 \left\| { \vartheta}_{j,m}^{\T} - { \vartheta}_{j,m,\{j-s\}}^{\T}\right\|_4  + \left\|\bar { \Delta}_{j}^{(d_n)}  - \bar { \Delta}^{(d_n)} _{j,\{j-s\}}\right\|_4  \left\| { \vartheta}_{j,m,\{j-s\}}^{\T} \right\|_4 \\ 
  & =  O\left\{m^{-1/2} r_{i,s,n}  +  m^{-3/2} \min(\chi^{s-m},1)\right\}.\label{eq:Pdv}
  \end{align}
  
  By Burkholder's inequality, by \eqref{eq:varstep3} and \eqref{eq:vartheta}, combining \eqref{eq:Pvv} and \eqref{eq:Pdv}, we have for $t \in \I$,
  \begin{align}
  &\left\|\check { \Sigma}_{d_n}(t) - \bar { \Sigma}_{d_n}(t) - \E(\check { \Sigma}_{d_n}(t) - \bar { \Sigma}_{d_n}(t))\right\|_{\kappa}\\
   &= O\left\{ \sum_{s=0}^{\infty} \left(\sum_{j=m}^{n-m} \omega^2(t,j)m^2\left\| \proj_{j-s} { \vartheta}_{j,m} { \vartheta}_{j,m}^{\T}  + \proj_{j-s}  \bar { \Delta}_{j}^{(d_n)} { \vartheta}_{j,m}^{\T} + \proj_{j-s}  \bar { \Delta}_{j}^{(d_n)} { \vartheta}_{j,m}^{\T}\right\|_{\kappa}^2 \right)^{1/2}\right\}\\ 
%   & =  O\left\{\sum_{s=0}^{\infty} (n\tau_n)^{-1/2} m^{-1/2} \left(\min\{\chi^{s-m},1\}  + d_n \sum_{i=-m+1}^{m}\min\{ \chi^{i-L+s}L^{-1}, (i+s)^{-1}\} 1(i+s > 0) \right) \right\}\\
  & = O\left(\surd{\frac{m}{n\tau_n}}\right),
  \end{align}
  where in the last equality, we consider $0<i+s<L$, and $i+s \geq L$ separately and use the fact $\sum_{i=1}^L i^{-1} = O(\log L).$ Then, \eqref{eq:local_step32} follows from the chaining argument in Proposition B.1 in  \citep{dette2018change}.
  
  \par
  \textbf{Step 4: Decomposition }
  Recall that $\tilde{Q}_{k,m} = \sum_{i=k}^{k+m-1} x_ie_i$, 
  \begin{align}
      \tilde{ \Delta}_{j}=\frac{\tilde{Q}_{j-m+1, m}- \tilde{Q}_{j+1, m}}{m}, \quad\tilde{{ \Sigma}}(t)=\sum_{j=m}^{n-m} \frac{m \tilde{ \Delta}_{j}\tilde{ \Delta}_{j}^{\T}}{2}\omega(t, j).
   \end{align}
     
  Define $\breve { \Delta}_j^{(d_n)} = \bar { \Delta}_j^{(d_n)} - \tilde{ \Delta}_{j}  = \frac{1}{m} \sum_{i=j-m+1}^j { \mu}_W(t_i) (\check e_{i,L}^{(d_n)} - e_i) - { \mu}_W(t_{i+m}) (\check e_{i+m,L}^{(d_n)} - e_{i+m})$.
  Let $$ \tilde{s}_1 (t) = \sum_{j=m}^{n-m} \frac{m \omega(t,j)}{2} \breve { \Delta}_j^{(d_n)} \breve { \Delta}_j^{(d_n), \T} ,\quad \tilde{s}_2(t) = \sum_{j=m}^{n-m} \frac{m \omega(t,j)}{2} \tilde { \Delta}_j \breve { \Delta}_j^{(d_n), \T} .$$
   Observe that 
  \begin{align}
  \bar{{ \Sigma}}_{d_n}(t) =  \tilde { \Sigma}(t)  +  \tilde {s}_{1} (t) + \tilde{s}_{2}(t) + \tilde{s}^{\T}_{2}(t).
  \end{align}
  
  \textbf{Step 5: Martingale approximation}
  
  Let 
  \begin{align}
    z_j = \sum_{i=j}^{\infty} \proj_j (x_i e_i),\quad z_j^{\circ} = z_j(t_j) = \sum_{i=j}^{\infty} \proj_j U(t_j, \FF_i).
  \end{align}
   Recall that  in \cref{lm:Sigma_d}, $\zeta_j = \sum_{i=j}^{\infty} \proj_je_{i}$, $\zeta_j^{\circ} = \zeta_j (t_j) = \sum_{i=j}^{\infty} \proj_j H(t_j,\FF_i)$. Let $z_{j,1}$ denote the first element in $z_i$. Then, it follows that $z_{j,1} =\zeta_j $, $ z_{j,1}^{\circ} = \zeta_j^{\circ} $. \par
   Define $\overline{Z}^{(d_n)}_{k,m} =  \sum_{i=k}^{k+m-1} \left\{z_i^{\circ}+ \sum_{j=1}^L \psi_j { \mu}_W(t_i) \zeta^{\circ}_{i-j}\right\}$, $${ \Delta}_j^{(d_n), \circ} = \frac{\overline{{Z}}^{(d_n)}_{j-m+1,m} -\overline{{Z}}^{(d_n)}_{j+1,m}}{m},\quad { \Sigma}_{d_n}^{\circ}(t) = \sum_{j=m}^{n-m} \frac{m { \Delta}_j^{(d_n), \circ} ({ \Delta}_j^{(d_n), \circ })^{\T}}{2}\omega(t,j).$$
   
  Similarly to $p_{j,k,m}$ defined in Step 3 of \cref{lm:Sigma_d}, we define $\underline{p}_{j,k,m} = \sum_{i=(j-L)_+}^{(m \wedge j) -1} \psi_{j-i} { \mu}_W\left(\frac{k-i}{n}\right)$.
   By \eqref{eq:step3_2} and similar arguments in Theorem 1(ii) of \citep{wu2007strong}, we have uniformly for $1 \leq k \leq n-m+1$,
   \begin{align}
   \| \bar{Q}_{k,m}^{(d_n)} -\overline{Z}^{(d_n)}_{k,m}\|_4  &\leq \left\|\sum_{j= 1}^{L+m-1} \underline{p}_{j,k+m-1,m}(e_{k+m-1-j} - \zeta^{\circ}_{k+m-1-j}) \right\|_4  \\ &+ \left\|   \sum_{i=k}^{k+m-1}x_i e_i - \sum_{i=k}^{k+m-1}z_i^{\circ}  \right\|_4  = O(1).
   \end{align}
  Since $\sup_j \|\bar { \Delta}_{j}^{(d_n)} \| = O(m^{-1/2}) $, by triangle inequality and Cauchy–Schwarz inequality, we have for $t \in \I$,
  \begin{align}
    \|\bar{{ \Sigma}}_{d_n}(t) - { \Sigma}_{d_n}^{\circ}(t)\| &\leq \sum_{j=m}^{n-m} \frac{m\omega(t,j)}{2}\left\|\bar { \Delta}_{j}^{(d_n)}(\bar { \Delta}_{j}^{(d_n)})^{\T} -  { \Delta}_j^{(d_n), \circ} ({ \Delta}_j^{(d_n), \circ })^{\T}\right\| 
    % & \leq m \max_{m \leq j \leq n-m} \left\{\|\bar { \Delta}_{j}^{(d_n)}\|_{4}\|\{\bar { \Delta}_{j}^{(d_n)}\}^{\T} - ({ \Delta}_j^{(d), \circ })^{\T}\|_{4} + \|\bar { \Delta}_{j}^{(d_n)} -  { \Delta}_j^{(d_n), \circ}\|_{4} \|({ \Delta}_j^{(d_n), \circ })^{\T} \|_{4} \right\}\\
   = O(m^{-1/2}).
  \end{align}
  By chaining argument in Proposition B.1 in  \citep{dette2018change},  we have 
  \begin{align}
  \sup_{t\in \I} |\bar{{ \Sigma}}_{d_n}(t) - { \Sigma}_{d_n}^{\circ}(t) | = \Op\{(m \tau_n)^{-1/2}\} =
  \op(1).\label{eq:local_step5}
  \end{align}
  \par 
  
  \textbf{Step 6}
  Observe that 
   \begin{align}
   \overline{Z}^{(d_n)}_{k,m} =  \sum_{i=k}^{k+m-1} \left\{z_i^{\circ}+ \sum_{j=1}^L \psi_j { \mu}_W(t_i) \zeta^{\circ}_{i-j}\right\} = \sum_{i=k}^{k+m-1} z_i^{\circ}+ \sum_{j=1}^{L+m-1} \underline{p}_{j, k+m-1, m} \zeta^{\circ}_{k+m-1-j}.
    \end{align}
%   By the definition of $\overline{Z}^{(d_n)}_{k,m}$, we have
%   \begin{align}
%   \E(\overline{Z}^{(d_n)}_{k + 1,m} \{\overline{Z}^{(d_n)}_{k + 1,m}\}^{\T}| \FF_{k-2L}) &= \E\left( \sum_{i=k+1}^{k+m} z_i^{\circ} z_i^{\circ, \T} + \sum_{j=1}^{L+m-1} \underline{p}_{j, k+m,m}\underline{p}_{j,k+m,m}^{\T} (\zeta^{\circ}_{k+m- j})^2 \right.\\ & \left.+ \sum_{j=1}^{m-1} \underline{p}_{j, k+m,m} \zeta^{\circ}_{k+m- j} z_{k+m-j} ^{\circ,\T} + \sum_{j=1}^{m-1}  z_{k+m-j}^{\circ} \underline{p}_{j, k+m,m}^{\T} \zeta^{\circ}_{k+m- j}| \FF_{k-2L}\right),
%     \end{align}
%   and 
%   \begin{align}
%   \E(\overline{Z}^{(d_n)}_{k-m+1, m} \{\overline{Z}^{(d_n)}_{k+1,m}\}^{\T}| \FF_{k-2L}) &= \E\left( \sum_{j=1}^{L-1} \underline{p}_{j, k,m}\underline{p}_{j + m , k+m,m}^{\T} (\zeta^{\circ}_{k- j})^2+ \sum_{j=1}^{m}  z_{k+1-j}^{\circ} \underline{p}_{j + m - 1, k+m,m}^{\T} \zeta^{\circ}_{k+1- j} | \FF_{k-2L}\right). 
%     \end{align}
   After a careful inspection of Step 4 of \cref{lm:Sigma_d}, we have
  \begin{align}
    \sup_{t \in \I}| { \Sigma}_{d_n}^{\circ}(t)- \E { \Sigma}_{d_n}^{\circ}(t)|  = \Op\left(\surd{\frac{1}{m \tau_n^{3/2}}}\right) = \op(1),\label{eq:local_step6}
  \end{align}
  \par
  \textbf{Step 7} Recall that $\check { \Sigma}(t) =  { \Sigma}(t) + (e^{c\alpha_1}-1)^2\sigma_H^2(t) { \mu}_W(t) \mu^{\T}_W(t) + (e^{c\alpha_1}-1)s_{UH}(t) { \mu}_W^{\T}(t) + (e^{c\alpha_1}-1){ \mu}_W(t)s^{\T}_{UH}(t)$.
  We shall show that uniformly for $t \in \I$,
  \begin{align}
  \E{ \Sigma}^{\circ}_{d_n}(t) = \check { \Sigma}(t) + O\{(\log n)^{-1}\}.
  \end{align}
%   where $$\check { \Sigma}(t) =  { \Sigma}(t) + (e^{c\alpha}-1)^2\sigma_H^2(t) { \mu}_W(t) \mu^{\T}_W(t) + (e^{c\alpha}-1)s_{UH}(t) { \mu}_W^{\T}(t) + (e^{c\alpha}-1){ \mu}_W(t)s^{\T}_{UH}(t).$$
   Let $\breve { \Delta}_k^{(d_n),\circ} = \frac{1}{m} \sum_{i=k-m+1}^k \sum_{j=1}^L \psi_j \left\{{ \mu}_W(t_i) \zeta^{\circ}_{i-j} - { \mu}_W(t_{i+m}) \zeta^{\circ}_{i+m-j} \right\}$, ${ \Delta}_k^{\circ} = \frac{1}{m} \sum_{i=k-m+1}^k (z_i^{\circ} - z_{i+m}^{\circ} )$. Define for $t \in [0,1]$, $\tilde { \Sigma}^{\circ} (t) = \sum_{j=m}^{n-m} \frac{m \omega(t,j)}{2} { \Delta}_j^{\circ} { \Delta}_j^{\circ,\T}$,
  \begin{align}
    %   &\tilde { \Sigma}^{\circ} (t) = \sum_{j=m}^{n-m} \frac{m \omega(t,j)}{2} { \Delta}_j^{\circ} { \Delta}_j^{\circ,\T},\\
      &\tilde{s}_1^{\circ} (t) = \sum_{j=m}^{n-m} \frac{m \omega(t,j)}{2} \breve { \Delta}_j^{(d_n),\circ} (\breve { \Delta}_j^{(d_n),\circ})^{\T} ,~\text{and}~ \tilde{s}_2^{\circ}(t) = \sum_{j=m}^{n-m} \frac{m \omega(t,j)}{2} { \Delta}_j^{\circ} (\breve { \Delta}_j^{(d_n), \circ})^{\T} .
  \end{align}
  Then it follows that
  \begin{align}
  \E\{{ \Sigma}_{d_n}^{\circ}(t) \}=  \E\{\tilde { \Sigma}^{\circ}(t)\}  +  \E \{\tilde {s}^{\circ}_{1} (t)\} + \E \{\tilde{s}^{\circ}_{2}(t)\} + \E \{\tilde{s}^{\circ, \T}_{2}(t)\}.
  \end{align}
  Following similar arguments in Step 5 of \cref{lm:Sigma_d}, 
   by the continuity of ${ \mu}_W$,
    we have
  when $j \leq m$, 
  \begin{align}
    \underline{p}_{j,k,m}  = (e^{c\alpha_1} - 1)  { \mu}_W \left(k/n\right) + O(m/n + d_n).\label{eq:pjm_local1}
  \end{align}
  when $j \geq m + 1$,
    since $L / n \to 0$,
  \begin{align}
    \underline{p}_{j,k,m} = { \mu}_W\left(k/n\right)  \{j^{d_n} - (j-m + 1)^{d_n}\} + O(m/n + d_n + m/L).\label{eq:pjm_local2}
   \end{align}
  Since $z_i^{\circ}$ are martingale differences, and $\sigma_H(t_j) = \| \sum_{i=j}^{\infty} \proj_j H(t_j,\FF_i)\| = \| \zeta^{\circ}_j \|$, 
  \begin{align}
    \E(\overline{Z}_{k-m+1,m}^{(d_n)} \overline{Z}_{k-m+1,m}^{(d_n), \T})&= \sum_{j= 1}^{L+m-1} \underline{p}_{j,k,m} \underline{p}^{ \T}_{j,k,m}  \E\{(\zeta^{\circ}_{k-j})^2\} + \sum_{i=k-m+1}^{k} \E(z_i^{\circ} z_i^{\circ,\T}) \\
    & + \sum_{j=1}^{m-1} \underline{p}_{j, k,m} \E(\zeta^{\circ}_{k- j} z_{k-j} ^{\circ,\T}) + \sum_{j=1}^{m-1} \E( z_{k-j}^{\circ}  \zeta^{\circ}_{k- j}) \underline{p}_{j, k,m}^{\T}\\
    & = Z_1 + Z_2 + Z_3 + Z_4.\label{eq:localZ0}
  \end{align}
  By \cref{cor:kar_dn}, we have
  \begin{align}
  Z_1/m  &= \frac{1}{m}\sum_{j= 1}^{m-1} \underline{p}_{j,k,m} \underline{p}^{ \T}_{j,k,m}  \E\{(\zeta^{\circ}_{k-j})^2 \}+ \frac{1}{m}\sum_{j= m}^{L+m-1} \underline{p}_{j,k,m} \underline{p}^{ \T}_{j,k,m}  \E\{(\zeta^{\circ}_{k-j})^2\} \\ 
  & = (e^{c\alpha_1} - 1)^2 { \mu}_W(k/n){ \mu}_W^{\T}(k/n)\sigma_H^2(k/n) + O\{m/n+d_n+(\log m)^{-1}\} .
  \end{align}
  Under Assumption  \ref{Ass-U}, we have 
  \begin{align}
  Z_2/m = { \Sigma}(t)  + O(m/n).
  \end{align}
  Observe that 
  \begin{align}
  \E( z_{j}^{\circ}  \zeta^{\circ}_{j})  = \sum_{k=-\infty}^{\infty} \mathrm{Cov}\{U(t_j, \FF_0), H(t_j, \FF_k)\}  = s_{UH}(t_j).
  \end{align}
  Under Assumption \ref{C:s}, similar arguments in the calculation of $Z_1$ and $Z_2$ imply, 
  \begin{align}
  Z_3/m = (e^{c\alpha_1}-1) { \mu}_W(k/n) s^{\T}_{UH}(k/n) + O(m/n + d_n),
  \end{align}
  and 
  \begin{align}
  Z_4/m = (e^{c\alpha_1}-1) s_{UH}(k/n)  \mu^{\T}_W(k/n) + O(m/n + d_n).
  \end{align}
  By \cref{cor:kar_dn} (a), \eqref{eq:pjm_local1} and \eqref{eq:pjm_local2}, similar techniques of \eqref{eq:localZ0} show that
  \begin{align}
    \E(\overline{Z}_{k+1,m}^{(d_n)} \overline{Z}_{k-m+1,m}^{(d_n), \T}) &= \sum_{j= 1}^{L-1} \underline{p}_{j+m, k+m,m} \underline{p}^{\T}_{j, k,m}  \E\{(\zeta^{\circ}_{k-j})^2\}  + \sum_{j=1}^{m} \underline{p}_{j+m-1, k+m,m} \E(\zeta^{\circ}_{k+1- j} z_{k+1-j} ^{\circ,\T})\\  &= O\{m(\log m)^{-1}\}.
    \label{eq:localZm}
  \end{align}
  Therefore, by  \eqref{eq:localZ0} and \eqref{eq:localZm}, we have
  \begin{align}
  \E { \Sigma}_{d_n}^{\circ}(t) &= \sum_{j=m}^{n-m} \check { \Sigma} (t_j)\omega(t,j)  + O\{m/n + d_n + (\log m)^{-1}\} 
      = \check  { \Sigma} (t) +O\{(\log n)^{-1}\}.
  \end{align}
  \par
  \textbf{Step 8} 
  Summarizing Step 1 - Step 7, we have 
  \begin{align}
  \sup_{t \in \I} |\hat { \Sigma}_{d_n}(t) - \check { \Sigma}(t)| =\Op\left\{\surd{\frac{m}{n\tau_n^{3/2+2\kappa}}} + \surd{m}\tau_n^{3-1/\kappa} +\surd{\frac{1}{m\tau_n^{3/2}}} + (\log n)^{-1}\right\} = \op(1).\label{eq:localdrate}
  \end{align} 
  \end{proof}

\section{Auxiliary results}\label{sec:aux}
   \begin{proposition}\label{physical}
         Suppose $Q_i = L(t_i, \FF_i)$, $t_i \in  \I$,  for $q \geq 1$, we have
         \begin{align}
           \|\proj_{i-l} Q_i \|_q \leq \delta_q(L, l, \I).
         \end{align}
       \end{proposition}
       \begin{proof}[Proof of \cref{physical}]
         The proposition follows after a careful investigation of Theorem 1 in \citep{wu2005nonlinear}.
       \end{proof}
       
\begin{proposition}\label{Prop31}
Under Assumption \ref{assumptionHp},  we have uniformly for $l \geq 0$, $0 < d < 1/2$,
\begin{align}
  \delta_p(H^{(d)}, l, (-\infty, 1]) = O\{(1+l)^{d-1}\}.
\end{align}
\end{proposition}

% \subsection{Proof of \texorpdfstring{\cref{Prop31}}{prop31}}
\begin{proof}[Proof of \cref{Prop31}]
Under Assumption \ref{assumptionHp}, by Lemma 3.2 of \citep{KOKOSZKA199519} and \cref{physical}, we have
\begin{align}
  \delta_p(H^{(d)}, l, (-\infty, 1]) 
%   &= \sup_{t \in (-\infty, 1]} \left\|\sum_{k=0}^{l} \psi_k (H(t - t_k, \F_{l-k}) - H(t - t_k, \F_{l-k}^*) )
%   \right\|_p\\
%   & \leq   \sum_{k=0}^{l} \psi_k \sup_{t \in (-\infty, 1]} \|H(t - t_k, \F_{l-k}) - H(t - t_k, \F_{l-k}^*)\|_p\\
 \leq \sum_{k=0}^{l} \psi_k(d) \delta_p(H,l-k, (-\infty, 1])  
 =O\{(1+l)^{d-1}\}.
  \end{align}
\end{proof}

   \begin{lemma}\label{lm:davis}
         Suppose $\left\| \sup_{t \in [0,1]}\left|\hat { \Sigma}(t)-{ \Sigma}(t)\right|\right\|  =  O(s_n)$, where ${ \Sigma}(t)$ is a covariance matrix with its eigenvalues bounded from zero, $\mathrm{dim}({ \Sigma}(t)) = p < \infty$. Then, we have 
         \begin{align}
           \sup_{t \in [0,1]}\left|\hat { \Sigma}^{1/2}(t)-{ \Sigma}^{1/2}(t)\right| = \Op(s_n^{1/2}).
         \end{align}
         \end{lemma}
         
       \begin{proof}[Proof of \cref{lm:davis}]
         Without loss of generality, suppose ${ \Sigma}(t)$ has eigenvalues $\lambda_1(t) \geq \cdots \geq \lambda_p(t)$, and eigenvector matrix $V(t) = (v_1(t), \ldots, v_p(t))$, ${ \Sigma}(t)v_j(t) =\lambda_j(t) v_j(t)$, ${ \Lambda}(t) = \mathrm{diag}\{\lambda_1(t) \geq \cdots \geq \lambda_p(t)\}$. Suppose $\hat { \Sigma}(t)$ has eigenvalues $\hat \lambda_1(t) \geq \cdots \geq \hat \lambda_p (t)$, and eigenvector matrix $\hat{V}(t) = (\hat {v}_1(t), \ldots, \hat {v}_p(t))$, $\hat { \Sigma}(t) \hat {v}_j(t) = \hat \lambda_j(t) \hat {v}_j(t)$, $\hat { \Lambda}(t) = \mathrm{diag}\{\hat \lambda_1(t) \geq \cdots \geq \hat \lambda_p (t)\}$.
         Suppose ${ \Sigma}(t)$ has $q$ distinct eigenvalues, $\tilde \lambda_1(t) > \cdots > \tilde \lambda_q(t)$. Let $ {Q}(t) = \{k: \exists j \neq i, \lambda_j(t) = \lambda_i(t)  = \tilde \lambda _k(t)\}$. 
           Let
          \begin{align}
           { \Sigma}^{\circ}(t) = \hat{V}(t){ \Lambda}(t)\hat{V}(t)^{\T}.
          \end{align}
          Then, we have 
          \begin{align}
            \E\sup_{t \in [0,1]}\left|\hat { \Sigma}^{1/2}(t)-{ \Sigma}^{1/2}(t)\right| &\leq \E \sup_{t \in [0,1]}\left|\hat { \Sigma}^{1/2}(t)-({ \Sigma}^{\circ})^{1/2}(t)\right| + \E \sup_{t \in [0,1]}\left|({ \Sigma}^{\circ})^{1/2}(t)-{ \Sigma}^{1/2}(t)\right| \\&= S_1 + S_2,\label{eq:S1S2}
          \end{align}
          where $S_1$ and $S_2$ are defined in the obvious way.
           $\hat{V}(t)$ is orthogonal, and $|\cdot|$ is the Frobenius norm, then
          \begin{align}
            S_1 &= \E \sup_{t \in [0,1]}|\hat{V}(t) \{{ \Lambda}^{1/2}(t) - \hat { \Lambda}^{1/2}(t)\}\hat{V}(t)^{\T}\|\\ &  \leq \left\| \sup_{t \in [0,1]}\left|{ \Lambda}^{1/2}(t) - \hat { \Lambda}^{1/2}(t)\right|\right\| \left\| \sup_{t \in [0,1]}\left|\hat{V}(t)\right|\right\| = O(s_n^{1/2}).\label{eq:S1}
          \end{align}
          By Corollary 1 in \citep{yu2015useful}, if $k \not \in  {Q}(t)$, suppose $\lambda_i(t) = \tilde \lambda_k(t)$. Then, we have
         \begin{align}
           |\hat {v}_i(t) - {v}_i(t)| \leq \frac{2^{3/2}\rho\left\{\hat { \Sigma}(t) - { \Sigma}(t)\right\}}{\min\{\lambda_{i-1}(t)-\lambda_{i}(t), \lambda_i(t)-\lambda_{i+1}(t)\}}.\label{eq:eigens}
         \end{align}
         If $j \in  {Q}(t)$, suppose  $\lambda_{r-1}(t) > \lambda_r(t) = \cdots = \tilde \lambda_j(t)  = \cdots = \lambda_s(t) > \lambda_{s+1}(t) $, and let ${V}_j(t) = ({v}_r(t), \ldots, {v}_s(t))$. Let $\hat{V}_j(t) = (\hat {v}_r(t), \ldots , \hat {v}_s(t))$.
           By Theorem 2 in \citep{yu2015useful}, $\exists \hat{O}_j(t) \in \R^{(s-r+1) \times (s-r+1)}$ which is orthogonal, s.t.
         \begin{align}
           \lt|\hat{V}_j(t)\hat{O}_j(t) - {V}_j(t)\rt|  \leq \frac{2^{3 / 2} \min \left\{(s-r+1)^{1 / 2}\rho\left|\hat{{ \Sigma}}(t)-{ \Sigma}(t)\right|, \left|\hat{{ \Sigma}}(t)-{ \Sigma}(t)\right|\right\}}{\min \left(\lambda_{r-1}(t)-\lambda_{r}(t), \lambda_{s}(t)-\lambda_{s+1}(t)\right)}.\label{eq:eigenm}
         \end{align}
         Without loss of generality, suppose $\lambda_1(t) > \cdots > \lambda_{s}(t) > \lambda_{s+1}(t) =  \cdots \lambda_{s+n_{s+1}}(t)> \lambda_{s+n_{s+1} + 1}(t)  = \cdots = \lambda_{s+n_{s+1}+n_{s+2}}(t) > \cdots > \lambda_{s+\sum_{i=s+1}^{q-1} n_{i}+1}(t)  = \cdots = \lambda_p(t)$, where $n_i$ is algebraic multiplicity of $\tilde \lambda_i$, and $\sum_{i=s+1}^{q} n_{i} = p-s$. Let 
         \begin{align}
          \hat{O}(t)  = \begin{pmatrix}
             I_{s}\\ 
               &\hat{O}_{s+1}(t)\\
               & & \ldots \\ 
             & & &\hat{O}_{q}(t) \\
           \end{pmatrix},
         \end{align}
         where $\hat{O}_{q}(t) \in \R^{n_q \times n_q} $.
         From \eqref{eq:eigens} and \eqref{eq:eigenm}, we have
         \begin{align}
          \left|\hat{V}(t) \hat{O}(t)  - V(t)\right| \leq \frac{2^{3 / 2} p^{3 / 2} \left|\hat{{ \Sigma}}(t)-{ \Sigma}(t)\right|}{\min_{1 \leq s \leq q+1} \left\{\tilde \lambda_{s-1}(t)-\tilde\lambda_{s}(t)\right\}},\label{eq:O1}
         \end{align}
         where $\lambda_0(t) = \infty$, $\lambda_{q+1}(t) = -\infty$.
         On the other hand, 
         \begin{align}
           \hat{V}(t) \hat{O}(t) { \Lambda}^{1/2}(t) \hat{O}^{\T}(t) \hat{V}^{\T} (t)  = \hat{V}(t)  { \Lambda}^{1/2}(t) \hat{V}^{\T} (t).\label{eq:O2}
         \end{align} 
         Therefore, by \eqref{eq:O1} and \eqref{eq:O2}, we have
         \begin{align}
           S_2  &\leq \E\left(\sup_{[0,1]}\left|\hat{V}(t)  { \Lambda}^{1/2}(t) \hat{V}^{\T} (t) - V(t) { \Lambda}^{1/2}(t) {V}^{\T} (t)\right| \right)\\ 
        %   V(t)  { \Lambda}^{1/2}(t) {V}^{\T} (t)|\\ 
           & \leq \E \left(\sup_{[0,1]}\left|\{\hat{V}(t) \hat{O}(t) - V(t)\} { \Lambda}^{1/2}(t) \hat{O}^{\T}(t) \hat{V}^{\T} (t)\right| \right) + \E \left(\sup_{[0,1]}\left|V(t){ \Lambda}^{1/2}(t) \{\hat{O}^{\T}(t) \hat{V}^{\T} (t) - {V}^{\T} (t)\}\right|\right)\\ 
           & \leq C  \left\| \sup_{[0,1]}\left| \hat{{ \Sigma}}(t) - { \Sigma}(t)\right| \right\| = O(s_n),\label{eq:S2}
         \end{align}
         where $C$ is a sufficiently large positive constant.
         Combining \eqref{eq:S1S2}, \eqref{eq:S1} and \eqref{eq:S2}, we have 
         \begin{align}
           \sup_{t \in [0,1]}\left|\hat { \Sigma}^{1/2}(t)-{ \Sigma}^{1/2}(t)\right| = \Op(s_n^{1/2}).
         \end{align}
         \end{proof}

% \subsection{Proofs of \cref{cor:kpss_trend}, \cref{lm:alt_approx_trend} and \cref{lm:alt_local_trend}}

% \cref{cor:kpss_trend}, \cref{lm:alt_approx_trend} and \cref{lm:alt_local_trend} are direct corollaries of \cref{thm:null_dist}, \cref{thm:alt_approx} and \cref{thm:local}. We omit their proofs for brevity.

  \begin{lemma}\label{lm:basic}
         The following argument shows the properties of long memory coefficient $\psi_j = \psi_j(d)$.\\
        $\psi_0 = 1$, and for $j \geq 1$, 
       \begin{align}
         \psi_j = j^{d-1}l_d(j),
       \end{align}
       where $l_d(j) = 1/\Gamma(d)\{1 + O(1/j)\}$.
       \end{lemma}
       \begin{proof}[Proof of \cref{lm:basic}]
       
        By Stirling's formula,
        \begin{align}
         \frac{\Gamma(j+d)}{\Gamma(j+1)}= \frac{\surd{\frac{2\pi}{j+d}}(\frac{j+d}{e})^{j+d}\{1+O(\frac{1}{j+d})\}}{\surd{\frac{2\pi}{j+1}}(\frac{j+1}{e})^{j+1}\{1+O(\frac{1}{j+1})\}} &= \frac{j^{d-1}(1 + d/j)^j(1 + d/j)^{d-1/2}\{1+O(1/j)\}}{e^{d-1}(1 + 1/j)^j(1 + 1/j)^{1/2}\{1+O(1/j)\}}\\
         & = j^{d-1} + O(j^{d-2}).
        \end{align}
          Since $\ln \Gamma(z) \sim z \ln z-z+\frac{1}{2} \ln \frac{2 \pi}{z}+\sum_{n=1}^{N-1} \frac{B_{2 n}}{2 n(2 n-1) z^{2 n-1}}$, the constant in the big $O$ of $\Gamma(z)=\surd{\frac{2 \pi}{z}}\left(\frac{z}{e}\right)^{z}\left\{1+O\left(\frac{1}{z}\right)\right\}$ is always $B_2/2 = 1/12$.
       \end{proof}
       
  \begin{lemma}\label{lm:delta_xed}
  Assuming  that 
  $ \sup_{ t \in (-\infty,1] }\left\|H\left(t, \mathcal{F}_{0}\right)\right\|_{2p}<\infty
  $, $ \delta_{2p}(H, k, (-\infty,1]) = O(\chi^k)$, $\delta_{2p}(W, k) = O(\chi^k), \chi \in (0,1)$, $ \sup_{ t \in [0,1] }\left\|W\left(t, \mathcal{F}_{0}\right)\right\|_{2p}<\infty
  $,  we have 
  \begin{align}
    \delta_p(U^{(d)}, k) = O(k^{d-1}).
  \end{align}
\end{lemma}
\begin{proof}[Proof of \cref{lm:delta_xed}]
  Note that for $j \leq i$,
  \begin{align}
    \delta_p(U^{(d)}, i-j) 
    &\leq \|W(t_i,\FF_i)\|_{2p} \delta_{2p}(H^{(d)},i-j)+ \|H^{(d)}(t_i,\FF_{i-j}^*)\|_{2p} \delta_{2p}(W,i-j).\label{eq:Pu}
  \end{align}
  By Burkholder's inequality and \cref{Prop31}, we have 
  \begin{align}
    \|H^{(d)}(t_i,\FF_i)\|^2_{2p}
    &\leq  M\left\|  \sum_{j\in  {Z}} \left\{\mathcal{P}_j H^{(d)} (t_i, \mathcal{F}_i) \right\}^2 \right\|_{p} \leq  M \sum_{j\in  {Z}} \left\| \mathcal{P}_j H^{(d)} (t_i, \mathcal{F}_i) \right\|_{2p}^2  = O(1),
  \end{align}
  where $M$ is a sufficiently large constant.
  
%   Along with conditions  $ \sup_{ t \in [0,1] }\left\|W\left(t, \mathcal{F}_{0}\right)\right\|_{2p}<\infty
%   $, $ \delta_{2p}(H, k, (-\infty, 1]) = O(\chi^k), \delta_{2p}(W, k) = O(\chi^k), \chi \in (0,1)$,
  Then by \cref{Prop31} and \cref{eq:Pu}, we have proved the desired result.
\end{proof}

  \begin{lemma}\label{cor:kar_dn}
        Let  $\psi_j$ denote $\psi_j(d_n)$.
        (a) For $h=o(n)$, $h \to \infty$ we have
            \begin{align}
            \sum_{l=0}^{\infty} \left(\sum_{j=l}^{h+l} \psi_j\right)^2 \sim  \sum_{l=0}^{\infty} \{(h+l)^{d_n} - l^{d_n}\}^2  = O\{h/(\log h)\}.
           \end{align}
        (b) If we further assume $h = \lf n^{\alpha_1}\rf$,  $\alpha_1 \in (0,1)$,  we have  
           \begin{align}
            h^{-1} \sum_{l=0}^{h-1} \left(\sum_{j=0}^{l} \psi_j\right)^2 \to e^{2 c \alpha_1}.
           \end{align}
       \end{lemma}

       \begin{remark}
         These two results correspond to the conclusions in Lemma 2 of \citep{shao2007local}. 
       \end{remark}

       \begin{proof}[Proof of \cref{cor:kar_dn}]
        Proof of (a). We first show 
        \begin{align}
          \sum_{l=0}^{\infty} \{(h+l)^{d_n} - l^{d_n}\}^2 = O\{h/(\log h)\}.
        \end{align}
        Let $N_1 = \lf h^{1-\alpha_h}\rf$, $N_2 = \lf h^{1+\alpha_h}\rf$, $
        \alpha_h = (\log h)^{-1}\log \log h$. Then, $h/N_1 = O(h^{\alpha_h}) = O(\log h)$. $N_1^{d_n} = O(1)$, $N_2^{d_n}= O(1)$. By \cref{lm:basic} and Taylor's expansion, we have
        \begin{align}
          \sum_{l=0}^{\infty}\{(h+l)^{d_n} - l^{d_n}\}^2 &=  \sum_{l=0}^{N_1}\{(h+l)^{d_n} - l^{d_n}\}^2 + \sum_{l=N_1 + 1}^{N_2}\{(h+l)^{d_n} - l^{d_n}\}^2 + \sum_{l=N_2+1}^{\infty}\{(h+l)^{d_n} - l^{d_n}\}^2\\  &= O\{h(\log h)^{-1}\}.
        \end{align}
        Then by \cref{lm:basic}, we have 
        \begin{align}
        \sum_{l=0}^{\infty} \left(\sum_{j=l}^{h+l} \psi_j\right)^2 
        = \sum_{l=0}^{\infty}\{(h+l)^{d_n} - l^{d_n}\}^2+ O\{d_n h(\log h)^{-1}\}.
        \end{align}
        
        Proof of (b). 
        By \cref{lm:basic}, we have 
        \begin{align}
          \sum_{l=0}^{h-1} \left(\sum_{j=0}^{l} \psi_j\right)^2 
          & = 1 + \sum_{l = 1}^{h-1} l^{2d_n} + O\left(d_n \sum_{l = 1}^{h-1} l^{2d_n}\right) ,\label{eq:b1}
        \end{align}
        where for $h = O(n)$,
        \begin{align}
           \sum_{l = 1}^h l^{2d_n} =   e^{2c\alpha_1} \sum_{l = 1}^h (l/h)^{2d_n} &=  e^{2c\alpha_1} h \int_0^1 t^{2d_n} dt + O(1)\\ &= e^{2c\alpha_1} h/(2d_n + 1) + O(1) = e^{2c\alpha_1} h + O(1).\label{eq:b2}
        \end{align}
        Combining \eqref{eq:b1} and \eqref{eq:b2}, we have shown the desired result.
        % as $h \to \infty$, 
        % \begin{align}
        %       h^{-1}\sum_{l=0}^{h-1} \left(\sum_{j=0}^{l} \psi_j\right)^2 \to  e^{2c\alpha}.
        % \end{align}
      
      \end{proof}

 \normalem
\bibliographystyle{apalike}
\bibliography{main}
\end{document}